\documentclass[letterpaper]{article}
\usepackage[utf8]{inputenc}

\title{\vspace{-1cm}Short Synchronizing Words for Random Automata}

\author{Guillaume Chapuy%
\thanks{Université Paris Cité, CNRS, IRIF, F-75013, Paris, France
Email:~{\tt guillaume.chapuy@irif.fr}.
}
\and Guillem Perarnau%
\thanks{Departament de Matem\`atiques and IMTECH, Universitat Polit\`ecnica de Catalunya (UPC), Barcelona, Spain. Centre de Recerca Matemàtica, Barcelona, Spain. Email:~{\tt guillem.perarnau@upc.edu}. 
}
}

\date{\today}

\usepackage{graphicx, enumitem}
\usepackage{amsmath,amsfonts,amssymb}
\usepackage{amsthm, xcolor}
\usepackage{comment, fullpage, placeins,hyperref}

\usepackage[margin=7mm]{caption}

\theoremstyle{theorem}
\newtheorem{theorem}{Theorem}
\newtheorem{proposition}[theorem]{Proposition}
\newtheorem{corollary}[theorem]{Corollary}
\newtheorem{lemma}[theorem]{Lemma}
\newtheorem{claim}[theorem]{Claim}

\numberwithin{theorem}{section}

\theoremstyle{definition}
\newtheorem{definition}[theorem]{Definition}
\newtheorem{remark}[theorem]{Remark}
\newtheorem{notation}[theorem]{Notation}

\newcommand{\cyc}[2]{\mathrm{cyc}_{#1}(#2)}
\newcommand{\Ccyc}[2]{\mathrm{Cyc}_{#1}(#2)}

\renewcommand{\Pr}{\mathbb{P}}

\newcommand{\ba}{{\mathbf{a}}}
\newcommand{\bb}{{\mathbf{b}}}
\newcommand{\be}{{\mathbf{e}}}
\newcommand{\hE}{{\hat{E}}}
\newcommand{\hF}{{\hat{F}}}
\newcommand{\tO}{{\widetilde{O}}}

\newcommand{\nota}[1]{{\it \textcolor{blue}{#1}}}
\newcommand{\AAA}{\mathcal{A}}
\newcommand{\AAb}{\mathcal{A}^\bullet}
\newcommand{\AAbb}{\mathcal{A}^{\bullet\bullet}}
\newcommand{\AAl}{\mathcal{A}^\diamond}
\newcommand{\AAll}{\mathcal{A}^{\diamond\diamond}}
\newcommand{\AAbl}{\mathcal{A}^{\bullet\diamond}}

\newcommand{\fourmarks}{\substack{\vspace{-.19em}\bullet\bullet \\ \vspace{-.1em} \diamond\diamond}}

\newcommand{\AAbbll}{\mathcal{A}^{\fourmarks}}

\newcommand{\thread}[3]{(#1,#2)\stackrel{#3}{\leadsto}*}

\newcommand{\Athread}[4]{(#1,#2)\stackrel{#3,#4}{\leadsto}*}

\newcommand{\arrows}[5]{(#1,#2)\stackrel{#3}{\leadsto}(#4,#5)}
\newcommand{\Aarrows}[6]{(#1,#2)\stackrel{#3,#4}{\leadsto}(#5,#6)}

\newcommand{\BB}[1]{\mathcal{B}^{\bullet\diamond}_{#1}}
\newcommand{\BBB}{\mathcal{B}^{\fourmarks}}

\newcommand{\FFF}{\mathcal{F}^{\fourmarks}}
\newcommand{\PPP}{\mathcal{P}^{\fourmarks}}
\newcommand{\QQQ}{\mathcal{Q}^{\fourmarks}}
\newcommand{\SSS}{\mathcal{S}^{\diamond\diamond}}

\newcommand{\CC}[1]{\mathcal{C}^{\diamond}_{#1}}

\newcommand{\cC}{\mathcal{C}}
\newcommand{\cE}{\mathcal{E}}
\newcommand{\cI}{\mathcal{I}}

\newcommand{\Geom}{\text{Geom}}

\begin{document}

\maketitle

\begin{abstract} 
	We prove that a uniformly random automaton with $n$ states on a 2-letter alphabet has a synchronizing word of length $O(n^{1/2}\log n)$ with high probability (w.h.p.).
	That is to say, w.h.p. there exists a word $\omega$ of such length, and a state $v_0$, such that $\omega$ sends all states to $v_0$. 	Prior to this work, the best upper bound was the quasilinear bound $O(n\log^3n)$ due to Nicaud~\cite{Nicaud}.
	The correct scaling exponent had been subject to various estimates by other authors between $0.5$ and $0.56$  based on numerical simulations, and our result confirms that the smallest one indeed gives a valid upper bound (with a log factor).

	Our proof introduces the concept of $w$-trees, for a word $w$, that is,
	automata in which the $w$-transitions induce a (loop-rooted) tree.
	We prove a strong structure result that says that, w.h.p., a random automaton on $n$ states is a $w$-tree for some word $w$ of length at most $(1+\epsilon)\log_2(n)$, for any $\epsilon>0$.
	The existence of the (random) word $w$ is proved by the probabilistic method. This structure result is key to proving that a short synchronizing word exists.

\end{abstract}

\section{Introduction and main results}

The notion of automaton in computer science is the most fundamental abstraction for a machine that produces an output from an input string in one pass, working only with a finite amount of memory. In this paper, we will work with deterministic automata with $n$ states.
The transitions between these states are indexed by letters of an underlying 2-letter alphabet $\{a,b\}$, but our results apply as well in the case of any finite alphabet of size at least $2$.
Formally, an \nota{automaton}\footnote{Everywhere in the paper, when a new notion or notation is introduced, we use an italic blue font.} is represented by two functions $a:[n] \rightarrow [n]$ and $b: [n] \rightarrow [n]$ that describe the (deterministic) transitions between states.
We do not need to specify any initial nor final states as these notions are irrelevant for the properties we will study.
Figure~\ref{fig:automaton}-Left represents an automaton with $n=4$ states.

\begin{figure}[h!!!]
	\begin{center}
		\includegraphics[width=0.7\linewidth]{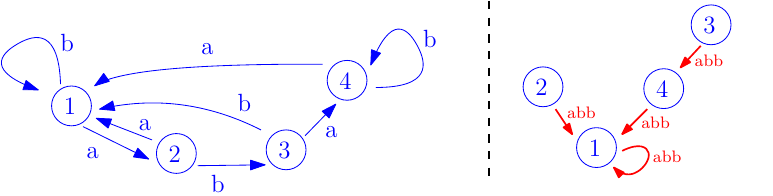}
		\caption{Left: An automaton $A$ with $n=4$ states.
		$A$ is synchronizable, since the word $\omega=a bbabb$ is synchronizing: it sends all the states to the state $v_0=1$. Right: the one-letter automaton $A_w$ induced by $w$-transitions, for $w=abb$. It has a unique cyclic vertex, i.e. it is loop-rooted tree (equivalently, $A$ is $w$-tree).}
		\label{fig:automaton}
	\end{center}
\end{figure}

Although the theory of automata and regular languages has been considerably developed since almost a century (see e.g.~\cite{Pin:handbook}), the field is still full of fascinating open problems (see e.g.~\cite{Pin:open, Volkov:survey}). One of them, the \emph{Černý conjecture}, related to the notion of synchronizing words, is the main inspiration of this paper. 
Given an automaton,
 $\omega \in \{a,b\}^*$ is a \nota{synchronizing word} (or \nota{reset word}) if there exists a state $v_0 \in [n]$ such that $\omega$ sends $u$ to $v_0$ for \emph{all} states $u\in [n]$ (i.e. if $\omega=\omega_1\dots \omega_t$, then $\omega(u)=\omega_t(\omega_{t-1}(\dots \omega_1(u)\dots))=v$). An automaton is \nota{synchronizable} if it admits a synchronizing word. See Figure~\ref{fig:automaton}-Left again for an example.

The notion of synchronizing word was introduced by Černý in a famous paper~\cite{Cerny:original}. Of course, not all automata are synchronizable (for example if both letters induce permutations, or if the automaton is not connected) but any synchronizable automaton admits a word of length at most $n^3$ that is synchronizing\footnote{It is easy to see that any pair of vertices that can be synchronized, can be synchronized \ by a word of length at most $n^2$, using pigeonhole on  pairs of visited vertices along any word. One can use this to synchronize vertices iteratively one at a time, hence the bound $n^3$.  This argument (due to~\cite{Cerny:original}) also shows that checking synchronizability is doable in polynomial time.}.

Černý conjecture~\cite{Cerny:original} asserts that the bound $n^3$ can be replaced by $(n-1)^2$ (this value is reached by a beautiful construction of Černý), and is among the most famous open problems in automata theory, see the survey~\cite{Volkov:survey}.
After almost 50 years, the cubic barrier has still not been broken, and the best known bounds are of the form $cn^3$, where the famous value $c=\frac{1}{6}$ due to Pin and Frankl~\cite{Pin:Cerny, Frankl:Cerny} has been improved only very recently by~Szyku\l{}a~\cite{Szykula} and Shitov~\cite{Shitov} (the best known value is $c=0.165\dots$).
Today, the notion of synchronization extends well beyond automata theory, see e.g.~\cite{CameronSteinberg,DoyenMassartMahsa}.

\medskip

The main focus of this paper is not Černy's conjecture in itself, but the synchronization of large \nota{random automata} (i.e. automata taken uniformly at random among the $n^{2n}$ automata with $n$ states, when $n$ goes to infinity). This question was considered by several authors at least since the 2010's~\cite{SZ:sync, ZS:sync}. Cameron~\cite{Cameron} conjectured that a random automaton is synchronizable with probability tending to $1$ when $n$ goes to infinity (i.e., \nota{with high probability}, or \nota{w.h.p.}). This was proved by Berlinkov~\cite{Berlinkov:proof} and quickly after by Nicaud~\cite{Nicaud}, who obtained a quantitative upper bound of $O(n\log^3 n)$ for the shortest synchronizing word, w.h.p..
Interestingly, the result not only shows that Černý conjecture holds for almost all automata, but also that typical automata are very far from the extremal value. However, the quasilinear bound appears to be still quite far from the truth.
Indeed, several authors have tried to estimate the correct order of magnitude for the length of the shortest reset word, using nontrivial simulations (finding the shortest synchronizing word is NP-hard~\cite{OU:complexity}, even approximately~\cite{GS:inapprox}).
Various estimates of the form $O(n^\alpha)$ have been proposed for the expectation, in particular the paper \cite{KKS-sqrtn} proposes $\alpha=1/2$ while other experiments~\cite{ST11:experimental,SzykulaZyzik} suggest a slightly larger value.
In this paper we will be interested instead in the \emph{typical} value (rather than expectation) of this parameter, and we prove that $n^{1/2}$ is indeed a correct upper bound up to a log factor. Indeed, we prove,

\begin{theorem}[Main result]\label{thm:mainIntro}
	A uniformly random automaton with $n$ states on a 2-letter alphabet has a synchronizing word of length at most $C \sqrt{n}\log(n)$ w.h.p., where $C$ is an absolute constant.
\end{theorem}
We insist that prior to our work, the best known result was the quasilinear bound $O(n\log^3 n)$ of Nicaud~\cite{Nicaud}. 
In a different, inhomogeneous, setting,
Gerencs\'er and V\'arkonyi~\cite{GV:arxivPermutations} recently gave synchronizing words of length $O(\sqrt{n}\log^{3/2}n)$ for the $3$-letter model constructed by superposing a $1$-letter random automaton and two random permutations. Their proof relies on the expansion properties of permutation-based graphs~\cite{FJRST:action}, which do not hold for automata. For other recent papers related to the topic see~\cite{refCitingNicaud1,BerlinkovNicaud, refCitingNicaud2}.

\medskip

Beyond automata theory, a strong motivation for this paper comes from the field of random graphs, since after all, random automata are random 2-out digraphs (or $r$-out digraphs for an $r$-letter alphabet), whose study from the probabilistic viewpoint is very natural. In fact, random automata and random $r$-out digraphs are now an active field inside random graph theory, see e.g. the survey~\cite{Nicaud:survey} or the recent papers~\cite{CaiDevroye,ABBP:rout-digraphs,QuattropaniSau}.

In fact, the main ingredient of our proof is a strong structure result for random automata. Given an automaton $A$ and a word $w\in \{a,b\}^*$, we let $A_w$ be the function $[n]\rightarrow [n]$ induced by $w$-transitions (it can be viewed as a one-letter automaton). We say that a function $f:[n]\rightarrow [n]$ is a \nota{loop-rooted tree} if it has a unique cyclic point, i.e. if there is a unique $x\in [n]$ such that $x \in \{f^k(x), k\geq 1\}$. Combinatorially, a loop rooted tree can be seen as a tree directed towards a distinguished vertex, to which a single loop-edge is attached -- hence the name. See Figure~\ref{fig:automaton}-Right.
If $A_w$ is a loop-rooted tree, we say that $A$ is a \nota{$w$-tree}.
Our main technical result says that, maybe surprisingly, almost all automata are $w$-trees for a very short word $w$ -- namely we only need the (random) word $w$ to have a length slightly more than $\log_2(n)$:
\begin{theorem}[Main structure result on random automata]\label{thm:structure}
	Let $A$ be a uniformly random automaton on $n$ states, and let $\epsilon>0$.
	Then, w.h.p.,
	there exists a word $w$ of length at most $(1+\epsilon) \log_2(n)$ such that $A$ is a $w$-tree.
\end{theorem}
As we will see, this structure result is key to our finding of short synchronizing words. Indeed, once we have found a word $w$ of length $(1+\epsilon)\log_2(n)$ such that $A_w$ is a loop-rooted tree (and the theorem says that we can w.h.p., as long as $\epsilon>0$), it is clear that the word $w^H$ is synchronizing, where $H$ is the height of that tree. It will be easy to show that $H\leq C\sqrt{n}$ w.h.p..
In particular, the word $\omega=w^{\lceil C\sqrt{n} \rceil}$ is synchronizing w.h.p., and it has length $O(\sqrt{n}\log n)$. Note moreover that $\omega$ can be encoded with very little data -- it suffices  to know the word $w$, which requires $(1+\epsilon)\log_2(n)$ bits.

To be complete, we believe that the constant $1$ in Theorem~\ref{thm:structure} is sharp (i.e. that no word $w$ of length less that $(1-\epsilon) \log_2(n)$ is such that $A$ is a $w$-tree, w.h.p.), but we do not prove this. Indeed, our proofs focus on a special class of $w$-trees (the ``0-shifted ones'', see below). For this class, the constant $1$ is indeed sharp. The study of general $w$-trees is of independent interest and is left as future work (however we believe that this will not lead to any improvement on our main results).

\bigskip

We conclude the introduction with some perspectives of our work, especially after the introduction of $w$-trees. The first one would be to perform a complete study of $w$-trees from the combinatorial and probabilistic viewpoints. The question of their exact enumeration (for example, for small words $w$) is puzzling, and asymptotic enumeration still poses interesting challenges.
Also, our proof heuristics suggest that properly normalized random $w$-trees could converge to the Continuum Random Tree (CRT, see e.g.~\cite{Aldous:CRT3}), but proving this may give rise to new difficulties. In a different direction, the concept of $w$-tree could be useful in general in the context of synchronisation (for automata, or with adaptations, for other models), and we believe it is worth studying whether it can play a role in addressing structural results
or algorithmic questions on synchronizing words, in the deterministic setting.

Finally, the study of the height of random $w$-trees could help reduce the logarithmic 
factor in
Theorem~\ref{thm:mainIntro}. 
\begin{theorem}\label{thm:height_implies_syncro}
Assume that there is a random variable $X_n$ such that the following holds. For any non self-conjugated word $w$  of length $(1+\epsilon) \log_2(n)$, the
	height $H_n$ (maximum distance of a vertex to the root) of a uniform random 
$0$-shifted marked\footnote{See Section~\ref{sec:w-trees} for the definition of a $0$-shifted marked tree.}
 $w$-tree is stochastically
dominated by $X_n$.
Then the shortest synchronizing word of a random automaton is
stochastically dominated by $(1+\epsilon)X_n \log(n)$.
\end{theorem}
We suspect 
that there exists such $X_n$ satisfying $X_n=O_p(\sqrt{n\log(n)})$, where $O_p$ denotes a big-O in probability.
Proving that statement would replace the w.h.p upper bound $C\sqrt{n}\log n$ in Theorem~\ref{thm:mainIntro} by a bound of the form $O_p(\sqrt{n\log n})$. 
The bijection and estimates presented in this paper give a way to approach this problem. 
Finally, although we do believe that $n^{1/2}$ is the correct scaling for the length of the (typical or expected) shortest synchronizing word up to logarithmic factors, we do not dare to make a precise conjecture on the form of the best possible logarithmic factor (we also note that the presence of logarithms might be in part responsible for the difficulty in estimating exponents numerically in the previous works cited above).
\medskip

It is natural to conjecture that the bound $n^{\frac{1}{2}}$ is sharp, at least up to sub-polynomial terms. We leave this question open, as we can only provide the following lower bound. The proof is very simple, but provides a heuristic for the validity of the conjecture, see Section~\ref{sec:lower}.
\begin{theorem}[Lower bound]\label{thm:lower}
	For any $\epsilon>0$, a uniformly random automaton with $n$ states has no synchronizing word shorter than $n^{\frac{1}{3}-\epsilon}$, w.h.p..	
\end{theorem}
\noindent {\bf Note on the short version of this work:} An extended abstract of this paper was published in the proceedings of the ACM-SIAM Symposium on Discrete Algorithms (SODA'23). The present version uses a slightly modified definition of cycle-good events, which decreases the number of cases in the proofs.

\section{Heuristics, main tools, and structure of the paper}

\subsection{Heuristics: one-letter automata}

Our approach to the synchronisation problem in random automata\footnote{Unless otherwise noted, our automata use two letters. We sometimes talk about one-letter automata,  which are just functions $[n]\rightarrow [n]$ when we find this terminology natural in the context.} is inspired by the situation for random \nota{one-letter} automata, i.e. for random functions $f:[n]\rightarrow [n]$. This is a very well understood problem, that we quickly recapitulate (see e.g.~\cite{Flajolet1989RandomMS}).

A function $f:[n]\rightarrow [n]$ divides the set $[n]$ into a set of \nota{cyclic points} $S\subset [n]$, restricted to which the function $f$ is a permutation (and thus forms a set of directed cycles), and the set $[n]\setminus S$ of remaining points, on which $f$ forms a forest of directed trees attached to these cycles, see Figure~\ref{fig:function}. The number of cyclic points $|S|$ is at least one, and if $|S|=1$ we say that the function is a \nota{loop-rooted tree}.
\begin{figure}
	\begin{center}
		\includegraphics[width=0.9\linewidth]{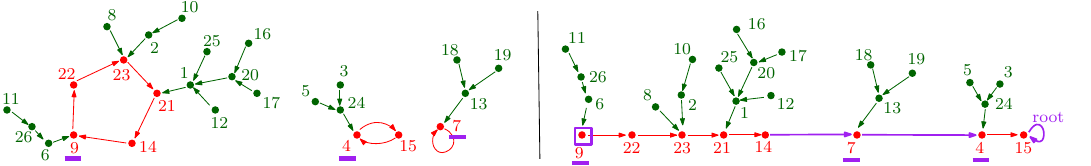}
		\caption{Left: A function $f:[n]\rightarrow [n]$. The set of cyclic points is in red. Right: the function $f$ transformed into a doubly-marked tree via the Joyal bijection. Cycles of $f$ have been cut prior to their minima (underlined), and concatenated by decreasing minima. 
		One obtains a tree with a root (here $15$) and a marked vertex (here $9$).
Lower records on the branch are underlined, they correspond to the cycle minima of the function $f$ (making the construction reversible).
		}
		\label{fig:function}
	\end{center}
\end{figure}

Note that loop-rooted trees are (up to the deletion of the loop edge at the root) precisely the same as rooted labelled trees (Cayley trees) on $n$ vertices, the number of which is famously equal to $n^{n-1}$.
Therefore, if the function $f$ is chosen uniformly at random, the probability that it is a loop-rooted tree is equal to
$$\frac{\mbox{nb. of trees}}{\mbox{nb. of functions}} = \frac{n^{n-1}}{n^n}= \frac{1}{n}.$$
Note that, viewed as a one-letter automaton, a function is synchronizable if and only if it is a loop-rooted tree. Indeed, to be synchronizable it is necessary to have a unique connected component, and clearly the cycle in this connected component needs to contain a unique element. Therefore, the probability that a random one-letter automaton is synchronizable goes to zero when $n$ goes to infinity. However, conditionnally to the fact that it is synchronizable, it can be synchronized by the word $a^H$ where $H$ is the height of the tree. Now, it is well known that the height of a random tree is of order $O(\sqrt{n})$ w.h.p. (e.g.~\cite{RS:height,ABD:randomTrees}). 
This observation is the starting point of our proof, and in fact this is truly where the ``$\sqrt{n}$'' factor in Theorem~\ref{thm:mainIntro} comes from.

\medskip

The main idea of the proof can be described informally as follows. We present it as a series of heuristics,
the rest of this paper will give a precise meaning to these statements (and prove them).

\begin{enumerate}[leftmargin=9pt, topsep=0pt, parsep=0pt, itemsep=1pt]
	\item[(1)] one can hope that, for a relatively short (logarithmic) word $w$, and a random automaton $A$, the one-letter automaton $A_w$ shares some qualitative properties with a random function $f$. For example, one can expect that the probability that it is a tree is of order roughly $\frac{1}{n}$.

	\item[(2)] if we take two words $w_1 \neq w_2$ that are generic enough, we can hope that the facts that $A_{w_1}$ and $A_{w_2}$ are trees, are somehow independent. This intuition relies on the fact that, from a given vertex~$v$, the $w_1$-transition and the $w_2$-transition outgoing from $v$, although they may share some underlying transitions of $A$, are unlikely to coincide and thus close to being independent.

	\item[(3)] consider the set of all words of length $(1+\epsilon)\log_2(n)$, for $\epsilon>0$. The number of such words is $n^{1+\epsilon} \gg n$. From (1), we can hope that, in average, roughly $n^{\epsilon} \gg 1$ of these words $w$ are such that $A$ is a $w$-tree. Moreover, from (2) we can hope that this average value is in fact a \emph{typical} value. That would imply in particular that the number of such words goes to infinity w.h.p., and in particular, that it is nonzero.
\end{enumerate}

There is a number of difficulties to overcome when making this heuristic correct. The first one is that although the definition of the automaton $A_w$ from $A$ is simple, its underlying structure is quite complicated. Indeed, two different edges in $A_w$ (two different $w$-transitions in $A$) may share a lot of underlying edges of $A$. For this reason, the $w$-transitions in a random automaton $A$ are not independent from each other, and this problem will be especially hard to circunvent when one considers \emph{global} properties of $A_w$, that depend on many edges (such as the fact of being a tree).

\subsection{Tools: Bijections and explorations}

In fact, our approach to the problem will combine bijective combinatorics and probability. 
\medskip

The first (and technically easier) part is inspired by the so-called \emph{Joyal bijection} from bijective combinatorics, see e.g.~\cite{Quebecois} or Figure~\ref{fig:function}-Right. This construction transforms a function $[n]\rightarrow [n]$ into a tree with a marked vertex, by cutting all the cycles of $f$ before their minima, and concatenating these cycles together along a branch, by decreasing minima.
In this construction, the minima along the cycles of the original function become the lower records along the branch of the constructed tree, which makes it invertible.

Joyal's bijection famously gives a ``proof from the book'' (see~\cite{proofFromTheBook}) of Cayley's $n^{n-1}$ formula for rooted trees. It also shows the remarkable fact that it is possible to couple a uniform random function and a uniform tree on $[n]$ in such a way that they differ on only $O(\log(n))$ edges w.h.p. (since the number of rewired edges is equal to the number of cycles, which is logarithmic w.h.p., see e.g.~\cite{Flajolet:book}).

In this paper, we will propose a variant $\phi_w$ of Joyal's construction that transforms a (2-letter) automaton into a  $w$-tree by cutting and rewiring some edges participating to the cycles of $A_w$ (the bijection will in fact be denoted $\phi_i$ below). Unfortunately, because rewiring an edge in $A$ may rewire \emph{many} edges in $A_w$, this construction may fail, and in fact, the construction $\phi_w$ is a bijection only between some subset of ``nice automata'' (called \emph{cycle-good}) and some subset of ``nice $w$-trees'' (called \emph{branch-good} and $0$-shifted).
The cycle-good and branch-good trees are the ones which avoid certain ``collision'' events under which the bijection fails, related to the fact that lower records or cycle minima could be unexpectedly connected together by portions of $w$-transitions that could altogether create  parasit cycles in the construction. We will have to prove that these collisions are in fact unlikely to happen.

\medskip

This is where the probabilistic (and longer) part of the paper comes into play. In order to show that collision events do not appear, we will introduce an \emph{exploration process}, which reveals the random automaton $A$ progressively starting from a subset of chosen vertices, by revealing all edges necessary to expose their $w$-transition iterates until cycles are created.
We will typically need to explore paths of length $O(\sqrt{n}\log(n))$, along which, by the birthday paradox (see e.g.~\cite{Flajolet:book}), many repeated vertices might occur. The paths relating these repeated vertices might, altogether, induce new structures (such as, for example, cycles in the graph~$A_w$), and taking this phenomenon into account is the most difficult part of the proof. In the end, the proof of our most difficult estimates will involve a number of case disjunctions (each of which will require to design appropriate probabilistic arguments to exclude certain events being considered), that are related to the intrinsic patterns that these ``induced structures'' could produce.

\medskip

Finally, we say a word about how the independence in the heuristic above is taken into account in the rigorous proof. For this, we use the classical second moment method. This requires to be able to estimate the number of automata $A$ which are \emph{at the same time} a $w_1$-tree and a $w_2$-tree. While that might seem impossible at first, in fact we can approach such structures using the composition $\phi_{w_2}\circ \phi_{w_1}$ of the two variants $\phi_{w_1}$ and $\phi_{w_2}$ of the Joyal bijection (denoted $\phi_1$ and $\phi_2$ below). Because this bijection rewires only $O(\log(n))$ edges w.h.p., one can expect that for most automata the structure of $w_1$-tree created after applying $\phi_{w_1}$ is not destroyed by the rewirings introduced by the bijection $\phi_{w_2}$, thus creating an automaton which is a $w_i$-tree for both values of $i$. This will turn out to be true, at the price of controlling certain complicated 2-word collision events, whose intrinsic complexity (in particular, the intrinsic complexity of the induced structures aforementioned) is the main responsible for the length of this paper.

\subsection{Structure of the paper}
In Section~\ref{sec:notation} we introduce terminology and notation for automata. Because we will require to use the Joyal bijection relative to two different words $w_1$ and $w_2$, it will be convenient to assume that our automata carry two independent permutations on their vertex set (that will be used to define local minima independently). This setup requires to introduce carefully the notation, but this is the price to pay to formally define the bijective part. In this section we also introduce the notion of \emph{thread} which is crucial for our exploration techniques, and we introduce the type of ``collision'' events that we will consider.
In Section~\ref{sec:backbone}, we present a list of lemmas and propositions that altogether form the backbone of our proof.
In Section~\ref{sec:bijProofs}, we describe the $w$-variant $\phi_i$ of the Joyal bijection and we prove certain of its properties announced in Section~\ref{sec:backbone}.
In Section~\ref{sec:proofsAsymptotic}, we prove certain asymptotic statements from Section~\ref{sec:backbone}, leaving to the remaining sections the proof of the most difficult ones which depend on the exploration process (in particular Lemmas~\ref{lemma:AL1},~\ref{lemma:AL4},~\ref{lemma:AL5} and~\ref{lemma:AL6}).
In Section~\ref{sec:exploration} we introduce the formalism of our \emph{exploration process}, and we state a number of useful facts, first about deterministic, and then about random, exploration processes in automata. We also present a ``toolbox'' of lemmas that will be used in the following sections.
In Section~\ref{sec:cycleGood}, using the formalism developed in Section~\ref{sec:exploration}, we prove Lemmas~\ref{lemma:AL1} and~\ref{lemma:AL4} from Section~\ref{sec:backbone}, which give control on collision events ``along cycles'' in random automata.
In Section~\ref{sec:branchGood}, we prove our most technical statement, Lemma~\ref{lemma:AL5}, which deals with collision events ``along branches''. The techniques are very similar to Section~\ref{sec:cycleGood} but this section is longer, since, as explained above, our proof technique requires to distinguish a number of cases according to a number of induced structures that our $w$-explorations may reveal (and the number of cases turns out to be higher for that proof).
In Section~\ref{sec:telescopic} we prove Lemma~\ref{lemma:AL6}, which is a (relatively) simple consequence of Lemma~\ref{lemma:AL5}. We conclude in Section~\ref{sec:lower} with the proof of the lower bound (Theorem~\ref{thm:lower}).

\section{Notation}
\label{sec:notation}

\subsection{Sets of automata, marked points, labellings}

In this paper $n\geq1$ is an integer. We will work with various kinds of automata on the vertex set $[n]$. To keep the text lighter the dependency in $n$ will not appear in our notation except for the set $[n]$ itself. Unless otherwise mentioned, all asymptotic notation such as $O(\cdot), o(\cdot), \sim$ are relative to $n$ going to infinity. We also write $\nota{u_n = \widetilde{O}(v_n)}$ if $u_n = O((\log(n)^K v_n)$ for some $K=O(1)$.

We let \nota{$\AAA$} be the set of automata with alphabet $\{a,b\}$ on $[n]$. An element of $\AAA$ is determined by the choice of two transition functions $[n]\rightarrow [n]$ (one for each letter) and we have $|\mathcal{A}|=n^{2n}$.
We let \nota{$\AAb:=\AAA\times [n]$} and \nota{$\AAbb:=\AAA\times [n]\times [n]$} be the set of automata carrying a marked vertex (resp. two marked vertices). 

Our main constructions will require to work with automata whose vertices carry an additional labelling, which will be encoded by a permutation in \nota{$\mathfrak{S} := \mathfrak{S}_n$}. Moreover, in some situations we will need to work with two such sets of labellings. For this we introduce the notation:
\begin{align}\label{eq:product}
\nota{\AAl:=\AAA\times \mathfrak{S}}
\ \ , \ \ 
\nota{\AAll:=\AAA\times \mathfrak{S}\times \mathfrak{S}}
\ \ , \ \ 
\nota{\AAbl:=\AAA\times [n] \times \mathfrak{S}}
\ \ , \ \ 
	\nota{\AAbbll := \mathcal{A}\times [n]\times [n]\times \mathfrak{S}\times \mathfrak{S}}.
\end{align}

	In what follows we will use the letter $\sigma$ to denote a generic element of $\mathfrak{S}$ and $v$ or $u$ for an element of $[n]$. These symbols will often carry an index $i\in\{1,2\}$ that will be relative to an underlying word\footnote{
		Later in the paper we will work with two words denoted by $w_1$ and $w_2$.
		Sometimes only one word comes into play but in order to avoid using more notation we still denote this word by $w_i$ (rather than just $w$ for example) for a fixed $i\in\{1,2\}$.}.
A generic element of $\AAbl$ will be denoted by $(A,v_1,\sigma_1)$ or $(A,v_2,\sigma_2)$ and a generic element of $\AAbbll$ by $(A,v_1,v_2,\sigma_1,\sigma_2)$.
	For $i\in \{1,2\}$ we let \nota{$\pi_i$} be the ''forgetful'' operation, defined on tuples, that erases any coordinate carrying an index different from $i$. For example, if $x=(A,v_1,v_2,\sigma_1,\sigma_2) \in \AAbbll$, then $\pi_1(x)= (A,v_1,\sigma_1) \in \AAbl$.
The use of this slightly abusive notation should be clear in context.

In the next section we will define several subsets of automata with or without marked vertices or extra labellings, and we find helpful to keep the $^\bullet$ and $^\diamond$ superscripts in the notation to indicate how many marked points and labellings are carried by the objects. For example, the set $\BBB$ defined below will be a subset of $\AAbbll$, but it does \emph{not} (at all) have a product structure similar to~\eqref{eq:product}.

\subsection{$w$-trees, $w$-threads, shifts and hats}\label{sec:w-trees}

	A one-letter automaton, which is the same as a function $f:[n]\rightarrow [n]$, can be pictured as a set of directed cycles on which directed trees are attached (see Figure~\ref{fig:function}). We say that a point $v\in [n]$ is \nota{cyclic} if it belongs to one of these cycles, or formally if there exists $k>0$ such that $f^k(v)=v$. As explained in the introduction, the one-letter automata with a unique cyclic point are in direct bijection with rooted Cayley trees, and, from now on, we will slightly abusively use the word \nota{rooted tree} for them.

\begin{definition}[Cyclic points and $w$-trees]
	If $w$ is a word and $A\in \AAA$ is an automaton, we let $A_w$ be the one-letter automaton defined by $w$-transitions in $A$. A vertex of $\AAA$ is \nota{$w$-cyclic} if it is a cyclic point of $A_w$.
	We say that $A$ is a \nota{$w$-tree} if $A_w$ is a rooted tree, or equivalently if it has a unique cyclic point.
\end{definition}
	If the automaton $A$ is chosen uniformly at random, the edges of $A_w$ are not independent from each other. This fact is the main difficulty underlying this paper. For example, it is not even clear how to generate a uniformly random $w$-tree other by try-and-reject.

	Let $k\geq1$ and let $w=w^1w^2\dots w^k$ be a word of length $k$. We denote by \nota{$[[k]]$} the set of integers considered modulo $k$, which indexes the letters of $w$.
\begin{definition}[$w$-thread]
Given an underlying automaton $A$, a vertex $u\in [n]$ and a congruence $r\in [[k]]$, we consider the infinite sequence of $A$-transitions
	\begin{align}\label{eq:infiniteThread}
u=: u_0 \stackrel{w^{r+1}}{\longrightarrow} u_1 \stackrel{w^{r+2}}{\longrightarrow} u_2 \stackrel{w^{r+3}}{\longrightarrow} u_3
	\stackrel{w^{r+4}}{\longrightarrow} \dots,
	\end{align}
with indices of $w$ considered modulo $k$.
	We let $T$ be the first time a pair (vertex, congruence time) is repeated in this sequence, %
	\begin{align}
		T &:= \min\{j \geq 1, \exists i <j \mbox{ such that } u_i=u_j \mbox{ and } i\equiv j \mod k\},
	\end{align}
	We define the \nota{$w$-thread of $(u,r)$} as the sequence of triples\footnote{Working with triples will be convenient later when we consider explorations with two different words simultaneously.}
	$$\nota{[\Athread{u}{r}{w}{A}]}:=((u_i,r+i,w))_{0\leq i \leq T-1},$$
where the value $r+i$ is considered modulo $k$.
\end{definition}

	Note
		that the infinite sequence~\eqref{eq:infiniteThread} becomes periodic after time $T$ which is why we cut it at this point. 
		More precisely, the pattern formed by the last $T-\tilde{T}$ values is eventually repeated, where $\tilde{T}$ is the unique $i<T$ such that $u_i=u_T$ and $i\equiv T \mod k$. Note that vertices visited at congruence $0$ along this pattern correspond to a cycle of length 
	$$\nota{\cyc{w}{v,r}}:=(T-\tilde{T})/k$$ in the one-letter automaton $A_w$.
We denote by \nota{$\cyc{w}{v}$}$:=\cyc{w}{v,0}$, which is the length of the unique cycle of $A_w$ in the connected component of $v$.
We say that the thread is \nota{cyclic} if $\tilde{T}=0$, i.e. if the thread is equal to the ultimately repeated subpattern in~\eqref{eq:infiniteThread}.

If $u,v$ are elements of $[n]$ and if $r,s \in [[k]]$, we write 
$$
\nota{\Aarrows{u}{r}{w}{A}{v}{s}}
$$
if $(v,s,w)$ appears in the $w$-thread of $(u,r)$ in $A$. If this is the case,  we write 
$$\nota{[\Aarrows{u}{r}{w}{A}{v}{s}]}:= ((u_i,r+i,w))_{0\leq i \leq \ell}$$
where $(u_\ell,r+\ell)=(v,s)$ for the partial thread between $(u,r)$ and $(v,s)$.

We extend this notion to edges. Given $e=(v_0,v)$ an edge of $A$, we write
$$
\nota{\Aarrows{u}{r}{w}{A}{e}{s}}
$$
if $(v_0,s-1,w)$ and $(v,s,w)$ both appear, and appear consecutively, in the $w$-thread of $(u,r)$ in $A$. Observe that $\Aarrows{u}{r}{w}{A}{e}{s}$ implies that $\Aarrows{u}{r}{w}{A}{v}{s}$, but the converse is not true.

If $A$ is clear by context, we write $[\thread{u}{r}{w}]$, $\arrows{u}{r}{w}{v}{s}$, $[\arrows{u}{r}{w}{v}{s}]$ and $\arrows{u}{r}{w}{e}{s}$.
For example, $\arrows{u}{0}{w}{v}{0}$ means that there is a directed path from $u$ to $v$ in the one-letter automaton $A_w$.

\medskip

For combinatorial purposes, it is natural to consider $w$-trees. However, for applications to synchronizing words, and in particular to simplify the further application of the second moment method, it will make things simpler to work with a particular class of $w$-trees, the $0$-shifted ones. 
	\begin{definition}[shifts, see Figure~\ref{fig:bs-and-betas0}]
	Let $(A,v)\in \AAb$ such that $A$ is a $w$-tree, and let, as above,  $T$ be the length of the $w$-thread of $(v,0)$. If $T$ is congruent to $r$ modulo $k$, we say that
		$(A,v)$ is \nota{$r$-shifted}.
\end{definition}

We end up with the following important notation convention. When this notation is used, there will always be a word $w$, or two words $w_1$ and $w_2$ which are fixed by the context.
\begin{notation}[hats]\label{not:hats}
	If $\mathcal{S}^\bullet$ is any single-bullet notation (with possibly other superscripts and subscripts) denoting a subset of $\AAb$, the notation \nota{$\widehat{\mathcal{S}^\bullet}$} denotes the set of elements $(A,v)$ in $\mathcal{S}^\bullet$ such that $A$ is a $w$-tree \emph{and} $(A,v)$ is $0$-shifted. We denote $\widehat{\AAbl}=\widehat{\AAb} \times \mathfrak{S}$.
Here either the word $w$ will be clear from context, or there will be an index $i\in\{1,2\}$ to the letter $\mathcal{S}$ refering to the word $w_i$ which is meaningful.

	Similarly, if $\mathcal{S}^{\bullet\bullet}\subset\AAbb$ is any two-bullet notation, we write \nota{$\widehat{\mathcal{S}^{\bullet\bullet}}$} for the set of $(A,v_1,v_2)$ in $\mathcal{S}^{\bullet\bullet}$  such that $A$ is a $w_i$-tree and $(A,v_i)$ is $0$-shifted, for both $i=1$ and $i=2$.

	In both cases there may be additional $^\diamond$ superscripts, in which case the underlying permutations are spectators and are unaffected by the $\widehat{~ ~ }$ notation (objects keep their underlying labelling).
\end{notation}

\subsection{Good events and collisions}
\label{sec:collisionDefs}

Recall that two words are \nota{conjugate} if they differ only by a cyclic permutation, and a word is \nota{self-conjugate} if it is nontrivially conjugate to itself, or equivalently if it is a power of a shorter word. We let $k\geq 1$ and we let \nota{$\mathcal{W}_k^{NC}$} be the set of words of length $k$ on $\{a,b\}$ which are not self-conjugate. 

From now on we fix\footnote{The constant $2$ in the bound $2\log(n)$ is arbitrary and could be replaced by any constant $>1$ at the cost of adapting certain absolute constants later in the paper.} an integer $\nota{k}\leq 2 \log(n)$ which will be the length of the words we consider.
Moreover we fix two words $\nota{w_1,w_2} \in \mathcal{W}_k^{NC}$.
Unless specified, we assume that $w_1$ and $w_2$ are not conjugate, in particular $w_1\neq w_2$.
The indices $i\in \{1,2\}$ for the various objects defined in the next definitions will be related to these two words.

The reader can use Figures~\ref{fig:bs-and-betas}-\ref{fig:bs-and-betas2} as a visual support for definitions in this section.

\begin{figure}
\centering
\includegraphics[width=0.49\linewidth]{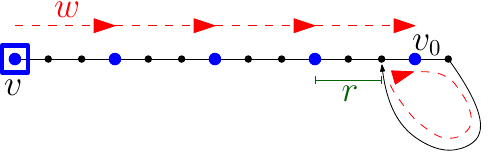}
	\caption{
	Pictorial view of an $r$-shifted marked $w$-tree $(A,v)$. Dotted red arrows represent $w$-transitions, and $v_0$ is the unique $w$-cyclic vertex in $A$. On the picture, $k=3$ and $r=2$.
	{\it We insist that some of the vertices or edges in these pictures could  be equal to some others, and that many vertices a priori appear outside of the branch and are not represented either. The same will be true in all pictures in this document.}}
\label{fig:bs-and-betas0}
\end{figure}

\begin{figure}
\centering
\includegraphics[width=0.89\linewidth]{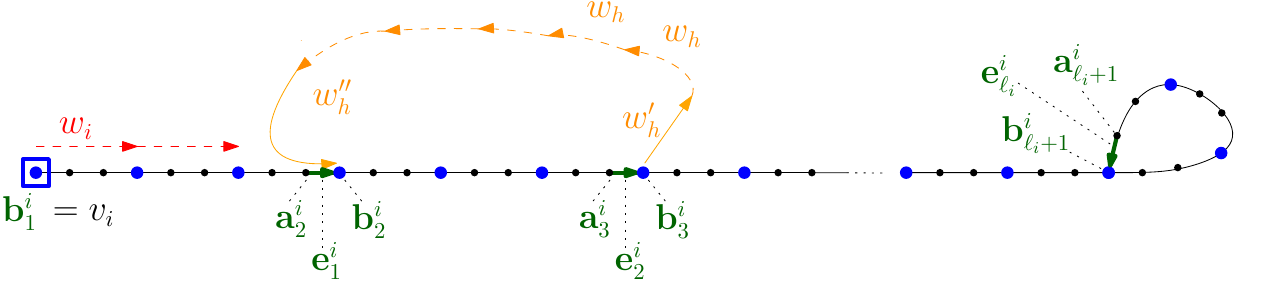}
	\caption{
		Pictorial view of the definition of vertices $\ba^{i}_p,\bb^{i}_p$ and edges $\be_p^{i}$, for an automaton $(A,v_i,\sigma_i)\in \AAbl$. Dotted red arrows represent $w_i$-transitions, and the $\bb^{i}_p$ for $p\leq \ell_i$ are lower records at congruence $0$ along the branch, for the permutation $\sigma_i$ (which is not represented here).  In this example, the vertex $\bb^i_{\ell_i+1}$ appears at congruence $0$ along the branch (this configuration is $0$-shifted).
	The orange pattern above the path pictures an $(i,h,i)$-collision. Here $w'_h$ and $w''_h$ are respectively a suffix and a  prefix of $w_h$ (the prefix is proper if $h=i$).  The $(i,h,j)$ collisions can be pictured similarly but two branches emanating from $v_1$ and $v_2$ (for the two words $w_1$ and $w_2$) have to be considered.
	}
\label{fig:bs-and-betas}
\end{figure}

\begin{figure}
\centering
\includegraphics[width=0.59\linewidth]{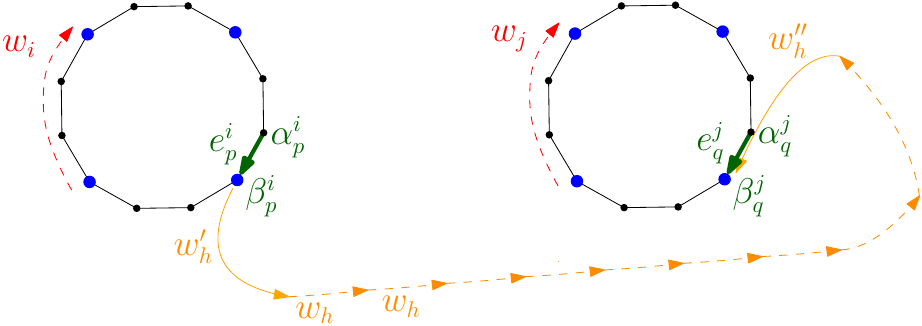}
	\caption{
		A representation similar to Figure~\ref{fig:bs-and-betas} of the vertices $\beta^{i}_p, \beta^{j}_q$, and a collision showing that this automaton is in $\SSS$ (provided $w''_h$ is a \emph{proper} prefix of $w_j$ if $h=j$).
	}
\label{fig:bs-and-betas2}
\end{figure}

\begin{notation}[Lower records along the branch of $v$]\label{not:branch_minima}
	Given $x=(A,v,\sigma) \in \AAbl$ and a word $w$, we let $\tilde{T}$ be as above and $v=u_0,u_1,u_2,\dots ,u_{\tilde{T}}$ be the sequence of vertices appearing in the $w$-thread of $v$ before the periodic part. Along this sequence we let $0=ki_1<ki_2<\dots<ki_{\ell_w(A,v)} \leq \tilde{T}$ be the times of lower records of $\sigma$-label, among times with congruence $0$. In other words, $i_1,\dots,i_{\ell_w(A,v)}$ are the lower records of the sequence $(\sigma(u_{ki}))_{i=0}^{\lfloor{}\tilde{T}/k\rfloor}$.
We let $\nota{(\bb^{(w)}_p(x))_{1\leq p \leq \ell_w(A,v)}}:=(v=u_{ki_1},u_{ki_2},\dots ,u_{ki_{\ell_w(A,v)}})$
	be the sequence of vertices realizing these lower records. We also let $\nota{\bb^{(w)}_{\ell_w(A,v)+1}(x)}:=u_{\tilde{T}}$ be the vertex corresponding to the first repeated pair (vertex, time congruence) along the $w$-thread of $v$.
	Note that it could be equal to $\bb^{(w)}_{\ell_w(A,v)}(x)$ if $\tilde{T}$ is a multiple of $k$ and is a time of lower record.
	For $p \in [\ell_w(A,v)]$, we let $\nota{\be^{(w)}_p(x)} = (\ba^{(w)}_{p+1}(x), \bb^{(w)}_{p+1}(x))$
	be the edge preceding $(\bb^{(w)}_{p+1}(x),0,w)$ on the thread
	$[\thread{v}{0}{w}]$, where \nota{$\ba^{(w)}_{p+1}(x)$} is the tail of $\be^{(w)}_p(x)$.
		
	Given two configurations $x_1=(A,v_1,\sigma_1)$ and $x_2=(A,v_2,\sigma_2)$, we will use the following shorter notation: for $i\in \{1,2\}$
\begin{equation*}	
\begin{aligned}	
\nota{\ell_{i}}&:=\ell_{w_i}(A,v_i),\\
\nota{(\bb^{i}_p)_{1\leq p \leq \ell_i+1}}&:=(\bb^{(w_i)}_p(A,v_i,\sigma_i))_{1\leq p \leq \ell_i+1},\\
\nota{(\ba^{i}_p)_{2\leq p \leq \ell_i+1}}&:=(\ba^{(w_i)}_p(A,v_i,\sigma_i))_{2\leq p \leq \ell_i+1}\\
\nota{(\be^{i}_p)_{1\leq p \leq \ell_i}}&:=(\be^{(w_i)}_p(A,v_i,\sigma_i))_{1\leq p \leq \ell_i}.
\end{aligned}
\end{equation*}
	\end{notation}

By construction, for any $1\leq p\leq q \leq \ell_w(A,v)$, the edge $\be_q^{(w)}(x)$ appears at congruence $0$ on the $w$-thread of $(\bb_p^{(w)}(x),0)$. In the following notion of ``branch-good'' configuration, we forbid all other kinds of appearances. We directly introduce the notation with an index $i$ relative to the word $w_1$ or $w_2$ considered.
\begin{definition}[Branch good events]\label{def:branch-good}
	We define subsets $\BB{1},\BB{2},\BBB$ of marked automata
	as follows:
	\begin{align*}
		\AAbl\supset	\BB{i} &:= \{(A,v_i,\sigma_i) \in \AAbl \mbox{ such that if }	\arrows{\bb_p^i}{r}{w_i}{\be_q^i}{s} \mbox{ for some }p,q\in[\ell_i]\mbox{ and }r,s\in[[k]]\mbox{ then }  s=0. \}\\
		\AAbbll \supset\BBB &:= 
		\{x\in\AAbbll, \pi_1(x)\in \BB{1} \mbox{ and } \pi_2(x)\in \BB{2}\}.
	\end{align*}
\end{definition}
Elements of the set $\BBB$ are ''branch-good'' independently for each of the two words $w_1,w_2$, but we need a finer understanding of how the two words interact. Given an element of $\AAbbll$ and $(i,h,j) \in [2]^3$ we say that we have an \nota{$(i,h,j)$-collision} if we can find $p\in [\ell_i], q\in[\ell_j], r,s\in[[k]]$,  such that
	\begin{align} \label{eq:excludeT}
		\arrows{\bb_p^i}{r}{w_h}{\be_q^j}{s}
		\mbox { with } (j,s)\neq (h,0). 
	\end{align}

For $\cI\subset[2]^3$ we let \nota{$\FFF(\cI) \subset \AAbbll$} be the set of marked automata with no $(i,h,j)$-collisions for all $(i,h,j)\in \cI$. For example,
	$$\BBB = \FFF((1,1,1),(2,2,2)).
$$
	We now define  $\PPP, \QQQ$ as the following subsets of $\BBB$,
	\begin{align*}
		\PPP &:=
		\FFF((1,1,1),(2,2,2), (1,2,1), (1,2,2), (2,2,1)), \\
		\QQQ &:=
		\FFF((1,1,1),(2,2,2), (2,1,2), (2,1,1), (1,1,2)). 
	\end{align*}
Note that
$\PPP\cap\QQQ =\FFF([2]^3)$ is the set in which all types of collisions are avoided.

\medskip

We now introduce similar notions for unmarked automata, in which the role of branch lower records is replaced by cycle minima.

\begin{notation}[Cycle minima]\label{not:cycle_minima}
	For $x=(A,\sigma) \in \AAl$ and a word $w$, let $\nota{\lambda_w(A)}$ be the number of cycles of the one-letter automaton $A_w$. For each $p\in [\lambda_w(A)]$, we let \nota{$\beta_p^{(w)}(x)$} be the vertex that minimises the $\sigma$-label among vertices appearing on the $p$-th cycle of $A_w$. Here we order cycles by decreasing minimum, i.e. we choose the indices so that  $\sigma(\beta_1^{(w)}(x))>\sigma(\beta_2^{(w)}(x))>\dots>\sigma(\beta_{\lambda_w(A)}^{(w)}(x))$. 
For $p\in [\lambda_w(A)]$, consider the edge $\nota{e_p^{(w)}(x)} = (\alpha_p^{(w)}(x), \beta_p^{(w)}(x))$ where $\alpha_p^{(w)}(x)$ is the last vertex on the thread $[\thread{\beta_p^{(w)}(x)}{A}{w}]$. Note that $e_p^{(w)}(x)$ is the edge preceding $\beta_p^{(w)}(x)$ along the path in $A$ formed by all edges corresponding to the cycle containing $\beta_p^{(w)}(x)$ in $A_w$.
We also let \nota{$\beta_{\lambda_w(A)+1}^{(w)}(x)$} be the vertex preceding the vertex $\beta_{\lambda_w(A)}^{(w)}(x)$ along its cycle, in the one-letter automaton $A_w$. Note that it could be equal to $\beta_{\lambda_w(A)}^{(w)}(x)$ if that cycle has length one.

	Given two permutations $\sigma_1,\sigma_2\in \mathfrak{S}$
	 we will use the following shorter notation for $i\in \{1,2\}$
	 \begin{equation*}
	 \begin{aligned}
	 \nota{\lambda_i}&:= \lambda_{w_i}(A)\\
	 \nota{(\beta_p^i)_{1\leq p \leq \lambda_i+1}}&:=(\beta_p^{(w_i)}(A,\sigma_i))_{1\leq p \leq \lambda_{w_i}(A)+1} \\
	 \nota{(\alpha_p^i)_{1\leq p \leq \lambda_i}}&:=(\alpha_p^{(w_i)}(A,\sigma_i))_{1\leq p \leq \lambda_{w_i}(A)} \\
	 \nota{(e_p^i)_{1\leq p \leq \lambda_i}}&:=(e_p^{(w_i)}(A,\sigma_i))_{1\leq p \leq \lambda_{w_i}(A)} .
	 \end{aligned}
	 \end{equation*}
\end{notation}

The following notion constrains the way edges can appear in the thread of $\beta$-vertices. It depends on an index $i\in \{1,2\}$ referring to the underlying word.
\begin{definition}[Cycle good events]\label{def:cycle-good}
	We define subsets \nota{$\CC{1},\CC{2}\subset \AAl$} as

	\begin{align*}
		\CC{i} &:= \{(A,\sigma_i) \in \AAl \mbox{ such that if }	(\beta_p^i,r) \stackrel{w_i}{\longrightarrow}	(e_q^i,s) \mbox{ for some }p,q\in [\lambda_i]
\mbox{ and }r,s\in[[k]]
		\text{, then }  s=0. \}.
	\end{align*}

	Finally, we define the set \nota{$\SSS\subset \AAll$} consisting of all automata which are such that there exists $(i,h,j)\in[2]^3$, $p\in [\lambda_i], q\in[\lambda_j]$ and $r,s\in[[k]]$ such that
	\begin{align}\label{eq:exclude}
		(\beta_p^i,r) \stackrel{w_h}{\longrightarrow}	(e_q^j,s),
	\end{align}
	with $(j,s)\neq (h,0)$. They can be thought of automata in which some ``bad'' collision appears, in the form of an edge $e$ that will be rewired by the bijection appearing in the thread of a $\beta$ vertex in an unexpected way, for any of the two words involved.
\end{definition}

\begin{remark}\label{rem:s=0}
	It is crucial to exclude the case $(j,s)=(h,0)$ from~\eqref{eq:exclude}. Indeed, for any vertex $u\in[n]$ and any congruence $r\in [[k]]$ the thread $[\thread{u}{r}{w_j}]$ eventually reaches a cycle of $A_{w_j}$, hence visits one of the edges $e^j_q$ at a time congruent to $s=0$. This is true in particular for $u=\beta_p^i$, so including $(j,s)=(h,0)$ in the definition of $\SSS$ would be pointless (all automata would be in $\SSS$). As we will see, apart from this obvious case, all the other types of collisions are unlikely to happen in random automata.
	Similarly, we excluded the case $(j,s)=(h,0)$ in~\eqref{eq:excludeT}. Indeed, 
	if the automaton is a $w_j$-tree
	the thread $[\thread{u}{r}{w_j}]$ eventually reaches the branch $[\thread{v_j}{0}{w_j}]$ of the marked vertex $v_j$, hence visits some of the edges $\be^j_q$ at congruence $s=0$.
\end{remark}

\begin{remark}\label{rem:not_LR}
	Observe that in Definition~\ref{def:branch-good}, we may have $q = \ell_j$. In this case, the endpoint of the "collision" path~\eqref{eq:excludeT} (i.e., the endpoint of $\be_{\ell_j}^j$) is the vertex $b_{\ell_j+1}^j$ which is not defined as a lower-record of $\sigma_i$ on the branch. This will add a level of complexity in the proofs as we will sometimes have to address this case separately. However, note that the initial point $b_{p}^i$ of that path is always a lower record since $p\in [\ell_i]$.
	Finally, note that the situation is simpler in Definition~\ref{def:cycle-good}, since both the initial point $\beta_p^i$ and the endpoint $\beta_q^j$ of the path~\eqref{eq:exclude} are lower records on their respective cycles by construction.
\end{remark}

\section{Proof backbone}
\label{sec:backbone}

\subsection{First moment}

In this section we are interested in counting $w$-trees (or more precisely, $0$-shifted branch-good  marked ones) for a fixed word $w$. For compatibility of notation with other sections, we will fix $i\in\{1,2\}$ and the role of the generic word $w$ will played by the word $w_i$. Recall that we assumed that $w_i\in \mathcal{W}_k^{NC}$, i.e. $w_i$ is not self-conjugate.
We have the following generalization of Joyal's classical bijection. It is defined only between the cycle-good and the $0$-shifted branch-good configurations previously introduced
(to anticipate, the reader can look at Figure~\ref{fig:wJoyal}).
\begin{proposition}[Joyal's bijection for $w$-trees] \label{prop:BP1}
	There is a bijection 
	$$\phi_i: \CC{i} \longrightarrow \widehat{\BB{i}}.$$
\end{proposition}
Here we recall (Notation~\ref{not:hats}) that $\widehat{\BB{i}}$ denotes the elements of $\BB{i}$ which are $w_i$-trees and are $0$-shifted. 
The next proposition shows that most labelled automata are in fact in  $\CC{i}$. Here and everywhere in the paper, given any finite set $\mathcal{U}$ we use \nota{$\mathbb{P}_\mathcal{U}$} for the uniform probability on the set $\mathcal{U}$.

\begin{lemma}\label{lemma:AL1} For $i=1,2$, we have
	\begin{align*}
		\mathbb{P}_{\AAl}((A,\sigma_i) \not \in \CC{i}) =  \widetilde{O}\left(\frac{1}{\sqrt{n}}\right).
	\end{align*}
\end{lemma}

\begin{proposition}[Main first moment estimate]\label{prop:AP1}
	We have $|\widehat{\BB{i}}| \sim \frac{|\AAbl|}{n}$.
\end{proposition}

\subsection{Second moment}

For the second moment we will need to count configurations in $\AAbb$ which are both (branch good, $0$-shifted) $w_1$-trees and $w_2$-trees. For this, we will apply the Joyal bijection twice (once for each word).

For $i\in \{1,2\}$ we extend the bijection $\phi_i^{-1}:\widehat{\BB{i}}\rightarrow \CC{i}$ to doubly marked elements $x = (A,v_1,v_2,\sigma_1,\sigma_2)\in \AAbbll$ by acting only on the triple $\pi_i(x)$ and leaving the other coordinates invariant. For example, in the case that $\phi_1^{-1}(A,v_1,\sigma_1) = (A', \mu_1)$, we set 
$\phi_1^{-1}(x):=(A',v_2,\mu_1, \sigma_2)$.
Of course, this is defined only if $(A,v_1,\sigma_1)$ is in the set $\widehat{\BB{1}}$. Similarly, we can act on $x':=(A',v_2,\mu_1, \sigma_2)$ with $\phi_2^{-1}$ by letting $y:=\phi_2^{-1}(x')=(A'',\mu_1, \mu_2)$ where $\phi_2^{-1}(A',v_2,\sigma_2)=(A'',\mu_2)$ -- again this is defined only if $(A',v_2,\sigma_2)$ is in $\widehat{\BB{2}}$. This convention enables us to compose the mappings $\phi_1^{-1}$ and $\phi_2^{-1}$ in various ways. The following proposition studies these compositions.

\begin{proposition} \label{prop:i-iv}
	Let $x=(A,v_1,v_2,\sigma_1,\sigma_2)$. The following are true:
	\begin{enumerate}[label=(\roman*)]
		\item  $\phi_2^{-1}\circ \phi_1^{-1}:
			\widehat{\PPP} \longrightarrow \AAll$ is well defined, that is to say, $\pi_1(\widehat{\PPP})\subset \widehat{\BB{1}}$ and $\pi_2(\phi_1^{-1}(\widehat{\PPP}))\subset \widehat{\BB{2}}$. 
			Moreover, suppose that $x\in\widehat{\PPP}$ and $y=\phi_2^{-1}\circ \phi_1^{-1}(x)$. 
			Then the branch lower records $(\bb_p^2)_{1\leq p \leq \ell_2+1}$  of $\pi_2(x)$ and their neighbours $(\ba_{p+1}^2)_{1\leq p \leq \ell_2}$,
		and the cycle minima  $(\beta_p^2)_{1\leq p \leq \lambda_2+1}$ of $\pi_2(y)$ and their neighbours $(\alpha_p^2)_{1\leq p \leq \lambda_2}$, defined as in Section~\ref{sec:collisionDefs}, coincide. Namely, we have
		       	$\lambda_2=\ell_2$, $\beta^2_p=\bb^2_p$ for all $p\in [\ell_2+1]$ and $\alpha^2_p=\ba^2_{p+1}$ for all $p\in [\ell_2]$.

		\item  $\phi_1^{-1}\circ \phi_2^{-1}:
			\widehat{\QQQ} \longrightarrow \AAll$ is well defined, that is to say, $\pi_2(\widehat{\QQQ})\subset \widehat{\BB{2}}$ and $\pi_1(\phi_2^{-1}(\widehat{\QQQ}))\subset \widehat{\BB{1}}$. 
			Moreover, suppose that $x\in\widehat{\QQQ}$ and $y=\phi_1^{-1}\circ \phi_2^{-1}(x)$. 
			Then the branch lower records $(\bb_p^1)_{1\leq p \leq \ell_1+1}$  of $\pi_1(x)$ and their neighbours $(\ba_{p+1}^1)_{1\leq p \leq \ell_1}$, and the cycle minima  $(\beta_p^1)_{1\leq p \leq \lambda_1+1}$ of $\pi_1(y)$ and their neighbours $(\alpha_p^1)_{1\leq p \leq \lambda_1}$, defined as in Section~\ref{sec:collisionDefs}, coincide. Namely, we have	
		       	$\lambda_1=\ell_1$, $\beta^1_p=\bb^1_p$ for all $p\in [\ell_1+1]$  and $\alpha^1_p=\ba^1_{p+1}$ for all $p\in [\ell_1]$.
\item $\phi_2^{-1}\circ \phi_1^{-1}= \phi_1^{-1}\circ \phi_2^{-1}$ on $\widehat{\PPP}\cap \widehat{\QQQ}$.			
		\item If there exist $x\neq y$ in $\widehat{\BBB}$, $z\in \AAbbll$,
			such that $\phi_1^{-1}\circ \phi_2^{-1}(x) = z = \phi_2^{-1}\circ \phi_1^{-1}(y)$, then $z\in \SSS$.
	\end{enumerate}
\end{proposition}

\begin{figure}
	\begin{center}
		\includegraphics[width=0.8\linewidth]{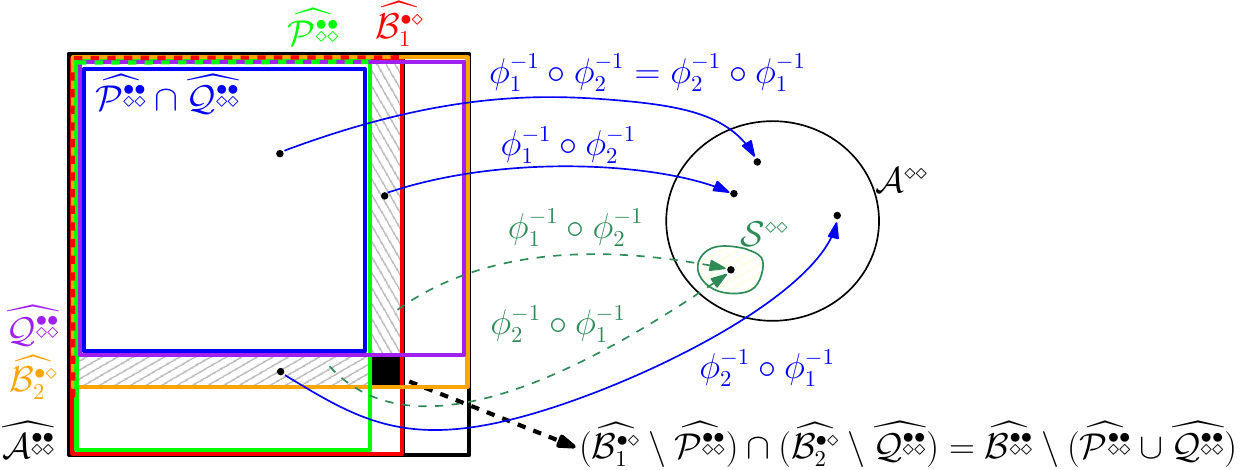}
	\end{center}
	\caption{The setup of
	Proposition~\ref{prop:set2nd}. Here the inclusions  $\widehat{\BB{1}}, \widehat{\BB{2}}\subset \widehat{A^{\protect\substack{\protect\vspace{-.19em}\bullet\bullet \\ \protect\vspace{-.1em} \diamond\diamond}}}$ implicitly use the forgetful maps $\pi_1$ and $\pi_2$.}
	\label{fig:set2nd}
\end{figure}

The previous claims imply (see Figure~\ref{fig:set2nd}): 
\begin{proposition}[Set cardinalities preparing the second moment bound]\label{prop:set2nd}
	We have 
	\begin{align}\label{eq:mainSetEqSecond}	
		\left| \widehat{\BBB}\right|
		\leq \left|\AAll\right|
		+ \left| \SSS \right|
		+ \left|
		\widehat{\BBB}\setminus ( \widehat{\PPP}\cup \widehat{\QQQ})\right|
	\end{align}
\end{proposition}

The next statements show that only the first term is asymptotically dominant in the last equation (Proposition~\ref{prop:AP2} below).  We start with the set $\SSS$. The following lemma says that most  doubly labelled automata avoid the set $\SSS$.

\begin{lemma}\label{lemma:AL4} We have
	\begin{align*}
		\mathbb{P}_{\AAll}( (A,\sigma_1,\sigma_2)\in \SSS) = \widetilde{O}\left(\frac{1}{\sqrt{n}}\right).
	\end{align*}
\end{lemma}

To estimate the cardinality of $\widehat{\PPP}\cup \widehat{\QQQ}$ in~\eqref{eq:mainSetEqSecond} we need bound the probability for a doubly marked automaton to have two different types of $(i,h,j)$-collisions while at the same time being a $w_1$- and a $w_2$-tree.
The next lemma gives a bound on a simpler event that replaces the last property by the fact that the cycles ending the $w_i$-threads of the two marked vertices are both of length $1$. Because this event only depends on the $w_i$-threads of the two marked vertices and their lower records, we will manage to attack it with the exploration techniques developped in Section~\ref{sec:exploration}.

Define $\cI_1=\{(1,2,1), (1,2,2), (2,2,1)\}$  and $\cI_2=\{(2,1,2), (2,1,1), (1,1,2)\}$. 
\begin{lemma}\label{lemma:AL5} For $I_1\in \cI_1$ and $I_2 \in \cI_2$, we have
	\begin{align*}
		\mathbb{P}_{\AAbbll}(
		(A,v_1,v_2,\sigma_1,\sigma_2) \not \in (\FFF(I_1)\cup \FFF(I_2))  \mbox{ and }
		\cyc{w_1}{v_1}=\cyc{w_2}{v_2}=1 )
		 = \widetilde{O}\left(\frac{1}{n^2}\right).
	\end{align*}
\end{lemma}

One then needs to improve on the previous lemma by enforcing configurations to be trees (rather than the weaker property that two marked vertices being attached to cycles of length $1$). Heuristically, in analogy with one-letter automata, this should decrease the probability by an additional factor of $\tO(1/n)$ leading to $\tO(n^{-3})$. However important technical difficulties appear and we will only show the following much weaker lemma. It only improves the previous bound by an arbitrary polylogarithmic factor, but is enough for our purposes. It will be proved in Section~\ref{sec:telescopic} using a many-vertex exploration procedure.

\begin{lemma}\label{lemma:AL6} %
For $I_1\in \cI_1$ and $I_2 \in \cI_2$, we have
	\begin{align*}
		\mathbb{P}_{\AAbbll}(
			(A,v_1,v_2,\sigma_1,\sigma_2) \not \in (\FFF(I_1)\cup \FFF(I_2)) \mbox{ and } (A,v_1,v_2)\in \widehat{\AAbb}   )
		 = o\left(\frac{1}{n^2}\right).
	\end{align*}
\end{lemma}

The previous lemmas easily imply

\begin{lemma}\label{lemma:AL7} We have %
	\begin{align*}
		\left|\widehat{\BBB}\setminus ( \widehat{\PPP}\cup \widehat{\QQQ})\right| 
		= o(|\AAll|).
	\end{align*}
\end{lemma}

\begin{proposition}[Main second moment bound]\label{prop:AP2} We have
	\begin{align*}
		|\widehat{\BBB}| \leq |\AAll|\cdot (1+o(1)).
	\end{align*}
\end{proposition}

\subsection{Probabilistic consequences}

\begin{theorem}\label{thm:momentEstimates}
Let $\epsilon\in (0,1)$ and 
$k=k(n)\sim (1+\epsilon)\log(n)$.
	For an automaton $A \in \mathcal{A}$, let $X_k(A)$ be the number of triples $(v,\sigma,w)$ such that $v\in [n]$, $\sigma\in \mathfrak{S}$, $w\in\mathcal{W}_k^{NC}$, and $(A,v,\sigma)$ is a marked $w$-tree which is $0$-shifted and branch-good (i.e. $(A,v,\sigma) \in \widehat{\BB{w}}$). Then we have when $n$ goes to infinity
\begin{align} %
		\frac{1}{n!}\mathbb{E}_\mathcal{A} [ X_k(A) ] \sim 2^k 
		\ \ , \ \  
		\frac{1}{n!^2}	\mathbb{E}_\mathcal{A} [ X_k(A)^2 ] \sim 4^k.
\end{align}
\end{theorem}

Theorem~\ref{thm:momentEstimates} directly implies:
\begin{corollary}[implies Theorem~\ref{thm:structure}]\label{cor:X>1}
	Let $\epsilon \in (0,1)$ and 
	$k=k(n)\sim (1+\epsilon)\log(n)$. 
	Then for a random uniform automaton $A \in \mathcal{A}$ we have  $X_k(A) \geq 1$ w.h.p..
	In particular, w.h.p. there exists a (random) word $w$ of length $k$ such that $A$ is a $w$-tree.
\end{corollary}

Define
the \nota{height} of a vertex $v \in [n]$ in a one-letter automaton $A$ as the length of the longest non-intersecting chain of transitions starting from $v$.

\begin{lemma}\label{lemma:boundedheight}
	Let $k \in [\log(n),2\log(n)]$. W.h.p., simultaneously for all words $w$ of length $k$, the one-letter automaton $A_w$ has height at most $5 \sqrt{n}$. 
\end{lemma}
Note that the previous lemma is not true for all values of $k$ (for example for $k=1$ the probability that the height of a random tree is at most $5\sqrt{n}$ converges to a constant in $(0,1)$).

The previous results imply easily the following, which is our main result

\begin{theorem}[main result, Theorem~\ref{thm:mainIntro} reformulated]\label{thm:main}
	There exists a constant $K$ such that, w.h.p, a uniform random automaton $A$ on $[n]$ has a synchronizing word of length at most $K\log (n) \sqrt{n}$.
\end{theorem}

The rest of this paper is dedicated to the proof of the statements of this section. To help the reader check completeness, here is where the proofs can be found:
Statements~\ref{prop:BP1},~\ref{prop:i-iv}, ~\ref{prop:set2nd} are proved in Section~\ref{sec:bijProofs}. 
Statements \ref{prop:AP1}, \ref{lemma:AL7}, \ref{prop:AP2}, \ref{thm:momentEstimates}, \ref{cor:X>1}, \ref{thm:main} are proved in Section~\ref{sec:proofsAsymptotic}.
Lemma~\ref{lemma:boundedheight} is proved in Section~\ref{sec:exploration} (see Remark~\ref{rem:proofImpliesProof}).
Lemmas \ref{lemma:AL1} and \ref{lemma:AL4}  are proved in Section~\ref{sec:cycleGood}, and Lemma~\ref{lemma:AL5} is proved in Section~\ref{sec:branchGood}.
Finally, Lemma~\ref{lemma:AL6} is proved in Section~\ref{sec:telescopic}.

\section{The $w$-variant of the Joyal bijection and its properties}
\label{sec:bijProofs}

In this section we present the needed variant of the Joyal bijection and prove its main properties.

\subsection{Construction of the mappings $\phi$ and $\psi$}

\begin{figure}[h!!!!!!]\centering
\includegraphics[width=\linewidth]{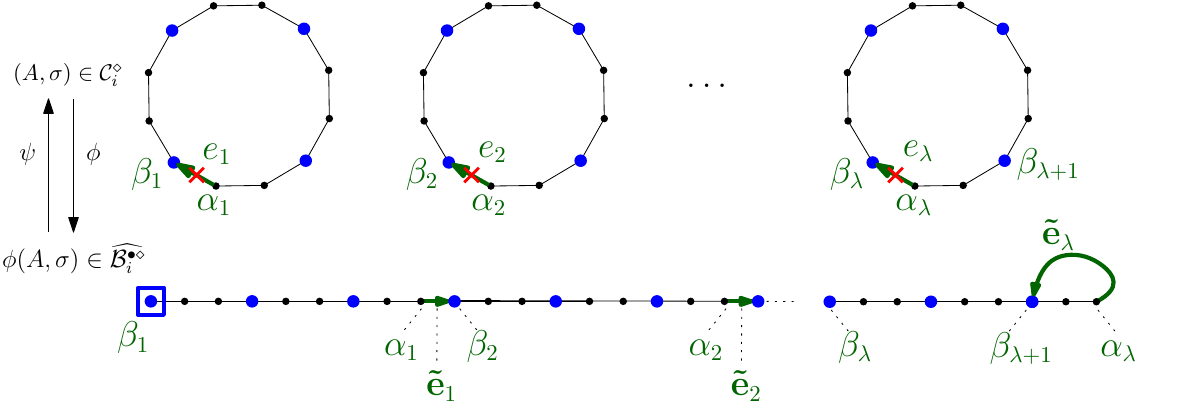}
	\caption{The mapping $\phi$ and its inverse $\psi$. Both are adaptation of the classical construction of Joyal, but they require that underlying automata and $0$-shifted marked trees are ''good'' (respectively in $\CC{}$ and $\widehat{\BB{}}$). Several vertices or edges appearing on this picture could be equal to some others, moreover the "tree-parts" attached to the cycles or to the branch are not represented (note however that the edges $e_p, \tilde{e}_p$ may be present also in these tree parts and induce rewirings in them).}\label{fig:wJoyal}
\end{figure}

Here we present the bijection $\phi$ between $\CC{i}$ and $\widehat{\BB{i}}$ which is claimed to exist in Proposition~\ref{prop:BP1}. In the rest of the section, and unless it is needed, we will drop the index $i$ from the notation.

For $(A,\sigma)\in \AAl$, recall the definition of $\beta_p$, $\alpha_p$ and $e_p$ in Notation~\ref{not:cycle_minima} (see also Figure~\ref{fig:wJoyal}).

We 
define $\nota{\phi: \AAl \longrightarrow \AAbl}$ by $\phi(A,\sigma)=(A_0,v_0,\sigma)$, where $v_0=\beta_1$ and $A_0$ is obtained from $A$ by replacing for all $p \in [\lambda]$ the edge $e_p$ by 
$$
\tilde{e}_p = (\alpha_p, \beta_{p+1}).
$$
In what follows we call $\tilde{e}_p$ a \nota{new} edge.

\medskip

We now construct another application that will play the role of inverse of $\phi$ when restricted to the desired domain.
For $(A,v,\sigma)\in \AAbl$, recall the definition of $\bb_p,\ba_p$ and $\be_p$ in Notation~\ref{not:branch_minima}.
We define $\nota{\psi: \AAbl \longrightarrow \AAl}$ by $\psi(A,v,\sigma)=(A_0,\sigma)$, where $A_0$ is obtained from $A$ by replacing for all $p \in [\ell]$ the edge ${\be}_p$ by 
$$
\tilde{\be}_p = (\ba_{p+1}, \bb_{p}).
$$
Again we call $\tilde{\be}_p$ a \nota{new} edge. The mappings $\psi$ and $\phi$ are illustrated on Figure~\ref{fig:wJoyal}.

\subsection{Properties of the mappings $\psi$ and $\phi$}

First of all, we prove two useful claims.
\begin{claim}\label{claim:useful1}
Let $(A,\sigma)\in \CC{}$ and $\phi(A,\sigma)=(A_0,v,\sigma)$. For all $p\in [\lambda]$, $[\Aarrows{\beta_p}{0}{w}{A}{\alpha_{p}}{k-1}]=[\Aarrows{\beta_p}{0}{w}{A_0}{\alpha_{p}}{k-1}]$. In particular, for all $p\in[\lambda+1]$
$$
\Aarrows{\beta_p}{0}{w}{A_0}{\beta_{\lambda+1}}{0}.
$$
\end{claim}
\begin{proof}
	Suppose that the thread in $A$ and in $A_0$ are different. Then, there exists $q\in [\lambda]$ such that $e_q$ appears in the thread in $A$, as these are the only edges that are changed by $\phi$. In particular, $\Aarrows{\beta_{p}}{0}{w}{A}{e_q}{s}$ for some $s\in [[k]]$. As $(A,\sigma)\in \CC{}$, we have $s=0$.
	But by definition, the cycle of $\beta_p$ contains a unique edge $e_i$ with $i\in [\lambda]$ appearing at congruence $(k-1,0)$, namely $e_p$. So $e_q=e_p$ and $(\alpha_{p},k-1)$ appears twice in the thread: one after the edge $e_q$ and one at the end of it. This is a contradiction with the fact that a thread contains no repeated pair (vertex,congruence).

The second statement is an immediate consequence of the first one, and the fact that the edges $\tilde{e}_p=(\alpha_p,\beta_{p+1})$ exist in $A_0$ and are traversed at congruence $(k-1)\to 0$ in each of the threads above.
\end{proof}

\begin{claim}\label{claim:useful2}
Let $(A,v,\sigma)\in \widehat{\BB{}}$ and $\psi(A,v,\sigma)=(A_0,\sigma)$. For all $p\in [\lambda]$, $[\Aarrows{\bb_p}{0}{w}{A}{\ba_{p+1}}{k-1}]=[\Aarrows{\bb_p}{0}{w}{A_0}{\ba_{p+1}}{k-1}]$. In particular, for all $p\in [\lambda]$,
	$\bb_p$ is in a cycle of $(A_0)_w$,
	and the edge $e_p$ preceding $(\bb_p,0)$ in $A_0$ is the unique edge among $\{e_q,  q \in [\ell]\}$ 
in $[\thread{\bb_p}{0}{w,A_0}]$.

\end{claim}
\begin{proof}
The proof is analogous to the proof of Claim~\ref{claim:useful1} using $(A,v,\sigma)\in \widehat{\BB{}}$. %
\end{proof}

\begin{proof}[Proof of Proposition~\ref{prop:BP1}]

We need to prove that if $x=(A,\sigma)\in \CC{}$, then $\phi(x)= y=(A_0,v,\sigma)\in \widehat{\BB{}}$. We do it in three steps:
\begin{itemize}

\item[-] $y\in \widehat{\AAbl}$: We need to show that $A_0$ is a $w$-tree and that $(A_0,v)$ is $0$-shifted.

	We first show that there is a cyclic thread\footnote{By a slight abuse we also view a cyclic thread as a closed path in $A$, by adjoining to it the last edge coming back to its starting element.} 
		$C_0$ of length $k$ in $A_0$ with $(\beta_{\lambda+1},0)\in C_0$.  By Claim~\ref{claim:useful1}, be have $\Aarrows{\beta_{\lambda+1}}{0}{w}{A_0}{\alpha_\lambda}{k-1}$. Together with the edge $\tilde{e}_\lambda=(\alpha_\lambda,\beta_{\lambda+1})$ traversed at congruence $(k-1)\to 0$, it forms the desired $C_0$.

Let us show that $C_0$ is the only cyclic thread starting at congruence $0$ in $A_0$, so $A_0$ is a $w$-tree. If $C$ is such a cyclic thread in $A_0$, then $C$ contains an edge $\tilde{e_q}$ for some $q\in [\lambda]$. Indeed, all cyclic threads in $A$ have lost at least one edge, so any cyclic thread in $A_0$ needs to contain a new edge
		$\tilde{e}_{q}$ for some $q \in [\lambda]$. Say $\tilde{e}_{q}$ is visited on the cycle at congruence $(s-1,s)$ for $s\in [k]$.
Let $(\beta_{p},r)$ with $p\in [\lambda]$ be the 
		last occurence of a vertex among $\{\beta_{p'}, p' \in [\lambda]\}$
		in $C$ preceding that edge;
		since $C$ is a cycle such an instance exists, with possibly $(p,r)=(q+1,s)$. 
		Moreover, we have
		$\Aarrows{\beta_p}{r}{w}{A}{e_q}{s}$ 
		since no other edge can have been rewired on this path as we took the first occurence of a $\beta$ vertex.
As $x\in \CC{}$, we have $s=0$,
		which also implies $\Aarrows{\beta_p}{r}{w}{A_0}{\beta_{q+1}}{0}$.
By Claim~\ref{claim:useful1}, $\Aarrows{\beta_{q+1}}{0}{w}{A_0}{\beta_{\lambda+1}}{0}$, so $(\beta_{\lambda+1},0)\in C$ and $C=C_0$.

Finally, Claim~\ref{claim:useful1} %
		and the fact that $x\in \CC{}$ also imply that the edge $\tilde{e}_{\lambda}$ is never visited on the path
		$\Aarrows{v=\beta_1}{0}{w}{A_0}{\beta_{\lambda+1}}{0}$. 
		
		So the thread reaches $C_0$ at $(\beta_{\lambda+1},0)$ and $(A_0,v)$ is $0$-shifted.

\item[-] $\ell=\lambda$, $\bb_p=\beta_{p}$ for $p\in [\lambda+1]$ and $\ba_{p+1}=\alpha_{p}$ for $p\in [\lambda]$:

By Claim~\ref{claim:useful1}, the vertices at congruence $0$ in $\Aarrows{v=\beta_1}{0}{w}{A_0}{\beta_{\lambda+1}}{0}$ appear in the same order as the vertices at congruence $0$ in the cycles of $A$, sorted by decreasing minimum and starting each cycle at the minimum label. It follows that lower-records in $A_0$ are in one-to-one correspondence with minima in $A$. By definition, we also have $\bb_{\lambda+1}=\beta_{\lambda+1}$. 

Similarly, by construction, we have that $\ba_{p+1}=\alpha_p$ for all $p\in[\ell]$.

\item[-] $y\in \BB{}$.

Suppose that
		$\Aarrows{\bb_p=\beta_p}{r}{w}{A_0}{\tilde{e}_q=\be_q}{s}$
		with $p,q\in[\lambda]$ and $s\neq 0$. As $x\in \CC{}$, this thread is not in $A$, and there exists $p'\in [\lambda]$ with $\tilde{e}_{p'}$ in it. Choose $p'$ so $\tilde{e}_{p'}$ is the new edge preceding $(\be_q,s)$ in the thread. Then %
$\Aarrows{\beta_{p'+1}}{r'}{w}{A}{e_q}{s\neq 0}$,
		a contradiction with $x\in \CC{}$ (if $p'+1=\lambda+1$ then it also holds that
		$\Aarrows{\beta_{\lambda}}{r'}{w}{A}{e_q}{s\neq 0}$).

\end{itemize}

We now show that if $y=(A,v,\sigma)\in \widehat{\BB{}}$ then $\psi(y)=x=(A_0,\sigma)\in\CC{}$.
\begin{itemize}
\item[-] $\lambda=\ell$ and $\beta_{p}=\bb_p$ for $p\in [\ell+1]$ and $\alpha_{p}=\ba_{p+1}$ for $p\in [\ell]$:

	For each cycle-thread $C$ in $A_0$, there exists $q\in[\ell]$ such that $\tilde{\be}_q$ is in $C$. Indeed, the only cycle in $A$ is the one containing $(\bb_{\ell+1},0)$, that loses the edge $\be_{\ell}$  in $A_0$. 
		So the edge $\tilde{\be_q}$ appears in $C$ at some congruence $(s-1,s)$ 
		for some $(q,s)\in [\ell]\times [[k]]$. Let $(\bb_p,r)$ be the last occurence of a vertex among $\{\bb_{p'},  p'\in [\ell]\}$ which precedes that edge
		in $C$. So, $\Aarrows{\bb_p}{r}{w}{A}{\be_{q}}{s}$, and since $y\in \BB{}$ and $p,q\in [\ell]$, we have $s=0$.
		By Claim~\ref{claim:useful2}, $[\Aarrows{\bb_{q}}{0}{w}{A}{\ba_{q+1}}{k-1}]=[\Aarrows{\bb_{q}}{0}{w}{A_0}{\ba_{q+1}}{k-1}]$, and $C$ is composed of the previous thread together with $\tilde{\be}_{q}$. Thus, for each cycle $C$ in $A_0$ there is a unique $q\in [\ell]$ such that $(\bb_q,0)\in C$ and $\bb_q$ is the vertex with minimum $\sigma$-label in it. By sorting the cycles by decreasing minimum and noting that $\beta_{\ell+1}=\bb_{\ell+1}$, we conclude the proof.
		
		Similarly, by construction, we have that $\alpha_p=\ba_{p+1}$ for all $p\in[\ell]$.

\item[-] $x\in \CC{}$.

Suppose that
		$\Aarrows{\beta_p=\bb_p}{r}{w}{A_0}{\tilde{\be}_q=e_q}{s}$ with $p,q\in [\lambda]$ and $s\neq 0$.
		As $y\in \BB{}$, this thread is not in $A$, and there exists $p'\in [\ell]$ with $\tilde{\be}_{p'}$ in it. Choose $p'$ so $\tilde{\be}_{p'}$ is the new edge preceding $(\be_q,s)$ in the thread. Then $\Aarrows{\bb_{p'}}{r'}{w}{A}{\be_{q}}{s\neq 0}$, a contradiction with $y\in \BB{}$. 
\end{itemize}

Finally, we show that, restricted to $\widehat{\BB{}}$, $\psi=\phi^{-1}$ and therefore $\phi$ is a bijection between $\CC{}$ and $\widehat{\BB{}}$. Recall that $\beta_p=\bb_p$ for all $p\in[\ell+1]$ and $\alpha_p=\ba_p$ for all $p\in[\ell]$. 
It follows that, for every $p\in[\ell]$, $\psi$ deletes the edge
$$
	\be_p = (\ba_{p+1}, \bb_{p+1})=(\alpha_{p}, \beta_{p+1})=\tilde{e}_p,
$$
and adds the edge
$$
\tilde{\be}_p = (\ba_{p+1}, \bb_{p})=(\alpha_p,\beta_p)=e_p,
$$
so $\psi=\phi^{-1}$.

\end{proof}

\begin{proof}[Proof of Proposition~\ref{prop:i-iv}]
For (i), we first focus on the well definiteness. Let $x\in \widehat{\PPP}$. Since $x \in \widehat{\FFF}(1,1,1)$, $\pi_1(x)\in\widehat{\BB{1}}$ and $z=\phi_1^{-1}(x)$ is well defined. So it suffices to check that $\pi_2(z)\in \widehat{\BB{2}}$. We will do it in two steps:
\begin{itemize}
	\item[-] $\pi_2(z)\in \widehat{\AAbl}$. Let $C$ be a $w_2$-cyclic-thread in $z$. Then, either $C$ is the only $w_2$-cyclic-thread in $x$, which contains $(\bb^2_{\ell_2+1},0)$, or there exists $p\in [\ell]$ such that $\tilde{\be}^1_p$ is in $C$. As $C$ is a cycle, in the second case there exists $(q,s)\in[\ell]\times[[k]]$ such that $(\bb_q^1,s)$ is the first visit of a vertex among $\{\bb^1_{q'},q'\in [\ell]\}$ after $(\bb^1_p,r)$ in $C$. In particular $\Aarrows{\bb^1_p}{r}{w_2}{x}{\bb^1_{q+1}}{s}$,
	which contradicts $x \in \FFF(1,2,1)$.

\item[-] $\pi_2(z)\in \BB{2}$. Suppose that $\Aarrows{\bb_p^2}{r}{w_2}{y}{\bb_q^2}{s}$ for some $p,q\in [\ell_2+1]$ and $s\neq 0$. As $x\in \FFF(2,2,2)$, there exists $p'\in [\ell]$ such that $\tilde{e}^1_p$ is in the thread. Then, $\Aarrows{\bb_{p'}^1}{r'}{w_2}{x}{\bb_{q+1}^2}{s}$,
contradicting $x \in \FFF(1,2,2)$. 
\end{itemize}
	Since $\pi_2(z)\in \BB{2}$, by Proposition~\ref{prop:BP1} we have $\phi^{-1}_2 \circ \phi^{-1}_1(\widehat{\PPP}) \subset \CC{2}$. To show that $\lambda_2=\ell_2$ and $\beta_p^2=\bb_p^2$ for $p\in [\ell_2+1]$, it suffices to prove that the thread $T=[\thread{v_2}{0}{w_2,x}]$ in $x$ is not altered in $y$. If it is altered, then there exists $q\in [\ell]$ with $(\bb_{q+1}^1,s)\in T$. Choose the first one to appear, then $\Aarrows{\bb_{1}^2}{0}{w_2}{x}{\bb_{q+1}^1}{s}$,
contradicting $x\in \FFF(2,2,1)$. The argument for (ii) is analogous.

For (iii), observe that the argument used in (i)-(ii) implies that 
the threads
	$[\thread{v_1}{0}{w_1,x}]$
	and
	$[\thread{v_2}{0}{w_2,x}]$
	in $x$
	are not altered in $\phi^{-1}_2(x)$ and $\phi^{-1}_1(x)$, respectively.
	Therefore, which edges and how they are altered is independent on the order we apply the inverse bijections. 

	Finally, for any $z\notin \SSS$, the $w_1$-cycles and the $w_2$-cycles in $z$ are not altered in $\phi_2(x)$ and $\phi_1(x)$, respectively (the proof of this fact is similar to the previous cases). This implies that $\phi_1\circ\phi_2(z)=\phi_2\circ\phi_1(z)$.  The statement (iii) is just the contrapositive of that fact.
\end{proof}

\begin{proof}[Proof of Proposition~\ref{prop:set2nd}]
	We consider the mapping $$\Psi:   
 \widehat{\PPP}		\cup 
		 \widehat{\QQQ}
			\longrightarrow \AAll$$ acting as $\phi_2^{-1}\circ \phi_1^{-1}$ on $X:=\widehat{\PPP}$ and as $\phi_1^{-1}\circ \phi_2^{-1}$ on $Y:=\widehat{\QQQ}\setminus\widehat{\PPP}$. By Proposition~\ref{prop:i-iv} (i-iii) 
	$\Psi$ is well defined. Moreover, since $\phi_1^{-1}$ and $\phi_2^{-1}$ are bijections, if there exist $x\neq y$ such that $\Psi(x)=\Psi(y)$, then necessarily we have $x\in X$ and $y\in Y$ up to permutation of $x$ and $y$ (this follows from the commutation in Proposition~\ref{prop:i-iv} (iii) and the injectivity of both $\phi_2^{-1}\circ \phi_1^{-1}$ and $\phi_1^{-1}\circ \phi_2^{-1}$). This implies that $z=\Psi(x)=\Psi(y)$ has exactly two preimages in this case, and moreover Proposition~\ref{prop:i-iv} (iv) implies that $z\in \SSS$.

	We have shown that all elements of $\AAll$ have at most one preimage by $\Psi$, except at most $|\SSS|$ elements which have at most two preimages. We conclude that $|\widehat{\PPP}		\cup 
		 \widehat{\QQQ}| \leq |\AAll| + |\SSS|$, which is equivalent to~\eqref{eq:mainSetEqSecond}.
\end{proof}

\section{First and second moment asymptotics}
\label{sec:proofsAsymptotic}

In this section we prove all the asymptotic statements from Section~\ref{sec:backbone} that do not require the exploration techniques introduced in the next sections.

\begin{proof}[Proof of Proposition~\ref{prop:AP1}]
By Proposition~\ref{prop:BP1} and Lemma~\ref{lemma:AL1}
\begin{align*}
|\widehat{\BB{i}}|= |{\CC{i}}|=\mathbb{P}_{\AAl} ((A,\sigma_i)\in \CC{i})\cdot |\AAl |  \sim  |\AAl | =\frac{|\AAbl|}{n}
\end{align*}

\end{proof}

\begin{proof}[Proof of Lemma~\ref{lemma:AL7}]
	Since $\widehat{\BBB}= \widehat{\FFF}((1,1,1),(2,2,2))$, if $x\in \widehat{\BBB}\setminus (\widehat{\PPP}\cup \widehat{\QQQ})$, then $x\notin \widehat{\FFF}(I_1)$ and $x\notin \widehat{\FFF}(I_2)$ for some $I_1\in \cI_1$ and $I_2\in \cI_2$.  By Lemma~\ref{lemma:AL6}, it follows that
\begin{align*}
|\widehat{\BBB}\setminus ( \widehat{\PPP} \cup \widehat{\QQQ})| &\leq |\AAbbll | \sum_{I_1,I_2} 		\mathbb{P}_{\AAbbll}(
			(A,v_1,v_2,\sigma_1,\sigma_2) \not \in (\FFF(I_1)\cup \FFF(I_2)) \mbox{ and } (A,v_1,v_2)\in \widehat{\AAbb}   )\\
			&=o(|\AAll|).
\end{align*}

\end{proof}

\begin{proof}[Proof of Proposition~\ref{prop:AP2}]
It follows by combining Proposition~\ref{prop:set2nd}, Lemma~\ref{lemma:AL4} and Lemma~\ref{lemma:AL7}.
\end{proof}

\begin{proof}[Proof of Theorem~\ref{thm:momentEstimates}]

We let $\nota{a_k} :=|\mathcal{W}_k^{NC}|$ be the number of words of length $k$ which are not self-conjugate.

By linearity of expectation we have 
	$$\mathbb{E}_\mathcal{A} [ X_k(A) ] = n! \sum_{w\in \mathcal{W}^{NC}_k} n \cdot  \mathbb{P}_{\mathcal{A}^{\bullet\diamond}}((A,v,\sigma)\in \widehat{\mathcal{B}}_w^\bullet)
	= n!\sum_{w\in \mathcal{W}^{NC}_k}  \frac{n\cdot|\widehat{\mathcal{B}}_w^{\bullet\diamond}|}{ |\mathcal{A}^{\bullet\diamond}|}
=n!	a_k (1+o(1)).
$$
by Proposition~\ref{prop:AP1}. 

For the second moment, we have again by linearity of expectation that
	\begin{align*}
		\frac{1}{n!^2}	\mathbb{E}_\mathcal{A}& [ X_k(A)^2 ]= 
		\sum_{w_1,w_2\in \mathcal{W}^{NC}_k} n^2\cdot\mathbb{P}_{\AAbbll}((A,v_1,\sigma_1,v_2,\sigma_2)\in \widehat{\BBB})
		\\&=
		\sum_{w_1,w_2\in \mathcal{W}^{NC}_k\atop w_1 \equiv w_2} n^2\cdot\mathbb{P}_{\AAbbll}((A,v_1,\sigma_1,v_2,\sigma_2)\in \widehat{\BBB})	
		+
		\sum_{w_1,w_2\in \mathcal{W}^{NC}_k\atop w_1\not \equiv w_2} n^2\cdot\mathbb{P}_{\AAbbll}((A,v_1,\sigma_1,v_2,\sigma_2)\in \widehat{\BBB})
\\&\leq
		\sum_{w_1,w_2\in \mathcal{W}^{NC}_k\atop w_1 \equiv w_2} n^2\cdot\mathbb{P}_{\AAbl}((A,v_1,\sigma_1)\in \widehat{\mathcal{B}}_{1}^{\bullet\diamond})	
		+
		\sum_{w_1,w_2\in \mathcal{W}^{NC}_k\atop w_1\not \equiv w_2} n^2\cdot\frac{|\widehat{\BBB}|}{|\mathcal{A}^{\bullet\bullet}|}.
 	\end{align*}
	By Proposition~\ref{prop:AP1}, the first sum is bounded by $n^2/n=n$ times the number $b_k$ of pairs of words $(w_1,w_2)$ which are conjugate. Moreover, by Proposition~\ref{prop:AP2} the second sum is bounded by $(1+o(1))$ times the number of pairs of words $(w_1,w_2)$ which are \emph{not} conjugate, which equals $(a_k)^2-b_k$. Overall, the total contribution we obtain is bounded by 
	$$
	((a_k)^2 - b_k) (1+o(1)) + O(n b_k) = a_k^2 (1+o(1)) + O(n b_k). 
	$$
Finally we have $b_k \leq k2^k$ for clear combinatorial reasons (choose $w_1$, then choose the conjugation shift).
Overall, we have proved:
	\begin{align} %
		\frac{1}{n!}\mathbb{E}_\mathcal{A} [ X_k(A) ] \sim a_k 
		\ \ , \ \  
		\frac{1}{n!^2}	\mathbb{E}_\mathcal{A} [ X_k(A)^2 ] \leq a_k^2  (1+o(1)) + O(kn2^k).
\label{eq:firstsecondMoment}
\end{align}

Now, since a word of length $k$ is self-conjugate only if it is the $p$-th power of a word of length $k/p$ for some $p>1$, we have
$$
2^k-a_k = \left|	\{a,b\}^k \setminus (\mathcal{W}^{NC}_k) \right| \leq \sum_{p|k, p>1} 2^{k/p} \leq k 2^{k/2} = o(2^k),
$$
which shows that $a_k \sim  2^k$ when $k$ goes to infinity.

	Therefore~\eqref{eq:firstsecondMoment}
	and the fact that $k$ goes to infinity with $n$   imply
	\begin{align}
		\frac{1}{n!}\mathbb{E}_\mathcal{A} [ X_k(A) ] \sim 2^k \ , \  
		\frac{1}{n!^2}\mathbb{E}_\mathcal{A} [ X_k(A)^2 ] \leq 4^k (1+o(1)),
\end{align}
which implies the result.
\end{proof}

\begin{proof}[Proof of Corollary~\ref{cor:X>1}]
This a direct application of Theorem~\ref{thm:momentEstimates} and the second moment method.
\end{proof}
We note that Corollary~\ref{cor:X>1} implies Theorem~\ref{thm:structure}. The only problem a priori is the restriction $\epsilon<1$ which is not present in that theorem. However, note that if $A$ is a $w$-tree it is also a $w^i$-tree for any integer $i\geq 1$, so the implication is, in fact, direct. Another way to prove this theorem is to use the observation that all results of our paper in fact hold replacing the assumption that $k\leq 2\log(n)$ by $k\leq K\log(n)$ (for any fixed constant $K>0$) at the cost of adapting some absolute constants, 
so Corollary~\ref{cor:X>1} is in fact true without the assumption that $\epsilon<1$ (this argument gives Theorem~\ref{thm:structure} for a non-self-conjugate word $w$).

\begin{proof}[Proof of Theorem~\ref{thm:main}]
	Let $k=\lfloor (1+\epsilon) \log(n)\rfloor$. By Corollary~\ref{cor:X>1} there exists w.h.p. a word $w$ of length $k$ such that the one-letter automaton $A_w$ is a tree, and moreover by Lemma~\ref{lemma:boundedheight}, the height $H=H_n$ of this tree is at most $L\sqrt{n}$ w.h.p., for some constant $L>0$. Let $v_0$ be the cyclic vertex of $A_w$. In the one-letter automaton $A_w$, the word $a^H$ maps every vertex to $v_0$. Equivalently, in $A$, the word $w^H$ sends every vertex to $v_0$. Therefore $w^H$ is a synchronizing word, and it has length $|w^H|=H\cdot|w|\leq  L\cdot\lfloor (1+\epsilon) \log(n)\rfloor\sqrt{n}$.
\end{proof}

We conclude the section proving our result relating height and synchronization.
\begin{proof}[Proof of Theorem~\ref{thm:height_implies_syncro}]
Let $k=\lfloor (1+\epsilon) \log(n)\rfloor$. By Corollary~\ref{cor:X>1} there exists w.h.p. a word $w$ of length $k$, such that the one-letter automaton $A_w$ is a tree. 
By Theorem~\ref{thm:momentEstimates}, we can assume that $w$ is not self-conjugated.
As before, the word $w^{H_n}$ sends every vertex to some $v_0$ and thus is a synchronizing word of length $(1+\epsilon) H_n \log_2(n)$.
By the hypothesis of the theorem, the height $H_n$ of the a random $w$-tree is stochastically dominated by $X_n$, and the result follows. 
\end{proof}

\section{Exploration process}
\label{sec:exploration}

In order to obtain explicit bounds on the probabilities of certain events in $\AAA$ (e.g. being cycle-good, branch-good, $\dots$) we use an exploration process that we now describe. The process will explore the automaton, thread by thread, starting at prescribed vertices and congruences given by the input.
 We recall that we have fixed in Section~\ref{sec:collisionDefs} an integer $k$ and two words $w_1$ and $w_2$ of length $k$. In agreement with this, our exploration processes are defined below relatively to words of length $k$ only (a more general definition could be given without this restriction, but we will not need this).

Let $A\in \AAA$ be an automaton. An \nota{exploration} of $A$ is a sequence $\cE=( (x_t,y_t,z_t)_{t\geq 0}, (t_j)_{j\geq 0})$, with $x_t\in [n]$, $y_t\in [[k]]$, $z_t\in\{a,b\}^k$ and $0=t_0\leq t_1\leq t_2< \dots $, such that each triplet $(x,y,z)$ appears exactly once in $\cE$ and $x_{t} \stackrel{z_t^{y_{t}+1}}{\longrightarrow} x_{t+1}$ in $A$, for all $t\neq t_j-1$.
\footnote{When we write $z_t^{y_{t}}$, the index $y_t$ is considered modulo $k$, and   $z_t^{y_{t}}$ is the $(y_t \mod k)$-th letter of $z_t$. The same convention will be used everywhere when an integer indexes the position of a letter in a word of length $k$.
In the graphical picture we consider this letter as a label (in $\{a,b\}$) for the edge  $(x_{t-1},x_t)$ in $A$ which is taken by the exploration at time $t$. Two copies of that edge (with the two different labels) could appear in $A$.
}
We write $\cE_t=((x_i,y_i,z_i)_{0\leq i\leq t}, (t_j)_{t_j\leq t})$ for the exploration up to time $t$. We denote by $(x,y,z)\in \cE_t$ if there exists $s\leq t$ with $(x_i,y_i,z_i)=(x,y,z)$. By a slight abuse of notation, we write $x\in \cE_t$
(or $(x,y)\in \cE_t$) 
to denote that there exists $y,z$
(or $z$, respectively)
such that $(x,y,z)\in \cE_t$.
The times $t_i$ can be thought of as ``starting times'' at which the exploration of the automaton $A$ is relaunched from a new vertex and congruence (while in between two such times, the exploration follows a thread in the automaton $A$).

It will be useful to consider the exploration without the starting points, we denote $\cE^*_t=((x_i,y_i,z_i): 0\leq i\leq t, i\neq t_j \text{ for all }j\geq 0, (t_j)_{t_j\leq t})$.
The \nota{partial automaton revealed} by $\cE_t$ has vertex set $[n]$ and a directed edge $uv$ is present if there exists $0\leq i< t$ with $i\neq t_j-1$ for all $j\in [d]$ such that $x_{i}=u$ and $x_{i+1}=v$.

Often, we will be interested in a partial exploration of $A$ that reveals information about certain threads. For $d\in \mathbb{N}$, an \nota{input} of size $d$, $U=((u_j,r_j,m_j))_{j\in [d]}$, is a sequence of triplets with $u_j\in [n]$, $r_j\in [[k]]$ and $m_j\in \{a,b\}^k$. In practice, $m_j$ will always be equal to $w_1$ or $w_2$.
A \nota{partial exploration} of $A$ with input $U$ is defined as follows:

\vspace{0.3cm}

\noindent 
\begin{itemize}
\item[(1)] Let $t=0$ and $j=1$.  

\item[(2)] If $j=d+1$, stop the exploration.

\item[(3)] Let $(x_t,y_t,z_t)=(u_j,r_j,m_j)$. 

\item[(4)] Expose the edge $e$ out-going from $x_t$ with label $z^{y_t+1}_t$.
\item[(5)] Let $x_{t+1}$ be the head of $e$, $y_{t+1}= y_t+1$ and $z_{t+1}=z_t$.
 
\item[(6)] if $(x_{t+1},y_{t+1},z_{t+1})\in \cE_{t}$, set  $t=t+1$, $j=j+1$ and go to (2).

\item[(7)] set $t=t+1$ and go to (4).

\end{itemize}
\vspace{0.3cm}

In words, the partial exploration process reveals the $w_j$-threads of $(u_j,r_j)$, i.e. $[\thread{u_j}{r_j}{w_j}]$, sequentially, stopping at the smallest time $t$ whenever the threads are determined by $\cE_t$.

\begin{remark}\label{rem:lengthvstime}
We denote by $t_U$ the final time of a partial exploration with input $U$, i.e. $t_U=t_d$. Given an input $U$, the partial automaton revealed by $\cE_{t_U}$ does not depend on the order of the sequence $U$, but the partial exploration $\cE_{t_U}$ does. However, if the maximum length of a thread in $A$ is $t_{\max}$, then $t_{j}-t_{j-1}\leq t_{\max}$ for all $j\geq 1$, and $t_U\leq d t_{\max}$.
\end{remark}

\subsection{Some deterministic properties of threads, hits, and balls}

In this section we establish a number of properties of the partial exploration $(\cE_t)$ with an input $U$, which hold deterministically.
Recall that $\mathcal{W}_k^{NC}$ is the set of not self-conjugated words in $\{a,b\}^k$ (see Section~\ref{sec:collisionDefs}). Many of the results herein, will strongly rely on exploring $w$-threads for $w\in \mathcal{W}_k^{NC}$.

Our first observation states that the number of vertices visited at a given congruence is roughly the same.
\begin{claim}\label{claim:equi}
Given an input $U$ of size $d$, then for all $r,s\in [[k]]$, $w\in\mathcal{W}_k^{NC}$ and $t\geq 0$
\begin{align*}
	|\{u: (u,r,w)\in \cE_t\}| - |\{u: (u,s,w)\in \cE_t\}| \in \{-d,\dots,d\}.
\end{align*}
\end{claim}
\begin{proof}
Restricted to a given thread, the discrepancy between congruence classes is at most $1$. The result follows since the input had size $d$.
\end{proof}

The second observation states that two threads cannot overlap for $k$ consecutive steps.
\begin{claim}\label{claim:trajectory}
Let $u,v\in [n]$,  $r_1,r_2,s_1,s_2\in [[k]]$, $t\geq 0$ and $w_1,w_2\in\mathcal{W}_k^{NC}$ be non conjugated (here we also admit $w_1=w_2$ if $r_1\neq r_2$). Suppose that $[\arrows{u}{r_1}{w_1}{v}{s_1})]=((x_i,y_i,z_i))_{0
\leq  i \leq t}$ and $[\arrows{u}{r_2}{w_2}{v}{s_2})]=((x'_i,y'_i,z'_i))_{0
	\leq  i \leq t}$. If $t\geq k$ and $x_{i}x_{i+1}=x'_{i}x'_{i+1}$ as labeled\footnote{Here by \emph{labeled} we mean that not only the edge is considered but also the letter $0$ or $1$ which induced this transition.}
	edges for all $0\leq i< t$, then $r_1=r_2$ and $w_1=w_2$.
\end{claim}
\begin{proof}
If $r_1\neq r_2$, by cyclically shifting $w_1$ exactly $r_2-r_1$ positions we would obtain $w_2$, a contradiction with the fact that they are not conjugated.
\end{proof}

The following definition is used almost everywhere in the rest of the paper.
\begin{definition}\label{def:hits}
	Given an exploration $\cE$ we say that $t\geq 0$ is an \nota{exploring time} if the labelled edge $x_tx_{t+1}$ taken at time $t$ of the exploration has not been visited at any previous time of the exploration. Otherwise, we call $t$ a \nota{following time}. If $t$ is an exploring time, then we call it a \nota{hitting time} if $x_{t+1}=x_i$ for some $0\leq i\leq t$, and we call $x_{t+1}$ the \nota{hit vertex}. 
\end{definition}

For $u\in [n]$ and $t\geq 0$, we denote by $\nota{h_t(u)}$ the number of hitting times $0\leq i\leq t$ with hit vertex $x_{i+1}=u$, and let $\nota{h_t}=\sum_{u\in [n]} h_t(u)$ be the total number of hitting times in the process up to time $t$. Both quantities implicitly depend not only on $t$ but also on $\cE_t$.
If $H$ is the partial automaton revealed by $\cE_t$, we denote by $\nota{d^-_t(u)}$ and $\nota{d^+_t(u)}$ the in- and out-degree of $u$ in $H$. 

The following statements shows that the number of hits controls (deterministically) many other quantities of interest.

\begin{claim}\label{claim:degrees}
For any $t\geq 0$,
\begin{align}
 \sum_{v\in [n], d_t^-(v)\geq 2} d_t^-(v) \leq  2h_t \quad \text{ and } \sum_{v\in [n], d_t^+(v)\geq 2} d_t^+(v) \leq  2h_t
\end{align}
\end{claim}
\begin{proof} 
	The first inequality is immediate by induction on $t$, however we will also reprove it as a byproduct of the proof of the second one, which requires a bit more care.

The in-degree of $u$ in $\cE_t$ is precisely the number of exploring times $i$ with $x_{i+1}=u$ and $i+1\neq t_j$ for $j\geq 1$. 
If $d^-_t(u)\geq 1$, let $i_0$ be the smallest such time. Then, for all such times $i>i_0$, $x_{i+1}=x_{i_0+1}$ with $i+1\neq t_j$, and so $i$ is a hitting time.
It follows that
\begin{align}\label{SORE}
h_t(u)\geq 
\begin{cases}
d^-_t(u)-1 & \text{if } d_t^-(u)\geq 2\\
0  & \text{otherwise}.
\end{cases}
\end{align}
It follows that
\begin{align*}
h_t = \sum_{u\in [n]} h_t(u) \geq \sum_{d_t^-(u)\geq 2} (d^-_t(u)-1) \geq  \sum_{d_t^-(u)\geq 2} d^-_t(u)-h_t
\end{align*}
where in the last inequality we used that the number of vertices of in-degree at least two is at most the number of hit vertices.

Denote by $n_k^{-}$ (resp. $n_k^{+}$) the number of vertices of in-degree (resp. out-degree) equal to $k$ in $\cE_t$. Write $n_{\geq k}^{-} = \sum_{j\geq k} n_{j}^-$, and similarly for $n_{\geq k}^{+}$. 

Let $\ell_t$ be the number of vertices $u$ in $U$ that have been hit within time $t$. We can refine the relation between hitting times and in-degrees by notting that for such vertices, $h_t(u)= d^-_t(u)$.

By the last remark and also using~\eqref{SORE}
\begin{align*}
h_t & = \ell_t+\sum_{d^-_t(u)\geq 2} (d^-_t(u)-1)=\ell_t+n_0^-+\sum_{u\in [n]} (d^-_t(u)-1)= \ell_t+n_0^- + \sum_{u\in [n]} (d^+_t(u) - 1) \\
&= \ell_t+n_0^- - n_0^+ -n_{\geq 2}^+ + \sum_{d_t^+(u)\geq 2} d^+_t(u).
\end{align*}
We claim that $n_{\geq 2}^+\leq h_t-\ell_t+n_0^+-n_0^-$. If so, we obtain
\begin{align*}
\sum_{d_t^+(u)\geq 2} d^+_t(u)\leq 2(h_t-(\ell_t+n_0^-)+n_0^+).
\end{align*}
Observe that $u_1$, the first vertex whose thread was exposed, either has in-degree $0$ in $\cE_t$ or it has been hit before time $t$. Moreover, by construction of the exploration process, $n_0^+\leq 1$. Indeed, only $x_t$ might have out-degree $0$ in the partial automaton revealed by $\cE_t$. It follows that $n_0^+\leq \ell_t+n_0^-$, which concludes the proof. 

It remains to prove the claim. Using~\eqref{SORE} once more,
\begin{align*}
n_1^+ + 2n_{\geq 2}^+ \leq \sum_{u\in [n]} d^+_t(u) =  \sum_{u\in [n]} d^-_t(u) =  \sum_{d_t^-(u)\geq 1} (d^-_t(u)-1) +n_{\geq 1} = h_t -\ell_t + n_{\geq 1}^-.
\end{align*}
We conclude by subtracting $n_0^+ + n_1^+ + n_{\geq 2}^+ = n_0^- + n_{\geq 1}^-$ from the previous inequality and rearranging.
\end{proof}

For $\ell\in \mathbb{N}$, define the in-ball $\nota{B^-_t(u,\ell)}$ as the set of vertices $v$ such that $u$ is at distance at most $\ell$ from $v$ in $\cE_t$. Similarly, define the out-ball $\nota{B^+_t(u,\ell)}$ as the set of vertices $v$ at distance at most $\ell$ from $u$ in $\cE_t$.
We denote by $\nota{B_t(u,\ell)}=B^-_t(u,\ell)\cup B^+_t(u,\ell)$.

\begin{lemma}\label{lem:ball}
For any $u\in [n]$, $\ell\geq 1$ and $t\geq 0$,
\begin{align}
|B_t^-(u,\ell)|\leq 2\ell (h_t+1) \quad \text{ and } \quad |B_t^-(u,\ell)|\leq 2\ell (h_t+1).
\end{align}
\end{lemma}
\begin{proof}

The number of vertices in $B_t^-(u,\ell)$ is at most the number of vertices in a tree rooted at $u$ of height $\ell$ where all edges are directed towards the root, and the in-degree multiset of the vertices at distance less than $\ell$ is a subset of the multiset $\{d_t^-(v)\}_{v\in [n]}$. Any such tree can be constructed as follows:  (1) build a model tree, a rooted directed tree where all vertices except the root have in-degree at least $2$ and (2) subdivide each edge of the model tree with at most $\ell-1$ vertices.
The number of edges in any model tree is at most $1+\sum_{d_t^-(v)\geq 2} d_t^-(v)\leq 1+2h_t$, by Claim~\ref{claim:degrees}.
It follows that $|B^-_t(u,\ell)|\leq 1+\ell(2h_t+1)\leq 2\ell(h_t+1)$. The same argument implies that $|B^+_t(u,\ell)|\leq 2\ell(h_t+1)$.

\end{proof}

\begin{lemma}\label{lemma:few_special_vertices}
For any $t\geq 1$, there are at most $10k h_t^2$ vertices $u$ such that $B_t(u,k)$ is not a directed path.
\end{lemma}
\begin{proof}
If $B^-_t(u,k)$ is not a directed path, then it either is a cycle, it contains a vertex of in-degree at least $2$ or a vertex of out-degree at least $2$.

If it is a cycle, then it contains a vertex that was hit at some point. Since the cycle has length at most $2k$, there are at most $2k h_t$ such vertices $u$.

If there exists $v\in B^-_t(u,k)$ of in-degree at least $2$, then $v$ is a hit vertex and $u\in B^+_t(v,k)$. By Lemma~\ref{lem:ball}, there are at most $2k (h_t+1)h_t$ such vertices $u$. 

Otherwise, there exists $v\in B^+_t(u,k)$ of out-degree at least $2$. By Claim~\ref{claim:degrees}, there are at most $h_t$ vertices of out-degree at least $2$. By Lemma~\ref{lem:ball}, there are at most $2k( h_t+1)h_t$ such vertices $u$.

The total number of vertices $u$ whose ball is not a directed path is at most $2kh_t+4k(h_t+1)h_t= 2k(2h_t+3)h_t\leq 10kh_t^2$.

\end{proof}

\begin{lemma}\label{lemma:few_followed}
For any $d$ and any $1\leq t \leq t_d$, the number of following times before time $t$ is at most $f_t:=4dk^2h_t^2$.
\end{lemma}
\begin{proof}
By Lemma~\ref{claim:trajectory}, any two trajectories coincide in less than $k$ consecutive transitions. Therefore, if time $s\in [t]$ is a following time, the current vertex $u$ must have a vertex of out-degree at least $2$ in $B^+_t(u,k)$. By the argument used in Lemma~\ref{lemma:few_special_vertices}, there are at most $2k(h_t+1)h_t\leq 4kh_t^2$ such vertices. Moreover, each triplet vertex/congruence/word appears at most once in $\cE_t$, and there are at most $d$ words whose threads are explored. So the total number of following times is at most $4dk^2h_t^2$.
\end{proof}

\subsection{Hits in random exploration process}

We now turn to properties of the exploration process with input $U$ (with a fixed set $U$ of size $d$) when it is run on a uniformly random element $A$ of $\AAA$. Then, $(X_t,Y_t,Z_t)_{t\geq 0}$, $(T_j)_{j\in [d]}$ and $(H_t)_{t\geq 0}$ are random sequences. In particular, $(X_t,Y_t,Z_t)$ and $H_t$ are measurable with respect to $\cE_t$, and $T_j$ with respect to $\cE_{t_j}$.

\begin{quote}
	{\it
Throughout the rest of the paper we assume that  $\log(n)\leq k\leq 2\log{n}$ and that $d$ is at most a polylogarithmic function of $n$.}
\end{quote}

\begin{definition}[Typical event $E_{typ}$]\label{thm:hits_random}
Define \nota{$t_{max}:=5k\sqrt{n}$} and for all $d\geq 1$, \nota{$h_{max}(d):=100(d k)^2$}. 
We define the event 
$\nota{E_{typ}}=E_{len}\cap E_{hit}\cap E_{ball}\cap E_{path}\cap E_{foll},$
	where
\begin{itemize}
\item[-] \nota{$E_{len}$}: for all input $U$ with $|U|=1$, we have $T_{U}\leq t_{\max}$;
\item[-] \nota{$E_{hit}$}: for all $d\geq 1$ and  input $U$ with $|U|=d$, we have $H_{T_U}\leq h_{\max}(d)$;
\item[-] \nota{$E_{ball}$}: for all $d\geq 1$ and  input $U$ with $|U|=d$, for all $u\in [n]$ and $\ell\geq 1$, we have $B_{T_U}(u,\ell)\leq 4\ell(h_{\max}(d)+1)$;
\item[-] \nota{$E_{path}$}: for all $d\geq 1$ and  input $U$ with $|U|=d$, there are at most $10k h_{\max}(d)^2$ vertices $u\in [n]$ such that $B_{T_U}(u,k)$ is not a 
		path;
\item[-] \nota{$E_{foll}$}: for all $d\geq 1$ and  input $U$ with $|U|=d$, all $\ell,t\geq 1$ with $t+\ell\leq T_U$, if all times from $s\in [t,t+\ell)$ are following times, then $\ell \leq 4dk^2h_{\max}(d)^2$.
\end{itemize}
\end{definition}

In words, $E_{len}$ implies that any thread of any vertex is short; $E_{hit}$ implies that the total number of hits is small; $E_{ball}$ implies that all balls have linear growth; $E_{path}$ implies that most of the vertices are in the center of a path of length $2k$; and $E_{foll}$ bounds the length of followed threads. In particular, by $E_{foll}$ there is no $w$-thread of $(u,r)$ with $(u,r,w)\notin \cE_{t}$ of length at least $4dk^2h_{\max}(d)^2$ that is determined in $\cE_{t}$, for $t\leq T_U$. 

Note that apart from $t_{\max}$ all the other bounds are polynomial on $d,k,\ell$. For applications the actual polynomial will be irrelevant. We will always assume $d$ and $\ell$ are at most logarithmic and just use the bound $\tO(1)$ for these quantities.

The following proposition, proved in the rest of this section, justifies the choice of the adjective ``typical''. As the proofs will show, the exponent $-2$ could be decreased at the only cost of adapting the absolute constants in the definitions of $t_{max}$ and $h_{max}$.
\begin{proposition}[Typical properties]\label{prop:Etyp}
 $E_{typ}$ fails with probability $o(n^{-3})$.
\end{proposition}

Before bounding the probability of events defined above, we focus on the number of hits in a single thread.
\begin{lemma}\label{lem:hits_distr}
For any input $U=\{(u,r,w)\}$, we have 
\begin{align*}
H_{T_U}\stackrel{st}{\leq} k^2 +H, 
\end{align*}
where $H\sim \Geom\Big(\frac{1}{k}\big(1-\frac{1}{k}\big)\Big)$ and $\stackrel{st}{\leq}$ denotes stochastic domination.
In particular, the probability that $H_{T_U} > 8k^2$ for some $U$ with $|U|=1$ is at most $o(n^{-3})$.
\end{lemma}
\begin{proof}
Write
\begin{align}
H_{T_U}=L_0+L_1
\end{align}
where $L_0$ is the number of hitting times in the first $m=k^2$ steps and $L_1=H_{T_{U}}-H_{m}$. 

Clearly, $L_0\leq m$. By Claim~\ref{claim:equi} and if $t\geq m$ is a hitting time, the probability that $(x_{t+1},y_{t+1},z_{t+1})\in \cE_t$ (which is a sufficient condition to end the current thread) is at least 
\begin{align*}
\frac{1}{k}-\frac{1}{t}\geq \frac{1}{k}\Big(1-\frac{1}{k} \Big),
\end{align*}
uniformly on $\cE_t$. It follows that $L_1$ is stochastically dominated by a Geometric random variable with that probability. Therefore, we have
$$
\Pr(H_{T_U}> 8k^2)\leq \Pr(H> 7k^2) = \left(1-\frac{1}{k}\Big(1-\frac{1}{k} \Big)\right)^{7k^2}\leq e^{-7k(1-1/k)}.
$$
There are at most $kn2^k\leq e^{3k}$ choices for $U$.
	By a  union bound over them, we get that the desired probability is at most $e^{-4k+7}=o(n^{-3})$.
\end{proof}

\begin{lemma}\label{lem:Elen}
 $E_{len}$ fails with probability $o(n^{-3})$.
\end{lemma}
\begin{proof}
Fix $u\in [n]$, $r\in [[k]]$, $w\in \mathcal{W}_k^{NC}$ and let $U=((u,r,w))$. By Lemma~\ref{lem:hits_distr}, it suffices to bound the probability that $E_{len}$ does not hold under $H_{T_U}\leq 8k^2$. Therefore, we can write 
$$
T_{U}\leq 2k+F_{T_U}+S,
$$ where $F_{T_U}$ is the number of following times up to time $T_U$ and $S$ is the number of exploring times $t$ with $t\geq 2k$. By Lemma~\ref{lemma:few_followed}, $F_{T_U}\leq c k^6$, for some constant $c>0$.

Let $t_0=2k+ck^6$ and observe that $t_0=o(k \sqrt{n})$, as $n$ goes to infinity. By Claim~\ref{claim:equi}, the probability that the thread finishes at time $t\geq 2k$ is at least $\frac{t/k-1}{n}\geq \frac{t}{2kn}$. It follows that 
$$
\Pr(L\geq t_{max}-t_{0}) \leq \prod_{t=t_0}^{t_{max}} \left(1-\frac{t}{2kn}\right)\leq e^{-\frac{(t_{max})^2-t_0^2}{4kn}}=o(e^{-6k}).
$$
There are at most $kn2^k\leq e^{3k}$ choices for $U$. 
The lemma follows by a union bound on them.
\end{proof}

\begin{remark}\label{rem:proofImpliesProof}
Observe that $E_{len}$ implies the statement of Lemma~\ref{lemma:boundedheight}. Indeed, for any $u\in [n]$ and $w\in \mathcal{W}_k^{NC}$, if $[\thread{u}{0}{w}]$ has length at most $t_{\max}$, the height of $u$ in the $A_w$ automaton is at most $t_{\max}/k = 5\sqrt{n}$. Therefore, Lemma~\ref{lemma:boundedheight} is a corollary of Lemma~\ref{lem:Elen}.
\end{remark}

\begin{lemma}\label{lem:Ehit}
$E_{hit}$ fails with probability $o(n^{-3})$.
\end{lemma}
\begin{proof}
By Lemmas~\ref{lem:hits_distr} and~\ref{lem:Elen}, we may assume that $E_{len}$ holds and that for any $U$ of size $1$ we have $H_{T_U}\leq 8k^2$.

Suppose that $|U|=d\geq 1$. Let $U_d=\{(u_d,r_d,m_d)\}$ and $U'=U\setminus U_d$. We can write 
$$
H_{T_U}=H_{T_{U'}}+H_{int}+H_{ext}
$$
where $H_{int}$ and $H_{ext}$ are the number of times $t\in (T_{U'},T_U]$ with $x_{t+1}\in \cE_t\setminus \cE_{T_{U'}}$ and $x_{t+1}\in \cE_{T_{U'}}$, respectively.

On the one hand, $H_{int}\leq H_{T_{U_d}}\leq 8k^2$. On the other hand $H_{ext}$ is stochastically dominated by a binomial random variable with parameters $t_{\max}$ and $(d-1)t_{\max}/n$. Indeed, by $E_{len}$ each thread has length at most $t_{\max}$ and at time $t$ there are at most $t$ vertices in $\cE_t$. If $m:=92dk^2$,
$$
\Pr(H_{ext}\geq m) \leq \binom{t_{\max}}{m} \left(\frac{(d-1)t_{\max}}{n}\right)^m \leq \left(\frac{e(d-1) (t_{\max})^2}{m n}\right)^m = \left(\frac{25e}{92}\right)^m = o(e^{-7dk}),
$$
where we used $\binom{t}{m}\leq \big(\frac{et}{m}\big)^m$.

By a union bound over the $d$ elements in $U$, with probability $o(e^{-6dk})$, we have
$$
H_{T_U}\leq 8dk^2+ 92(d k)^2\leq h_{\max}(d).
$$
Moreover, the number of choices for $U$ of size $d$ is $(kn2^k)^d\leq e^{3dk}$. By a union bound over all $d\geq 1$ and all inputs $U$ of size $d$, the probability that there are more than $h_{\max}(d)$ hits for some $d$ and some input $U$ of size $d$ is $o(e^{-3k})$.
\end{proof}

We now show that $E_{ball}\cap E_{path}\cap E_{foll}$ are a consequence of $E_{hit}\cap E_{len}$.

\begin{lemma}\label{lemma:smallballs}
	$E_{hit}\cap E_{len}$ implies $E_{ball}$.
\end{lemma}
\begin{proof}
This follows from combining Lemma~\ref{lem:ball} with the bound on $h_{\max}(d)$ under $E_{hit}\cap E_{len}$.
\end{proof}

\begin{lemma}\label{lemma:few_exepcional}
	$E_{hit}\cap E_{len}$ implies $E_{path}$.	
\end{lemma}
\begin{proof}
This follows from combining Lemma~\ref{lemma:few_special_vertices} with the bound on $h_{\max}(d)$ under $E_{hit}\cap E_{len}$.
\end{proof}

\begin{lemma}\label{lemma:followshort}
$E_{hit}\cap E_{len}$ implies $E_{foll}$.	
\end{lemma}
\begin{proof}
This follows from combining Lemma~\ref{lemma:few_followed} with the bound on $h_{\max}(d)$ under $E_{hit}\cap E_{len}$.
\end{proof}

Note that Proposition~\ref{prop:Etyp} follows from the preceding lemmas.

\subsection{Probability of connecting to a target}

In this section we give some helpful results on the probability that certain paths between designated vertices are revealed during the exploration process, and have certain properties. One can think of the following statements as the ``toolbox'' we will need to prove Lemmas~\ref{lemma:AL1}-\ref{lemma:AL4}-\ref{lemma:AL5}-\ref{lemma:AL6} in the next sections. 
Through the rest of the paper we assume that $d=O(1)$. %

Given two subsets $F$ and $F'$ of labeled edges we say that is \nota{$F$ is determined by $F'$} if $F$ is included in $F'$. We will repeatedly use this notion where $F$ and $F'$ will be all the edges appearing in a given thread or a given exploration, and when we use it we will identify the thread or the exploration with the edges it contains. For example, we say that \nota{a thread $[\thread{u}{r}{w}]$ is determined by a partial exploration $\cE_t$} if all the labelled edges needed to explore $[\thread{u}{r}{w}]$ have been  explored at least once in the exploration $\cE$ up to time $t$.
We will also say, for example, that \nota{$[\arrows{u}{r}{w}{v}{s}]$ is not determined} by $F'$ if either the property $\arrows{u}{r}{w}{v}{s}$  is not true in the underlying (full) automaton $A$, or if it is true but the partial thread $[\arrows{u}{r}{w}{v}{s}]$ uses at least one labeled edge not in $F'$ (here again, $F'$ will often be a thread or an exploration, identifying these objects with the labeled edges they contain).

The next result will be crucial in all our arguments. It provides the probability that at a given time the exploration process is at a given target vertex.

\begin{lemma}[Probability of a path]\label{lemma:findpath2}
For any input $U$ of size $d$, we have the following under $E_{typ}$.
	Let $0\leq t\leq T_U$ and let $(x,i),(y,j)\in [n]\times [[k]]$ and $w\in\mathcal{W}_k^{NC}$. Suppose $[\arrows{x}{i}{w}{y}{j}]$ is not determined by $\cE_t$ and $(y,j,w)\notin \cE_t^*$. 

Conditional on $\cE_{t}$, for any $t'\geq 0$ with $t+t'\leq T_U$, the probability that $\arrows{x}{i}{w}{y}{j}$ with a path of length $t'$ is $\tO(1/n)$. In particular, the probability that $\arrows{x}{i}{w}{y}{j}$ is $\tO\left(n^{-1/2}\right)$.

\end{lemma}
\begin{proof}
	We stop the exploration $\cE$ at time $t$, and we start exploring the thread $[\thread{x}{i}{w}].$
	As  $[\arrows{x}{i}{w}{y}{j}]$ is not determined in $\cE_t$, either this path does not exist in $A$ or 
	there exists at least one exploring time between times $t$ and $t+t'$; let $s\in [t,t+t')$ be the latest exploring time, which is a hitting time
or the first visit at the vertex $y$.
By $E_{typ}$, there are at most $\ell=\tO(1)$ followed edges up to time $t+t'$. So $s\geq t+t'-\ell$ and $x_s\in B^{-}_{s}(y,\ell)$ since no other edges are revealed after time $s$ and before time $t+t'$. By $E_{typ}$ and a union bound over the choice of $s$, the desired probability is at most
\begin{align*}
\sum_{s=t+t'-\ell}^{t+t'}\frac{|B^{-}_{s}(y,\ell)|}{n} = \tO\left(n^{-1}\right).
\end{align*}
	The last statement follows by definition of $E_{typ}$ and union bound on $t'\leq t_{max}=\tilde{O}(n^{1/2})$.
\end{proof}

We next exploit the fact that, if a vertex (called $y$ below) has not played any special role inside the exploration, we can consider it to be random by reshuffling the labels of other vertices.
In the following lemma, there is a certain subset $Z$ of vertices which we can think of as the vertices which have "already been named" and which are not allowed to take part in the reshuffling, however the conclusion still holds.

\begin{lemma}[Probability of a path with new source]\label{lemma:P1}
For any input $U$ of fixed size $d$, we have the following under $E_{typ}$.
Let $t,t'\geq 0$ with $t+t'\leq T_U$. Let $(x,i),(y,j)\in [n]\times [[k]]$ with $x\neq y$, and $w\in\mathcal{W}_k^{NC}$.

Fix $H$ a partial automaton in $[n]$ and $Z\subseteq [n]\setminus\{x\}$ of size at most $d$. 
Let $\mathfrak{E}_H(Z,y)$ be the set of explorations $\cE$ such that 
\begin{itemize}
\item[i)] there is an isomorphism from 
	the partial automaton revealed by $\cE$
		to $H$ that fixes each element of $Z\cup\{y\}$;
\item[ii)] $(y,j,w)\notin \cE^*$.
\end{itemize}
Conditional on $\cE_{t}\in \mathfrak{E}_H(Z,y)$, the probability that $\arrows{x}{i}{w}{y}{j}$ with a path of length $t'$  is $\tO(n^{-1})$.
\end{lemma}
\begin{proof}
The desired probability is invariant with respect to any relabeling $\sigma$ of $[n]$ that fixes each element in $Z\cup\{y\}$. By abusing notation, we will denote $\sigma(\cE_t)$ the relabeled partial automaton. 

Let us compute the probability that $\arrows{x}{i}{w,\cE_t}{y}{j}$ given that $\cE_t\in \mathfrak{E}_H(Z,y)$, which by the previous observation is equal to $\arrows{\sigma(x)}{i}{w,\sigma(\cE_t)}{\sigma(y)}{j}=(y,j)$, for a random $\sigma$.
We may assume that the path has length $t'$. Then $\sigma(x)\in B^-_t(y,t')$ in $\sigma(\cE_t)$. Since $(\sigma(y),j,w)\notin \sigma(\cE_t^*)$, by $E_{typ}$ we have $t'=\tO(1)$. By $E_{typ}$, we also have $|B^-_t(y,t')|=\tO(1)$.
	Since $x\not \in \{y\}\cup Z$ and since $\sigma$ is random, it follows that the probability that $\arrows{\sigma(x)}{i}{w,\sigma(\cE_t)}{y}{j}$ is $\tO(n^{-1})$.

Now, for any $\cE_t$ not satisfying $\arrows{x}{i}{w,\cE_t}{y}{j}$, let us compute the probability that $\arrows{x}{i}{w,A}{y}{j}$ with a path of length $t'$. Since the path is not determined in $\cE_t$ and $(y,j,w)\notin \cE_t^*$, we can use Lemma~\ref{lemma:findpath2} to deduce that this probability is also $\tO(n^{-1})$.
\end{proof}

The previous lemma can be strengthened, as in many applications, the target of such path will be a lower-record for some permutation.

\begin{lemma}[When the target is a lower-record]\label{lemma:P1_enhanced}
	For any input $U$ of fixed size $d$, we have the following under $E_{typ}$. Let $t\geq 0$ with $t< T_U$, $x,y\in [n]$ with $x\notin U\cup\{y\}$, $w\in\mathcal{W}_k^{NC}$ and $\sigma$ a permutation of length $n$. Suppose that $(y,0,w)\notin \cE^*_{t}$. Define $\mathfrak{E}_H(U,y)$ as in Lemma~\ref{lemma:P1}.

Conditional on $\cE_{t}\in \mathfrak{E}_H(U,y)$, the probability that $\arrows{x}{0}{w}{y}{0}$ and that $(y,0)$ is a $\sigma$-lower-record at congruence $0$ on $[\thread{x}{0}{w}]$  is $\tO(n^{-1})$.
\end{lemma}
\begin{proof}
	Let $1\leq t'\leq t_{\max}$. We first compute the probability that $\arrows{x}{0}{w}{y}{0}$ and $[\arrows{x}{0}{w}{y}{0}]$ has length $t'$. Observe that we are in the setting of Lemma~\ref{lemma:P1} with $U=Z$. Therefore, such probability is $\tO(n^{-1})$. On the other hand, if the path has length $t'$, the probability it is a lower record is $k/t'=\tilde{O}(1/t')$, and is independent from $\cE_{t+t'}$, 
	so we obtain the bound $\tilde{O}(n^{-1}/t')$.

	We can then use a union bound over the length $t'$ which under $E_{typ}$ is at most $t_{max}$, so the desired probability is 
$$
\sum_{t'=1}^{t_{\max}} \frac{1}{t'} \tO(n^{-1}) = \tO(n^{-1}).
$$
\end{proof}

By $E_{typ}$, most of the vertices in the exploration process have a local neighbourhood which is simply a path. We exploit this fact for the source, to give an enhanced version of the probability of path existence. This lemma will be very useful when exposing two paths at once.

\begin{lemma}[Probability of a path with in-degree at least two among the first steps]\label{lemma:findpath4}
For any input $U$ of fixed size $d$, we have the following under $E_{typ}$.
Let $t,t'\geq 0$ with $t+t'\leq T_U$, let $(x,i),(y,j)\in [n]\times [[k]]$ with $j -i \equiv t'$ and $r\in [[k]]$ and let $w\in\mathcal{W}_k^{NC}$.

Suppose that
\begin{itemize}
\item[i)] $x\notin \cE_{t}$;
\item[ii)] $(y,j,w)\notin \cE^*_{t}$.
\end{itemize}
		Stop the exploration $\cE$ at time $t$ and restart it at time $t+1$ by exploring the thread $[\thread{x}{i}{w}]$.
	Then, conditional on $\cE_{t}$, the exploration $\cE_{t+t'}$ satisfies the following with probability $\tO(n^{-3/2})$:
\begin{itemize}
	\item[a)] $\arrows{x}{i}{w}{y}{j}$, and $[\arrows{x}{i}{w}{y}{j}]$ has length $t'$;
\item[b)] $y$ only appears at the end of $[\arrows{x}{i}{w}{y}{j}]$;
\item[c)] there is a vertex of in-degree at least $2$ different from $y$ in the first $k$ steps of $[\arrows{x}{i}{w}{y}{j}]$.
\end{itemize}
\end{lemma}

\begin{proof}
	Let $P=[\arrows{x}{i}{w}{y}{j}]$. Let $z$ be the first vertex in $P$ that appears at least twice in $\cE_{t+t'}$ (it exists by c)). Note that $z\neq y$ by b). %

We now define two times, $t_1$ and $t_2$, as follows. We split into two cases:
\begin{itemize}
	\item[-] If $z\notin \cE_{t}$, then we let $t_1\in [t,t+k)$ be the time of the first visit to $z$, so $z$ is at distance at most $k$ from $x$ at time $t_1$. Let $t_2\in [t_1+1,t+t')$ be the time of the second visit to $z$. Then, at time $t_2$ there is a hit at a vertex in $B^+_{t_1}(x,k)$. By $E_{ball}\subset E_{typ}$, there are $\tO(1)$ choices for $z$, and by $E_{len}\subset E_{typ}$, there are $\tO(\sqrt{n})$ choices for $t_2$.
	\item[-] If $z\in\cE_{t}$, then let $t_1=t$ and let $t_2\geq t$ be the first time after $t$ that the exploration is in $z$. By i), we have that  $z\neq x$ and $t_2\geq t+1$. We claim that $t_2$ is a hitting time. Indeed it suffices to see that it is an exploring time, as $z\in \cE_t$. Suppose it is a following time and let $t_2'\geq t$ be the last exploring time before $t_2$ (by i), it exists). If $z'=x_{t_2'+1}$ is the target vertex at time $t_2'$ then $z'\in \cE_{t_2}$ and it has degree at least $2$ in $\cE_{t+t'}$. As $t_2'<t_2$, $z'$ appears in $P$ previous to $z$, which is a contradiction with the choice of $z$.	
		Moreover, by the choice of $z$ and c), we have $t_2< t+k$ and there are $\tilde{O}(1)$ choices for $t_2$. By $E_{len}\subset E_{typ}$, there are $dt_{\max}=\tO(\sqrt{n})$ candidates for $z$.
\end{itemize}
Let $i_1,i_2\in [[k]]$ the congruence classes of $t_1$ and $t_2$. Note that $\arrows{z}{i_1}{w}{z}{i_2}$ is not determined by $\cE_{t_1}$ (since $t_2$ is a hitting time) and clearly $(z,i_2,w)\notin \cE_{t_1}$ as such a triplet will be visited at time $t_2$. Conditional on $\cE_{t_1}$, we can apply Lemma~\ref{lemma:findpath2} to deduce that the probability of hitting $z$ at time $t_2$ is $\tO(n^{-1})$. By a union bound over the choices of $z$ and $t_2$, regardless of whether $z$ appeared in $\cE_t$ or not, the contribution of this hit is $\tO(n^{-1/2})$.

As $z\neq y$, we have $t_2\neq t+t'$. Also observe that $(y,j,w)\notin \cE_{t_2}^*$ as $y$ has not been visited after time $t$ by b).
Conditional on $\cE_{t_2}$ and by Lemma~\ref{lemma:findpath2}, the probability of $\arrows{z}{i_2}{w}{y}{j}$ with a path of length $t'-t_2$ is $\tO(n^{-1})$. Therefore the total probability is $\tO(n^{-3/2})$ as desired.
\end{proof}

\subsection{Probability of short cycles}

The next two lemmas complete our toolbox and control the probability that threads end with a short cycle.

\begin{lemma}[Probability of a thread ending in a short cycle]\label{lemma:thread_short_cycle}
For any input $U$ of size $d$, we have the following under $E_{typ}$.
	Let $t\geq 0$ with $t < T_U$, $(x,i)\in [n]\times [[k]]$ and $w \in \mathcal{W}_k^{NC}$. Suppose that 
\begin{itemize}
\item[i)] $[\thread{x}{i}{w}]$ is not determined by $\cE_t$ and $(x,i,w)\notin \cE_t^*$;
	
\item[ii)] there is no $y\in [n]$ such that $[\thread{y}{0}{w}]$ is determined by $\cE_t$ and has length $k$.
\end{itemize}
Conditional on $\cE_{t}$, the probability that $\cyc{w}{x,i}=1$ is $\tO\left(n^{-1/2}\right)$. Moreover, conditional on $\cE_{t}$, the probability that $[\thread{x}{i}{w}]$ has length $k$ is $\tO\left(n^{-1}\right)$.
\end{lemma}
\begin{proof}
	Let $t'\in [1,t_{\max})$ with $t+t'\leq T_U$. We will compute the probability of the wanted event, with the extra condition that $[\thread{x}{i}{w}]$ has length $t'$. If the event $\cyc{w}{x,i}=1$ holds, we let $\Ccyc{w}{x,i}$ be the sequence of edges in $A$ participating to the corresponding cycle of $A_w$ (so $\Ccyc{w}{x,i}$ is a closed path of length $k$ in $A$).

By condition ii) there is at least one edge in $\Ccyc{w}{x,i}$ not revealed in $\cE_t$. Let $s\geq t$ be the time when the last edge of $\Ccyc{w}{x,i}$ is revealed
		and $x_s\rightarrow x_{s+1}$ the edge traversed at that time.
		Then $x_{s+1}\in B^-_s(x_s,k)$. We will argue differently depending on whether $s$ is the last edge revealed or not. If it is, then $s\in [t+t'-\tilde{O}(1),t+t']$
		since under $E_{foll}\subset E_{typ}$ the number of following times is at most $\tilde{O}(1)$.
Using Lemma~\ref{lemma:findpath2} and by a union bound over $s$ and over the elements of $B^{-}_{s}(x_s,k)$ (at most $\tO(1)$ by $E_{ball}\subset E_{typ}$), the desired probability is at most
\begin{align*}
\sum_{s=t+t'-\tO(1)}^{t+t'}\frac{|B^{-}_{s}(x_s,k)|}{n} = \tO\big(1/n\big).
\end{align*}
		If $s$ is not the last time at which an edge is revealed, then we apply Lemma~\ref{lemma:findpath2} and by a union bound over $s\in [t,t+t']$ and over the elements of $B^{-}_{s}(x_s,k)$ we get a probability $\tO(n^{-1/2})$ for such edge. As $[\thread{x}{i}{w}]$ is still not revealed at time $s$ (and moreover $(x,i,w)\notin \cE_s^*)$, we can apply once more Lemma~\ref{lemma:findpath2} with $t'$ being $t'-s$ and $z$ any vertex in $[\thread{x}{i}{w}]$.
		Using union bound over $z$ (at most $t_{max}$ choices under $E_{len}\subset E_{typ}$) we get an extra contribution of $\tilde{O}(n^{-1/2})$ to the probability, hence an upper bound of $\tilde{O}(n^{-1})$ in total.

The first statement follows by a union bound over $t'\leq t_{\max}$. The second statement follows from fixing $t'=k$.
\end{proof}

\begin{lemma}[Probability of two threads ending in short cycles]\label{lemma:thread_short_cycle_double}
For any input $U$ of size $d$, we have the following under $E_{typ}$. Let $t\geq 0$ with $t < T_U$,  $x_1,x_2\in [n]$, $i_1,i_2\in [[k]]$ and let $w_1,w_2\in\mathcal{W}_k^{NC}$ be non-conjugated words.
Suppose that 
\begin{itemize}
\item[i)] for any $j\in [2]$, $[\thread{x_j}{i_j}{w_j}]$ is not determined by $\cE_t$;

\item[ii)] there is no $y\in [n]$ and $j\in [2]$ with $(y,0,w_j)\in \cE_t$, such that $[\thread{y}{0}{w_j}]$ is determined by $\cE_t$ and has length $k$.	 

\end{itemize}
Conditional on $\cE_{t}$, the probability that $\cyc{w_1}{x_1,i_1}=\cyc{w_2}{x_2,i_2}=1$ is $\tO\left(1/n\right)$.  Moreover, conditional on $\cE_{t}$, the probability that $\cyc{w_1}{x_1,i_1}=1$ and $[\thread{x_2}{i_2}{w_2}]$ has length $k$ is $\tO\left(n^{-3/2}\right)$.
\end{lemma}
\begin{proof}
Let $P=[\thread{x_2}{i_2}{w_2}]$ and $Q=[\thread{x_1}{i_1}{w_1}]$. Expose $P$ using the first part of Lemma~\ref{lemma:thread_short_cycle} with $x=x_2$, $i=i_2$ and $w=w_2$. Let $t'$ be the time at the end of the exploration of $P$.  We split into two cases.

If exposing $P$ does not determine $Q$, then, as all triplets exposed between time $t$ and $t'$ have the word $w_2$, we can apply again the first part of Lemma~\ref{lemma:thread_short_cycle} with $x=x_1$, $i=i_1$ and $w=w_1$. The desired probability is $\tO(1/n)$.

If exposing $P$ determines $Q$, then there are at least two cycles of length $k$ in $\cE_{t'}$, since otherwise, $w_1$ and $w_2$ would be conjugated words. Let $t_1,t_2\in [t,t')$ be the times where the last edges of the each of the two cycles were exposed. Note that, for $j\in [2]$, $x_{t_j+1}\in B^-_{t_j}(x_{t_j},k)$. Moreover, by $E_{ball}\subset E_{typ}$, these balls have size $\tO(1)$. 

If $t_1\neq t_2$, then, by a union bound over the choices of $t_1,t_2\in [0,t_{\max})$ a union bound of the vertices of $B^-_{t_j}(x_{t_j},k)$, and two applications of Lemma~\ref{lemma:findpath2}, the probability of the whole event is $\tO(1/n)$.

If $t_1=t_2$, then the two cycles share at least one edge. In particular, the two cycles span a connected subgraph $H$ with at most $2k$ vertices. Let $t_0$ be the first time that a cycle from $H$ is exposed. Then the same argument as before with $t_0$ and $t_1$ (taking  into account that $x_{t_0+1}\in B^-_{t_0}(x_{t_0},2k)$) gives the desired result.

To prove the second part, just observe that when applying Lemma~\ref{lemma:thread_short_cycle} to reveal $P$, we can use the second part of it, obtaining an additional $\tO(n^{-1/2})$ factor. In case $P$ does not determine $Q$, this yields to the final $\tO(n^{-3/2})$. In case it does, now $t'=k$ and since $t_0,t_1,t_2\leq t'$ the union bound over them only adds an additional factor $\tO(1)$, giving a stronger bound $\tO(n^{-2})$.
\end{proof}

	We conclude our toolbox with two last lemmas. Roughly speaking, they strengthen the bounds of the previous lemmas (Lemma~\ref{lemma:thread_short_cycle},~\ref{lemma:findpath2}) 
		by an extra factor of $\tilde{O}(n^{-1/2})$ if one requires that an additional $w_2$-cycle of length $1$ is created by the thread to be revealed, for a second word $w_2$.
\begin{lemma}[Thread ending in a short cycle, strenghtened version]\label{lemma:1st_ad_hoc}
For any input $U$ of size $d$, we have the following under $E_{typ}$. Let   $x\in [n]$, $i\in [[k]]$ and let $w_1,w_2\in\mathcal{W}_k^{NC}$ be non-conjugated words.

The probability that $\cyc{w_1}{x,i}=1$ and that there exists $z\in [n]$ such that $[\thread{z}{0}{w_2}]$ is determined by $[\thread{x}{i}{w_1}]$ and has length $k$, is $\tO\left(1/n\right)$.
\end{lemma}
\begin{proof}
Suppose the event holds. Let $P=[\thread{x}{i}{w_1}]$ and $P'=[\thread{z}{0}{w_2}]$, for $z$ the vertex for which it holds. Let $t_1$ the hitting time that fully reveals the cycle, that will be at the end of $P$ (note that this might not be the last hitting time). Let $t_2$ be the hitting time that reveals $P'$ while exploring $P$. 

If $t_2\neq t_1$, observe that $x_{t_2+1}\in B^-_{t_2}(x_{t_2},k)$ and $x_{t_1+1}\in B^-_{t_1}(x_{t_1},k)$. By $E_{typ}$, both balls have size $\tO(1)$. By a union bound over the values of $t_2$ and $t_1$ the probability of such event is $\tO(1/n)$ and we are done.
If $t_2= t_1$, we do the same trick as in the proof of Lemma~\ref{lemma:thread_short_cycle_double}. Since the two cycles share an edge, they form a subgraph of size at most $2k$ that contains at least two cycles. Define $t_3< t_1$ as the first time we close a cycle in $H$. The same argument as before for the times $t_3$ and $t_1$ gives the bound $\tO(1/n)$.
\end{proof}

\begin{lemma}[Probability of a path, strenghtened version]\label{lemma:2nd_ad_hoc}
For any input $U$ of size $d$, we have the following under $E_{typ}$. 
Let $t\geq 0$ with $t < T_U$,  $x,y \in [n]$, $i,j\in [[k]]$ and let $w_1,w_2\in\mathcal{W}_k^{NC}$ be non-conjugated words.
Suppose that 
\begin{itemize}

\item[i)]  $[\arrows{x}{i}{w_1}{y}{j}]$ is not determined by $\cE_t$; 

\item[ii)] $(y,j,w_1)\notin \cE_t^*$;

\item[iii)] there is no $z\in [n]$, such that $[\thread{z}{0}{w_2}]$ is determined by $\cE_t$ and has length $k$.
\end{itemize}

Conditional on $\cE_{t}$, the probability that $\arrows{x}{i}{w_1}{y}{j}$ and that there exists $z\in [n]$ such that $[\thread{z}{0}{w_2}]$ is determined at the end of the exploration and has length $k$ is $\tO\left(n^{-1}\right)$.

Moreover, if $i=j=0$, the probability of the previous and that $(y,0)$ is a $\sigma$-lower-record at congruence $0$ on $[\thread{x}{0}{w_1}]$ is $\tO\left(n^{-3/2}\right)$.
\end{lemma}
\begin{proof}
We will bound the probability that $\arrows{x}{i}{w_1}{y}{j}$ with a thread of length $t'$.
Let $t_1$ be the last hitting time in the exploration of the thread and let $t_2$ be the hitting time that fully reveals $[\thread{z}{0}{w_2}]$.  Observe that $t< t_2\leq t_1\leq t+t'$. 

Suppose first that $t_2<t_1$. Then, since $x_{t_2+1}\in B^-_{t_2}(x_{t_2},k)$, which has size $\tO(1)$, the probability it happens at time $t_2$ is $\tO(1/n)$. Also note that at time $t_2$ the thread is still not determined and $(y,j,w_1)\notin \cE_{t_2}$, so we can apply Lemma~\ref{lemma:findpath2} with $t=t_2$ to show that the probability $x_{t_2+1}$ connects to $y$ with a path of length $t+t'-t_2$ is $\tO(1/n)$. By a union bound over $t_2$, the total contribution is $\tO(n^{-3/2})$. 

Suppose now that $t_2=t_1$. Let $x_1=x_{t_1+1}$ and $x_2=x_{t_1}$, and let $x_3=x_{t_3+1}$, where $t_3$ is the time of the last hit before time $t_2$, which exists. Observe that
\begin{align*}
x_{1}\in B^-_{t_1}(y,f_{t_1}), \quad x_{2}\in B^+_{t_1}(x_1,k)\quad \text{and}\quad x_{3}\in B^-_{t_3}(x_{2},f_{t_1})
\end{align*}
It follows that there are at most $\tO(1)$ candidates for $x_3$, and the probability of hitting them eventually before time $t_2$ is $\tO(n^{-1/2})$. Given $t'$ there are at most $\tO(1)$ possible values for $t_2$ by $E_{typ}$, thus the probability of the thread of length $t'$ is $\tO(1/n)$. All together we obtain a probability of $\tO(n^{-3/2})$. 

	For the first part it suffices to do a union bound over all values of $t'$. For the second part, we do the union bound over $t'$ with the extra $k/t'$ factor which stands for the probability that $y$ is a lower record, and we conclude as before using that $\sum_{t'\leq t_{max}} \frac{1}{t'} = \tilde O(1)$.
\end{proof}

\begin{remark}\label{rem:exploration}
In the rest of the paper we will use the lemmas of this section to bound the probability of certain events. In practice, we will often not directly refer to the notational formalism of an exploration $\cE$ with an input set $U$, but rather explore certain specified threads in some specified order. The two are of course equivalent. The parameter $d$ of this section will be the number of threads explored.
\end{remark}

\section{Proof of Lemmas~\ref{lemma:AL1} and~\ref{lemma:AL4}}
\label{sec:cycleGood}

In this section we prove the two lemmas that control the collisions of ``cycle''-type. 

Recall that we would like to control the probability of the following events: for $(i,h,j)\in[2]^3$, $p\in [\lambda_i], q\in[\lambda_j]$ and $r,s\in[[k]]$ with $(j,s)\neq (h,0)$
\begin{align}\label{eq:exclude2}
		(\beta_p^i,r) \stackrel{w_h}{\longrightarrow}	(e_q^j,s).
	\end{align}
	
We will replace this event by a more symmetric one:  for $(i,h,j)\in[2]^3$, $p\in [\lambda_i], q\in[\lambda_j]$ and $r,s\in[[k]]$ with $(j,s)\neq (h,0)$
\begin{align}\label{eq:exclude3}
		(\beta_p^i,r) \stackrel{w_h}{\longrightarrow}	(\beta_q^j,s).
	\end{align}
As $e_q^i=(\alpha_q^i,\beta_q^i)$, the former event implies the latter.
Thus it suffices to bound from above the probability of the latter.

We will proceed by union bound on the vertices $\beta_p^i$ and $\beta_q^j$ involved in the collision considered. The role of $\beta_p^i$ and $\beta_q^j$ will be played by vertices $u$ and $v$ below. In order to control the wanted events, we will need to explore the threads of $u$ and $v$ (plus some other vertices involved in the event if necessary), using the toolbox designed in the last section.
\medskip

We define the following events. For $I=(i,h,j)\in [2]^3$, $u,v\in [n]$ and $r,s\in [[k]]$, let $E^{r,s}_{u,v}(I)$ be the intersection of the following events:
\begin{itemize}
	\item[] (i): $u$ is on a cycle of $A_{w_i}$, and $(u,0)$ is a $\sigma_i$-minimum in $[\thread{u}{0}{w_i}]$,
	\item[] (ii): $v$ is on a cycle of $A_{w_j}$, and $(v,0)$ is a $\sigma_j$-minimum in $[\thread{v}{0}{w_j}]$,
	\item[] (iii): $\arrows{u}{r}{w_h}{v}{s}$.
\end{itemize}
Observe that if $u=v$ and $i=j$, the first two events are the same. The case $h=j$ and $s=0$ will not be considered throughout the section since it does not appear in the events we need to consider (see Definition~\ref{def:cycle-good}) and, moreover, the results stated below do not hold for it (see Remark~\ref{rem:s=0}). 

It will be useful to define
$$
F^{r,s}_{u,v}(I)= 
\begin{cases}
	E^{r,s}_{u,v}(I) & \text{if }u=v\\
	E^{r,s}_{u,v}(I)\setminus \left(\bigcup\limits_{x\in {n}}  \bigcup\limits_{r',s'\in [[k]]}E^{r',s'}_{x,x}(I)\right) & \text{otherwise,}
\end{cases}
$$
where if $h=j$ we exclude the value $s'$ from the last union.
Observe that the probability that none of the $E_{u,v}^{r,s}(I)$ hold is the same as the probability that none of the $F_{u,v}^{r,s}(I)$ hold, we will estimate the latter.
\begin{lemma}\label{lemma:cycle-good}
Let $I\in [2]^3$, $u,v\in [n]$ and $r,s\in [[k]]$, with either $h\neq j$ or $s\neq 0$. Then, we have
$$
\Pr(F^{r,s}_{u,v}(I))=
\begin{cases}
\tO(n^{-3/2}) & \text{if }u=v\\
\tO(n^{-5/2}) & \text{otherwise.}
\end{cases}
$$
\end{lemma}

\begin{proof}%
	By Proposition~\ref{prop:Etyp}, it is in fact enough to show this for the event  $F=F_{u,v}^{r,s}(I)\cap E_{typ}$. 
	Without loss of generality, we will assume $h=1$ to simplify the notation in the proof. We will bound the probability of $F$ in a different way depending on the values of $u,v,r,s,i,j$.

	\begin{itemize}
		\item Subevent 1:  $u=v$, $r=0$ and $i=1$. 
			
			If $j=1$, we let $P_1=[\arrows{u}{0}{w_1}{u}{s}]$, $P_2=[\arrows{u}{s}{w_1}{u}{0}]$. If $j=2$, we let $P_1=[\arrows{u}{0}{w_1}{u}{0}]$, $P_2=[\arrows{u}{0}{w_2}{u}{0}]$.
All these paths exist under the event we are considering.
			In the rest of the proof, we will distinguish a number of cases regarding the way these paths interact with each other, such as for example one path being determined by some other. For each case, we will then explore  (in the sense of Section~\ref{sec:exploration}, see Remark~\ref{rem:exploration}) these paths in an appropriate order, to bound the corresponding probability. The strategy will be used throughout the proofs of Sections~\ref{sec:cycleGood} and~\ref{sec:branchGood}.

Note that, since we are proving the lemma for all values of $s$,
			we may assume that $(u,s)$ is the first appearance of $u$ in $[\thread{u}{0}{w_1}]$, so there is no $u$ in the interior of $P_1$.

\begin{itemize}
	\item Subevent 1a: $P_2$ is not determined by $P_1$. By Lemma~\ref{lemma:findpath2}, the probability of $P_1$ is $\tO(n^{-1/2})$. Fix $t'\leq t_{\max}$. By Lemma~\ref{lemma:findpath2}, the probability of $P_2$ of length $t'$ is $\tO(1/n)$.  
	Note that if $(u,0)$ is $\sigma_1$-minimal on its $w$-cycle, then it is a $\sigma_1$-lower-record when starting the permutation at any point in the cycle. The probability of being a lower record starting from the beginning of $P_2$ is at most $k/t'$, and is independent of the rest. It follows that the total probability is at most
	\begin{align}\label{eq:lowerRecArgument}
	\sum_{t'=1}^{t_{\max}} \frac{k}{t'}\tO(n^{-3/2})=\tO(n^{-3/2}).
	\end{align}

	\item Subevent 1b: $P_2$ is determined by $P_1$. We claim that, after exposing $P_1$, there exists a vertex of out-degree at least $2$ in the first $k$ steps of $P_1$. Suppose that this was not the case. If both $P_1$ and $P_2$ have length at least $k$, they would coincide for the first $k$ steps. If one of the two is shorter than $k$, then in must be $P_1$ as otherwise $P_1$ would contain $u$ which we excluded. Then, $P_1$ would be a simple cycle and $P_2$ would need to wrap around it. In both cases (using Claim~\ref{claim:trajectory} in the first case), this implies that $w_1$ is self-conjugated if $j=1$ and that $w_1$ and $w_2$ are conjugates if $j=2$, and these cases were excluded by hypothesis. 
	
		The existence of such vertex of large out-degree implies the existence of a vertex of in-degree at least $2$ in the first $k$ steps of $P_1$ (indeed, since $u$ is not traversed twice by $P_1$, the trajectory responsible for the vertex of outdegree $2$ needs to merge with $P_1$ during the exploration at one of its first $k$ vertices). Since $u$ does not appear inside $P_1$, we can apply Lemma~\ref{lemma:findpath4} to obtain that the probability of this case, with $P_1$ of a given length $t'$, is $\tO(n^{-3/2})$. By the same argument used in~\eqref{eq:lowerRecArgument} to take into account the lower-record, the total probability is at most $\tO(n^{-3/2})$.
	\end{itemize}

\item Subevent 2: $u\neq v$, $r=0$ and $i=1$.

	Recall that, by definition of $F_{u,v}^{r,s}$, we can assume that none of the events $E_{x,x}^{r',s'}$ hold, in particular we can exclude Subevent 1. We let $P_1=[\arrows{u}{0}{w_1}{v}{s}]$, $P_2=[\arrows{v}{s}{w_1}{u}{0}]$, $P_3=[\arrows{v}{0}{w_j}{v}{0}]$, which exist under the event we consider.  Let $t_1,t_2,t_3$ be the lengths of the corresponding paths. We can assume that $(v,s)$ is the first occurrence of $v$ in the cycle containing $u$. Moreover, $P_2$ cannot be determined by $P_1$ since we excluded Subevent 1 and thus $u$ has indegree $0$ in $P_1$. We now explore $P_1$, then $P_2$.
We distinguish two cases:

\begin{itemize}

	\item Subevent 2a: $P_1$ and $P_2$ together do not determine $P_3$.
		We use Lemma~\ref{lemma:findpath2} three times, for the probability to find the paths $P_1,P_2,P_3$. The hypothesis of the lemma are always satisfied by the assumption of this case. Given $t_1,t_2,t_3$, each application of the lemma yields a $\tO(n^{-1})$ factor. Again we use that $u$ and $v$ are $\sigma_i$-minima and $\sigma_j$-minima respectively, an event independent from the rest (note that the sets of vertices visited at congruence $0$ on these two cycles are disjoint, so the two minima are also independent) which has probability $\frac{k}{t_1+t_2}\times \frac{k}{t_3}$.
Thus, the probability is upper bounded by
	$$
	\sum_{t_1,t_2,t_3 =1}^{t_{\max}} \   
		 \frac{k^2}{(t_1+t_2)t_3}
		\tO(n^{-3})
	 = \tO(n^{-5/2}).
	$$

\item Subevent 2b: $P_1$ and $P_2$ together determine $P_3$. By $E_{foll}\subset E_{typ}$, $P_3$ has length at most $\ell=\tO(1)$. We split the proof depending on whether $u$ appears in $P_3$:

		\begin{itemize}
			\item Subevent 2bi: $u\in P_3$.
			In this case, then Subevent 2 also holds with $u$ and $v$ reversed. Since we can exclude Subevent 2a with $u$ and $v$ reversed, we can assume $P_3$ also determines~$P_1$ and~$P_2$. Therefore by $E_{typ}$, all lengths of the paths are $\tO(1)$. 
			
				We explore $P_1$, then $P_2$. Observe that $\cE_{t_1+t_2}$ must contain a vertex $x$ of in-degree at least~$2$: Indeed, otherwise it would form a simple cycle containing the $w_1$-cycle and $w_j$-cycle of vertices $u$ and $v$ respectively, which is impossible since $w_1$ is not self-conjugated and $w_1$ and $w_2$ are non-conjugated.
The time $t'$ when $x$ acquires its second incoming edge, is a hit to a vertex already visited (only $\tO(1)$ choices for the target). Moreover, $t'$ is different from $t_1$ and $t_2$, since the vertex visited at these time ($v$ and $u$ respectively) receives its first incoming edge.

				We can now use union bound on $t_1,t_2,t'$, for which there are only $\tO(1)$ value, and for $x$, for which there are also at most $\tO(1)$ choices given $\cE_{t'}$. By using Lemma~\ref{lemma:findpath2} three times, we get that the probability for this case if bounded by $\tO(n^{-3})$ which is more than what we need.

			\item 
Subevent 2bii: $u\notin P_3$.
				We explore $P_1$, then $P_2$, and we let $t'>t_1$ be the last hitting time in $P_2$ at which an edge  $e=(x,y)$ present in $P_3$ is revealed. Since $P_3$ is not determined by $P_1$ ($v$ has out-degree $0$ in $P_1$), such a time exists. Moreover, since $u$ does not belong to $P_3$, the hit vertex $y$ is different from $u$, and $t'<t_1+t_2$. Finally, by $E_{foll}\subset E_{typ}$, $x$ and $y$ belong respectively to $B^+_{t'}(v,\ell)$ and $B^-_{t'}(v,\ell)$.

				Conditionally to $\cE_{t'}$, the probability to hit a vertex in $B^-_{t'}(v,\ell)$ at time $t'$ is at most $\tO(n^{-1})$ by $E_{ball}\subset E_{typ}$. The probability to hit $v$ and $u$ respectively at times $t_1$ and $t_1+t_2$ conditioned on the automaton revealed at these times is $\tO(n^{-2})$, by Lemma~\ref{lemma:findpath2}. Therefore, recalling that $(X_t,X_{t+1})$ is the edge traversed at time $t$ of the exploration process, the probability of our event is bounded by
				\begin{align*}
\sum_{t_1,t_2,t'=1}^{ t_{\max}} &
				\frac{k}{t_1+t_2} \tO(n^{-2})
\mathbb{P}(X_{t'} \in  B^+_{t'}(v,\ell),
					X_{t'+1} \in  B^-_{t'}(v,\ell), t' \mbox{ is a hit})\\
& = \sum_{t_1,t_2,t'=1}^{ t_{\max}}
					\frac{k}{t_1+t_2} \tO(n^{-3})
\mathbb{P}(X_{t'} \in  B^+_{ t'}(v,\ell))\\
&= \sum_{t_1,t_2=1}^{ t_{\max}}
					\frac{k}{t_1+t_2} \tO(n^{-3})
\mathbb{E}(N_{t_1+t_2}),
				\end{align*}
				where the first equality uses that the source and the target of a hit are independent, the second equality uses linearity of expectation and where $N_t$ is the number of times $s\leq t$ such that $X_s$ belongs to $B^+_{s}(v,\ell)$.
				Note that we also have included the factor $\frac{k}{t_1+t_2}$ for the probability that $u$ is minimum on its cycle.
				Now, by $E_{ball}\subset E_{typ}$, and since a given vertex can be visited at most $2k$ times in the process (once for each congruence and word), we have $\mathbb{E}(N_{t_1+t_2})=\tO(1)$. It follows that the above sum is bounded by $\tO(n^{-5/2})$.
\end{itemize}
\end{itemize}

\item Subevent 3: $u\neq v$ and $r\neq 0$, with either $j\neq 1$ or $s\neq 0$. 

	We let $P_1=[\arrows{u}{0}{w_i}{u}{0}]$, $P_2=[\arrows{u}{r}{w_1}{v}{s}]$, $P_3=[\arrows{v}{0}{w_j}{v}{0}]$ and $t_1$, $t_2$, $t_3$ be the lengths of these sequences. 
	By definition of $F_{u,v}^{r,s}$ we know that Subevent 1 does not hold, and moreover we can assume that Subevent 2 does not hold since we already addressed this case (and it can be assumed that Subevent $2$ does not hold with $u$ and $v$ exchanged). Thus $u,v$ do not appear inside $P_1$ and $P_3$.

We notice that $P_2$ is not determined by $P_1$ since $v$ does not appear on $P_1$.
Moreover we can assume that $(v,s)$ is the first occurrence of $v$ along $P_2$ (indeed if $j=1$ and $(v,0)$ was appearing first, this would imply that $(v,s)$ appears in the future of $(v,0)$, which
is impossible since we exclude $E_{v,v}^{0,s}$).
Therefore $v$ has outdegree $0$ in $P_1\cup P_2$, so $P_3$ cannot be determined by $P_1\cup P_2$.

Given $t_1,t_2,t_3$, we can therefore use Lemma~\ref{lemma:findpath2} thrice to bound the probability to find the paths $P_1,P_2,P_3$. We obtain the upper bound 
	$$
	\sum_{t_1,t_2,t_3=1}^{ t_{\max}} \   
		\frac{k^2}{t_1 t_3}
		\tO(n^{-3})
		= \tO(n^{-5/2}).
	$$
		where as before the factor $\tfrac{k^2}{t_1t_3}$ is the contribution for $u$ and $v$ being $\sigma_i$- and $\sigma_j$-minimal on their cycle.

	\item Subevent 4: $u=v$ and $i=2$, or $u=v$, $r\neq 0$ and $i=1$.
			We let $P_1=[\arrows{u}{0}{w_i}{u}{0}]$, $P_2=[\arrows{u}{r}{w_1}{u}{s}]$ of length $t_1,t_2$.

	\begin{itemize}
	\item Subevent 4a: $P_2$ is not determined by $P_1$.
This case is very similar to case 1a. We can apply Lemma~\ref{lemma:findpath2} twice to get the bound
		$$
	\sum_{t_1,t_2=1}^{ t_{\max}} \frac{k}{t_1} \tO(n^{-2})=  \tO(n^{-3/2}).
		$$
	\item subevent 4b: $P_2$ is determined by $P_1$. Observe that $P_1$ has length at least $k$. We argue as in the case 1b. 
	
	Suppose first that there is a vertex of out-degree at least $2$ in the first $k$ steps of $P_1$. This implies the existence of a vertex of in-degree at least $2$ in the first $k$ steps of $P_1$. Since $u$ does not appear inside $P_1$, we can use Lemma~\ref{lemma:findpath4}, to obtain that the probability of $P_1$ of length $t_1$ is $\tO(n^{-3/2})$. Since $u$ is a $\sigma_i$-minimum in its cycle, the total probability is
			$$
	\sum_{t_1=1}^{ t_{\max}} \frac{1}{t_1} \tO(n^{-3/2})=  \tO(n^{-3/2}).
		$$
	Now suppose that there is no vertex of out-degree at least $2$ in the first $k$ steps of $P_1$. Since, by Claim~\ref{claim:trajectory}, $P_1$ and $P_2$ can only coincide in less than $k$ steps, this implies that $P_2$ is strictly included in the first $k$ steps of $P_1$. However, $P_2$ ends in $u$ but $P_1$ does not contain the vertex $u$ in its interior, a contradiction. 
	\end{itemize}
	\end{itemize}
\end{proof}

We can finally give the proof of Lemma~\ref{lemma:AL1} and~\ref{lemma:AL6}. Recall Definition~\ref{def:cycle-good} of the sets $\CC{i}$ and $\SSS$.

\begin{proof}[Proof of Lemma~\ref{lemma:AL1}]
	An automaton is in $\CC{i}$ if and only if it satisfies none of the events $E_{u,v}^{r,s}(i,i,i)$ for $s\neq 0$. By Lemma~\ref{lemma:cycle-good}
	we have for $i\in [2]$
$$
\Pr_{\AAl}((A,\sigma_i) \not \in \CC{i})= \Pr\left(\bigcup_{u,v,r,s\atop s\neq 0} E^{r,s}_{u,v}(i,i,i)\right)= \Pr\left(\bigcup_{u,v,r,s\atop s\neq 0} F^{r,s}_{u,v}(i,i,i)\right) \leq \sum_{u,v,r,s\atop s\neq 0} \Pr\left(F^{r,s}_{u,v}(i,i,i)\right)= \tO(n^{-1/2}),
$$
	where we used that $r,s$ take at most $k=\tilde{O}(1)$ values, and $(u,v)$ at most $n^2$ (among which $n$ are such that $u=v$).
\end{proof}

\begin{proof}[Proof of Lemma~\ref{lemma:AL4}]
An automaton is in $\SSS$ if and only if it satisfies at least one event of the form $E_{u,v}^{r,s}(i,h,j)$ with either $h\neq j$ or $s\neq 0$. 
Arguing as in the previous proof, by Lemma~\ref{lemma:cycle-good}
	we have 
$$
\Pr_{\AAll}( (A,\sigma_1,\sigma_2)\in \SSS)\leq \sum_{I\in [2]^3}\sum_{u,v,r,s} \Pr\left(F^{r,s}_{u,v}(I)\right)= \tO(n^{-1/2}),
$$
	where $s=0$ is excluded from the sum whenever the summand $I=(i,h,j)$ is such that $h=j$.
\end{proof}

\section{Proof of Lemma~\ref{lemma:AL5}}
\label{sec:branchGood}

This lemma is the most technical part of the paper. In this section we give a full proof of it, which is very close in spirit to the proof of the last section. 

{As in the previous section it will be much simpler to control larger, more symmetric, events. For this, we will consider a less restrictive notion of $(i,h,j)$-collision, namely we will exclude situations where we can find $p\in [\ell_i], 2\leq q\leq \ell_j+1,$ and $ r,s\in[[k]]$,  such that
	\begin{align} \label{eq:excludeT2}
		\arrows{\bb_p^i}{r}{w_h}{\bb_q^j}{s}
		\mbox { with } (j,s)\neq (h,0). 
	\end{align} 
Note that the event~\eqref{eq:excludeT2} is less restrictive than the collisions defined in~\eqref{eq:excludeT}, and thus to prove Lemma~\ref{lemma:AL5} it suffices to bound from above the probability of collisions of the type~\eqref{eq:excludeT2} instead of~\eqref{eq:excludeT}. 

More precisely, for any $I_1\in \cI_1$ and $I_2\in \cI_2$, we  now introduce events $E(I_1), E(I_2)$
which differ from the complements of $\FFF(I_1),\FFF(I_2)$ in two ways: first, $(i,h,j)$-collisions~\eqref{eq:excludeT} are replaced by the weaker notion~\eqref{eq:excludeT2}, and secondly, we insist that the threads considered end with a $w$-cycle of length 1. 
\begin{definition}
	For $I_1=(i_1,h_1,j_1) \in \cI_1$, we let $E(I_1)$ be the subset of $\AAbl$ formed by all $(A,x_1,\sigma_1)$ that are such that
	\begin{itemize}[leftmargin=29pt, topsep=0pt, parsep=0pt, itemsep=1pt]
		\item there exist $p_1\in [\ell_{i_1}], 2\leq q_1\leq \ell_{j_1}+1, r,s\in[[k]]$ with $(j_1,s)\neq (h_1,0)$ such that 
			$\arrows{\bb_{p_1}^{i_1}}{r}{w_{h_1}}{\bb_{q_1}^{j_1}}{s}\} ;
	$
		\item $\cyc{w_1}{x_1}=1$.
	\end{itemize}
	Similarly, for $I_2=(i_2,h_2,j_2) \in \cI_2$, we let $E(I_2)$ be the subset of $\AAbl$ formed by all $(A,x_2,\sigma_2)$ that are such that
	\begin{itemize}[leftmargin=29pt, topsep=0pt, parsep=0pt, itemsep=1pt]
		\item there exist $p_2\in [\ell_{i_2}], 2\leq q_2 \leq \ell_{j_2}+1, r,s\in[[k]]$ with $(j_2,s)\neq (h_2,0)$ such that 
			$\arrows{\bb_{p_2}^{i_2}}{r}{w_{h_2}}{\bb_{q_2}^{j_2}}{s}\} ;
	$
		\item $\cyc{w_2}{x_2}=1$.
	\end{itemize}
\end{definition}

From the above discussion, to prove Lemma~\eqref{lemma:AL5}, it is sufficient to prove
\begin{lemma}\label{lemma:AL5var} For $I_1\in \cI_1$ and $I_2 \in \cI_2$, we have
	\begin{align*}
		\mathbb{P}_{\AAbbll}(
		(A,v_1,v_2,\sigma_1,\sigma_2) \in E(I_1)\cap E(I_2) ) 
		 = \widetilde{O}\left(\frac{1}{n^2}\right).
	\end{align*}
\end{lemma}

Recall the definitions 
\begin{align*}
\cI_1&=\{(1,2,1), (1,2,2), (2,2,1)\}\\
\cI_2&=\{(2,1,2), (2,1,1), (1,1,2)\}
\end{align*}
Now, observe that from the 9 possible pairs in $\cI_1\times \cI_2$, the following are equivalent two by two by permuting the words $w_1$ and $w_2$,
\begin{align*}
I_1=(1,2,1), I_2=(1,1,2) &\text{ and } I_1=(2,2,1), I_2=(2,1,2),\\
I_1=(1,2,1), I_2=(2,1,1) &\text{ and } I_1=(1,2,2), I_2=(2,1,2),\\
I_1=(2,2,1), I_2=(2,1,1) &\text{ and } I_1=(1,2,2), I_2=(1,1,2).
\end{align*}
From these pairs, we will only check the ones in the LHS of the previous display.
\medskip

In the rest of this section we prove  Lemma~\ref{lemma:AL5var} successively for the six pairs $I_1,I_2$ to consider. 
In each case, the events $E(I_1)$ and $E(I_2)$ involve some lower records appearing in the threads of $x_1$ and $x_2$. In what follows we will consider vertices $u_1,v_1, u_2,v_2$ playing the role of these lower records. This is similar to what we did with the cycle minima in Lemma~\ref{lemma:cycle-good}, except that it will be simpler here to take directly these vertices at random rather than using union bound. In the core of the proof, we will only need to consider the threads starting from these vertices, and in fact the vertices $x_1$ and $x_2$ will play a role only at the very end. For this reason most of our work will concern modified events $\hE(I_1)$ and $\hE(I_2)$ in which the vertices $x_1$ and $x_2$ do not appear, but only the $u_i,v_i$. The precise definition of these events is given in each case, see below.

\smallskip

We introduce a terminology that will be useful in this section. The vertices $\bb_p^j$ are $\sigma_j$-lower-records for all $p\in [\ell_j]$, but  $\bb_{\ell_j+1}^j$ might not be one. We say that $v$ is a \nota{generalized $\sigma_j$-lower record} if $v=\bb_p^j$ for some $p\in [\ell_j+1]$.
\begin{remark}\label{rem:lower_record}
It is well-known that the expected number of lower-records in a random permutation of fixed length at most $n$ is of order $\log(n)$. In fact, the probability that it has more than $\log(n)^2$ lower-records is smaller than any polynomial. Therefore, the probability that a uniformly picked vertex in $[n]$ is a generalized lower-record on a given thread is $\tO(n^{-1})$. We will use this argument several times below.
\end{remark}

For the rest of the section, we fix two vertices $x_1,x_2 \in [n]$ and two different words $w_1,w_2\in \mathcal{W}_k^{NC}$.

}

\subsection{Case $I_1=(1,2,1)$}
\label{subsec:121}

\begin{figure}[h]
	\begin{center}
		\includegraphics[width=0.9\linewidth]{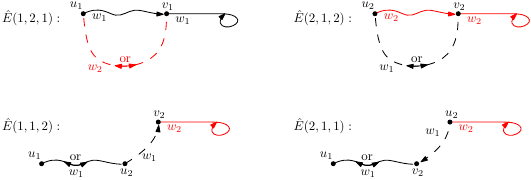}
			\caption{Pictorial view of the events of Section~\ref{subsec:121}. 
		In order to realize the full event $E(I_1)\cap E(I_2)$, we will need to connect in each case a branch emanating from vertices $x_1, x_2$, to some vertices in this picture. In each case the path that realizes the collision between lower records is in dashed line, while plain lines will {\it in fine} be parts of the branches of the two distinguished vertices $x_1$ and~$x_2$.}
		\label{fig:E1E2-1}
	\end{center}
\end{figure}

Define \nota{$\hat{E}(1,2,1)$} to be the event that for randomly and independently chosen $u_1,v_1\in [n]$, the following holds
\begin{itemize}
\item[] \nota{$\hE_1(1,2,1)$}: $\cyc{w_1}{v_1}=1$;
\item[] \nota{$\hE_2(1,2,1)$}: $(v_1,0)$ is a generalized $\sigma_1$-lower-record at congruence $0$ in the $w_1$-thread of $(u_1,0)$;	
\item[] \nota{$\hE_3(1,2,1)$}: either $\arrows{u_1}{r_1}{w_2}{v_1}{s_1}$ or $\arrows{v_1}{s_1}{w_2}{u_1}{r_1}$, for some $r_1,s_1\in [[k]]$.
\end{itemize}

	The vertices $u_1,v_1$ will play the role of the generalized lower records $\bb_{\min(p_1,q_1)}^{1}, \bb_{\max(p_1,q_1)}^{1}$ in the event $E(I_1)$, respectively.

 We refine the event $\hE(1,2,1)$ by considering the following additional subevent:
\begin{itemize}
	\item[] \nota{$\hE_4(1,2,1)$}: there exists $z\in [n]$ such that $[\thread{z}{0}{w_2}]$ is determined by $\cE_t$ and has length $k$, where $t$ is the time at the end of the exploration of $\hE(1,2,1)$.
\end{itemize}
Throughout the proof we will strengthen certain statements under the condition that $\hE_4$ holds, or benefit from the fact that $\hE_4$ does not hold to prove others.

\subsubsection{Contribution of $I_2=(2,1,2)$}

Define \nota{$\hat{E}(2,1,2)$} to be the event that for randomly and independently chosen $u_2,v_2\in [n]$, the following holds
\begin{itemize}
\item[] \nota{$\hE_1(2,1,2)$}: $\cyc{w_2}{v_2}=1$;
\item[] \nota{$\hE_2(2,1,2)$}: $(v_2,0)$ is a  generalized $\sigma_2$-lower-record at congruence $0$ in the $w_2$-thread of $(u_2,0)$;
\item[] \nota{$\hE_3(2,1,2)$}: either $\arrows{u_2}{r_2}{w_1}{v_2}{s_2}$ or $\arrows{v_2}{s_2}{w_1}{u_2}{r_2}$, for some $r_2,s_2\in [[k]]$.
\end{itemize}

	The vertices $u_2,v_2$ will play the role of the generalized lower records $\bb_{\min(p_2,q_2)}^{2}, \bb_{\max(p_2,q_2)}^{2}$ in the event $E(I_2)$, respectively.

We will start by studying $\hE(2,1,2)$ given $\hE(1,2,1)$.
In fact, we first study the case where we exclude the event $\hE_4(1,2,1)$.
\begin{lemma}\label{lemma:212}
We have
$$
\Pr(\hE(2,1,2)\mid \hE(1,2,1),\hE_4(1,2,1)^c)=\tO(n^{-2}).
$$
\end{lemma}
\begin{proof}
 By Proposition~\ref{prop:Etyp}, we can assume that we are under $E_{typ}$. Throughout the proof, we write $\hE_i=\hE_i(2,1,2)$ for $i\in [3]$.
	We first explore all the threads involved in $\hE(1,2,1)$, and we
	let $t_1$ be the time at the end of the exploration.
Let us denote by $\cC_{t}$ the set of vertices $u$ for which $B_{t_1}(u,k)$ is not a 
	path. By $E_{path}\subset E_{typ}$, $|\cC_{t_1}|=\tO(1)$. 

We first show it is unlikely that $\cE_{t_1}$
	determines $\hE(2,1,2)$. Suppose this is the case. Then $(v_2,0,w_2)\notin \cE_{t_1}$ as otherwise $\cyc{w_2}{v_2}\neq 1$ by $\hE_4$. By Lemma~\ref{claim:trajectory}, $B_{t_1}(v_2,k)$ is not in a 
	path and by $E_{path}\subset E_{typ}$ there are $\tO(1)$ choices for $v_2$. Since $[\thread{v_2}{0}{w_2}]$ is determined by $\cE_{t_1}$ but $(v_2,0,w_2)\notin \cE_{t_1}$, by $E_{foll}\subset E_{typ}$, it has length $\ell=\tO(1)$. Thus, $u_2\in B_{t_1}(v_2,\ell)$, which by $E_{ball}\subset E_{typ}$ has size $\tO(1)$. So there are constantly many choices for $(u_2,v_2)$, and the probability of choosing them is $\tO(n^{-2})$. In what follows we will assume that $\cE_{t_1}$ does not determine $\hE(2,1,2)$.

We split into two cases.
\begin{itemize}
	\item $u_2=v_2$. The probability of such event is $1/n$. For each $i\in [[k]]$, consider the events
\begin{itemize}
\item[] $\hE_1(i)$: $\cyc{w_2}{u_2,i}=1$, let $P=[\thread{u_2}{i}{w_2}]$;
\item[] $\hF_1$: $u_2$ does not appear in the interior of $P$.
\end{itemize}
Observe that if $\hE(1,2,1)$ holds, then  $\hE_1(i)\cap \hF_1 \cap\hE_2\cap \hE_3$ will hold for some $i\in [[k]]$; we will bound the latter probability. Note that the event $\hE_2$ trivially holds in this case.

Let $Q=[\arrows{u_2}{r_2}{w_1}{u_2}{s_2}]$. We further split into two cases.
\begin{itemize}
\item[-] $u_2\notin \cE_{t_1}$. If so $P$ cannot be determined in $\cE_{t_1}$ as $u_2$ has out-degree $0$ in it. 
By $\hE_4$, we can use Lemma~\ref{lemma:thread_short_cycle} to $P$ to show that the probability of $\hE_1(i)$ is $\tO(n^{-1/2})$. Let $t_2$ be the time after exposing $P$. Observe that $Q$ is not determined in $\cE_{t_2}$, as by $\hF_1$, $u_2$ has in-degree $0$ in it. By Lemma~\ref{lemma:findpath2}, the probability of $\hE_3$ is $\tO(n^{-1/2})$. The total contribution of this case is $\tO(n^{-2})$.

\item[-] $u_2\in \cE_{t_1}$. The probability of such event is $\tO(n^{-1/2})$ since by $E_{len}\subset E_{typ}$, $t_1=\tO(n^{1/2})$. Since the event $\hE(2,1,2)$ is not determined, either $P$ or $Q$ are not determined. In the first case we use Lemma~\ref{lemma:thread_short_cycle} to reveal $P$ and in the second we use Lemma~\ref{lemma:findpath2} to reveal $Q$. So $\hE_1(i)\cap \hE_3$ has probability $\tO(n^{-1/2})$ to hold. The total contribution is $\tO(n^{-2})$.
\end{itemize}

\item $u_2\neq v_2$. Let $P_1=[\arrows{u_2}{0}{w_2}{v_2}{0}]$, $P_2=[\thread{v_2}{0}{w_2}]$, $P$ the concatenation of both and $Q=[\arrows{u_2}{r_2}{w_1}{v_2}{s_2}]$ or $Q=[\arrows{v_2}{s_2}{w_1}{u_2}{r_2}]$ depending on the case considered.
Consider the following events,
\begin{itemize}
\item[] $\hF_1$:  $u_2,v_2$ do not appear in $P$ (apart from the two required appearances);
\item[] $\hF_3$:  $v_2$ does not appear in the interior of $Q$.
\end{itemize}

		We may assume $\hF_3$ as we can always take the first or last appearance of $v_2$ in $Q$ (up to possibly changing the value of $s_2$). The reason why we may assume $\hF_1$ is more subtle. Suppose that $u_2$ appears in $P$ at a certain congruence $i$. Then
\begin{align}
\arrows{u_2}{0}{w_2}{u_2}{i} &\text{ and }\cyc{w_2}{u_2,i}=1
\end{align}
holds. However, this event was already excluded by the previous case $u_2=v_2$ for $r_2=0,s_2=i$. The case of $v_2$ appearing in $P$ is similar, as now we have that for some $j\in [[k]]$
\begin{align}
\arrows{v_2}{j}{w_2}{v_2}{0} &\text{ and }\cyc{w_2}{v_2,0}=1
\end{align}
and this event was already excluded by the previous case $u_2=v_2$ for $i=0,r_2=j,s_2=0$. So our goal is to bound the probability of $\hE_1\cap \hF_1\cap \hE_2\cap \hE_3\cap \hF_3$.
\begin{itemize}

	\item[-] $u_2\in \cE_{t_1}$. The probability of such event is $\tO(n^{-1/2})$. Since $\hE(2,1,2)$ is not determined by $\cE_{t_1}$ either $P$ or $Q$ is not revealed at time $t_1$. If $P$ is not determined, we can use Lemma~\ref{lemma:thread_short_cycle} and the random placement of $v_2$ as a generalized $\sigma_2$-lower-record at congruence $0$ in $P$ (see Remark~\ref{rem:lower_record}), to get probability $\tO(n^{-3/2})$. Otherwise, $Q$ is not determined. Again, randomly placing $v_2$ as a generalized $\sigma_2$-lower-record at congruence $0$ in $P$, and using Lemma~\ref{lemma:findpath2} to $Q$ gives the result.

	\item[-] $v_2\in\cE_{t_1}$, $u_2\notin\cE_{t_1}$. This event has probability $\tO(n^{-1/2})$. Since $P_2$ contains a cycle of length $k$, by $\hE_4$, it cannot be determined by $\cE_{t_1}$.
Apply Lemma~\ref{lemma:thread_short_cycle} to reveal $P_2$ and get the extra probability $\tO(n^{-1/2})$. Let $t_2$ be the time after exposing $P_2$ and note that $u_2$ has out-degree $0$ in $\cE_{t_2}$ so $P_1$ is not determined by it. So we can use Lemma~\ref{lemma:P1_enhanced} for $P_1$ to obtain another contribution of $\tO(n^{-1})$, leading to the bound $\tO(n^{-2})$ in total.

\item[-] $u_2,v_2\notin \cE_{t_1}$. Expose $P_2$ which is clearly not determined. We first assume that $v_2$ is a $\sigma_2$-lower-record.
By Lemma~\ref{lemma:thread_short_cycle}, the probability of $\hE_1$ is $\tO(n^{-1/2})$. Let $t_2$ be the time after exposing $P_2$.
Note that $P_1$ has length at least $k$.
After exposing $P_1$, but before revealing $Q$, one of the following cases hold:

\begin{itemize}
	\item[-] There is a vertex of in-degree at least $2$ in the first $k$ steps of $P_1$.  By our assumption, $u_2\notin \cE_{t_2}$, $v_2\notin \cE_{t_2}^*$ and $v_2$ does not appear inside $P_1$. We can thus apply Lemma~\ref{lemma:findpath4} to obtain a probability of $O(n^{-3/2})$ for the existence of $P_1$ of a given length $t'$. Taking into account the lower-record in $\hE_2$, we obtain a contribution of
\begin{equation}\label{SPSA}
\sum_{t'=1}^{t_{\max}} \frac{k}{t'} \tO(n^{-3/2})= \tO(n^{-3/2}),
\end{equation}
			which together with the contribution of $\hE_1$ gives the bound $\tO(n^{-2})$.

	\item[-] There is no vertex of in-degree at least $2$ in the first $k$ steps of $P_1$. 
We first consider the case $Q=[\arrows{u_2}{r_2}{w_1}{v_2}{s_2}]$. 
We apply Lemma~\ref{lemma:P1_enhanced} to $P_1$ to get that the probability of $\hE_2$ is $\tO(1/n)$. Let $t_3$ be the time after exposing $P_1$. In $\cE_{t_3}$, $u_2$ is at the start of a path of length $k$ with vertices of in- and out-degree $1$ (there is no in-degree at least $2$ by assumption, and since the in-degree of $u_2$ is $0$, there is also no out-degree at least $2$). Since $P_1$ has length $k$ and finished at $v_2$, $v_2$ cannot appear in the first $k$ steps of $P_1$, in particular $Q$ is not determined. Lemma~\ref{lemma:findpath2} gives that the probability of $\hE_3$ is at most $\tO(n^{-1/2})$. The total contribution of this case is $O(n^{-2})$, as required.

Now consider the case $Q=[\arrows{v_2}{s_2}{w_1}{u_2}{r_2}]$. Since $Q$ is not determined ($u_2$ has in-degree $0$ in $\cE_{t_2}$) by Lemma~\ref{lemma:findpath2}, the probability of $\hE_3$  is $\tO(n^{-1/2})$.  Let $t_3'$ be the time after exposing $Q$. Since $P_1$ is not determined ($v_2$ has in-degree $0$ in $\cE_{t_3'}$),
		if we fix the length $t'$ of $P_1$, we can apply
		Lemma~\ref{lemma:findpath2} to obtain a contribution of $\tO(n^{-1})$, which together with contributions of $\hE_1$ and $\hE_3$ gives $\tO(n^{-2})$. It remains to perform union bound on $t'$, but since by $\hE_2$, $(v_2,0)$ is a $\sigma_2$-lower-record at congruence $0$ in $[\thread{u_2}{0}{w_2}]$, this only gives a contribution of 
		\begin{equation}\label{SPSA2}
		\sum_{t'=1}^{t_{max}}\frac{k}{t'} = \tO(1).	
		\end{equation}
\end{itemize}
		Now suppose that $v_2$ is not a $\sigma_2$-lower record (but a generalized one). Since $P_2$ is not determined by $\cE_{t_1}$ and has length $k$, when applying Lemma~\ref{lemma:thread_short_cycle} to reveal $P_2$ we obtain a probability of $\tO(1/n)$, we thus have a gain of $\tO(n^{-1/2})$ compared to the previous proof. The rest of the proof can be repeated identically, removing the factor $1/t'$ from~\eqref{SPSA} and~\eqref{SPSA2}, which gives an extra factor of $\tO(n^{1/2})$, which is compensated by the previous gain. Thus as before we obtain a total probability of $\tO(n^{-2})$.

\end{itemize}
\end{itemize}
\end{proof}

In fact, the previous proof also shows: 
\begin{lemma}\label{lemmaEXcoro:121}
We have
$$
\Pr(\hE(1,2,1))=\tO(n^{-2}).
$$
\end{lemma}
\begin{proof}
	The proof of Lemma~\ref{lemma:212} can be copied verbatim, up to defining in the first paragraph $t_1=0$ and $\cE_{t_1}=\emptyset$, thus proving that 
$\Pr(\hE(2,1,2))=\tO(n^{-2})$.
This is what we want, up to exchanging the roles of $1$ and $2$.
\end{proof}

\begin{remark}\label{rem:212+}
	Observe that in the proof of Lemma~\ref{lemma:212}, we always obtained a contribution $\tO(n^{-1/2})$ for the event $\cyc{w_2}{v_2}=1$ (or $\cyc{w_2}{v_2,i}=1)$ which came from applying Lemma~\ref{lemma:thread_short_cycle} (this lemma is invoked exactly once in each case of the analysis). If we condition on $\hE_4(1,2,1)$, then Lemma~\ref{lemma:thread_short_cycle} cannot be applied and this contribution is no longer present. Nevertheless, the rest of the proof is still valid. It follows that
$$
\Pr(\hE(2,1,2)\mid \hE(1,2,1), \cyc{w_2}{u_1,r_1}=1)=\tO(n^{-3/2}).
$$
\end{remark}

In order to have control on the event $\hE_4(1,2,1)$
 we will prove the following modification of Lemma~\ref{lemmaEXcoro:121}.

\begin{lemma}\label{lemma:121+}
We have
$$
\Pr(\hE(1,2,1), \hE_4(1,2,1))=\tO(n^{-5/2}).
$$
\end{lemma}
\begin{proof}
Assume that we are under $E_{typ}$ and write $\hE_i=\hE_i(1,2,1)$ for $i\in [4]$. 
Assume that the event $\hE_3$ is $\arrows{v_1}{s_1}{w_2}{u_1}{r_1}$, the other case is handled similarly. Define,

\begin{itemize}
\item[] $\hE_1(i)$: $\cyc{w_1}{u_1,i}=1$, let $P=[\thread{u_1}{i}{w_1}]$;
\item[] $\hF_1$: $u_1$ does not appear in $P$ (apart from the beginning);
\item[] $\hE_5$: there exists $z\in [n]$ such that $P'=[\thread{z}{0}{w_2}]$ is determined by $P$ and has length $k$.
\end{itemize}

If the event of the lemma holds, then $\hE_1(i)\cap\hE_2\cap\hE_3\cap\hE_4\cap\hF_1$ holds for some $i\in[[k]]$. Indeed, we may assume that $\hF_1$  holds for the following reason: Consider the set of $i\in [[k]]$ such that that $\hE_1(i)\cap\hE_2\cap\hE_3\cap\hE_4$ holds. Then, we can choose $i$ in this set so there is no appearance of $u_1$ in $P=[\thread{u_1}{i}{w_1}]$, which implies that $\hE_1(i)\cap\hE_2\cap\hE_3\cap\hE_4\cap\hF_1$ holds for some $i\in[[k]]$.
We will split the proof depending on $\hE_5$.   

	We first consider the case where $\hE_5$ holds. We can apply Lemma~\ref{lemma:1st_ad_hoc} with $x=u_1$ to bound the probability of $\hE_1(i)\cap \hE_5$ by $\tO(1/n)$. Let $t_1$ be the time after the exploration that reveals $P$. Choose $v_1$ uniformly at random, the probability it is a generalized lower record in $P$ is $\tO(1/n)$.
By $\hF_1$, we may assume that $u_1$ has in-degree $0$ in $\cE_{t_1}$, so $\hE_3$ is not determined. By Lemma~\ref{lemma:findpath2} with $x_1=v_1$, $x_2=u_1$, $i_1=s_1$ and $i_2=r_1$, the probability of $\hE_3$ is $\tO(n^{-1/2})$. Thus, the total probability is $\tO(n^{-5/2})$ as desired.

	Now let us consider the case where $\hE_5$ does not hold. We can apply Lemma~\ref{lemma:thread_short_cycle} with $x=u_1$ and choose $v_1$ at random to obtain a joint contribution of $\tO(n^{-3/2})$. Since after exposing $P$, $u_1$ has in-degree $0$, and since $\hE_5$ does not hold, we may apply Lemma~\ref{lemma:2nd_ad_hoc} with $x_1=v_1$, $x_2=u_1$, $i_1=s_1$ and $i_2=r_1$ to obtain an additional contribution of $\tO(1/n)$. The total contribution is $\tO(n^{-5/2})$.

\end{proof}

We finally obtained the required result to treat the case of $I_1=(1,2,1)$ and $I_2=(2,1,2)$.
\begin{lemma}\label{lemma:121-212}
We have
$$
\Pr(E(1,2,1), E(2,1,2))=\tO(n^{-2}).
$$
\end{lemma}
\begin{proof}
Combining  Lemmas~\ref{lemma:212},~\ref{lemmaEXcoro:121} and~\ref{lemma:121+} and Remark~\ref{rem:212+} we have that
\begin{align}\label{eq:121-212_hat}
\Pr(\hE(1,2,1), \hE(2,1,2))=\tO(n^{-4}),
\end{align}
	where as before $u_1,v_1,u_2,v_2$ are chosen uniformly at random.

Recall that $x_1$ and $x_2$ are fixed but play no role
	in the events $\hE(1,2,1), \hE(2,1,2)$.
	Consider the events
\begin{itemize}
	\item[] $G_1$: $(u_1,0)$ appears in $[\thread{x_1}{0}{w_1}]$ and is a $\sigma_1$-lower-record at congruence $0$ on that thread;
\item[] $G_2$: $(u_2,0)$ appears in $[\thread{x_2}{0}{w_2}]$ and is a $\sigma_2$-lower-record at congruence $0$ on that thread.
\end{itemize}
Observe that the event $E(1,2,1)\cap E(2,1,2)$
	is equal to the union over all $u_1,v_1,u_2,v_2$ of the intersection $\hE(1,2,1)\cap \hE(2,1,2) \cap G_1\cap G_2$.
	To see this, note that in the event $\hE(1,2,1)\cap \hE(2,1,2) \cap G_1\cap G_2$ the vertices $u_1,v_1$ play the same role as the generalized lower records $\bb_{\min(p_1,q_1)}^{1}, \bb_{\max(p_1,q_1)}^{1}$ in the event $E(1,2,1)$, and similarly  $u_2,v_2$ play the same role as the generalized lower records $\bb_{\min(p_2,q_2)}^{2}, \bb_{\max(p_2,q_2)}^{2}$ in the event $E(2,1,2)$ (note that since $\min(p_g,q_g)\leq p_g  \leq \ell_g$ for $g\in [2]$, $\bb_{\min(p_1,q_1)}^{1}$ and $\bb_{\min(p_2,q_2)}^{2}$ are lower records and not only \emph{generalized ones}, in accordance with the requirement on $u_1, u_2$ in the events $G_1,G_2$).

	Observe also that the event $\hE(1,2,1)\cap \hE(2,1,2)$ is determined by the exploration of at most four threads emanating from $u_1,v_1,u_2,v_2$.

	Now, the probability of $G_1$ conditional on $\hE(1,2,1)\cap \hE(2,1,2)$ is $\tO(n^{-1})$: indeed, either $\{u_1,v_1,u_2,v_2\}\ni x_1$ (which happens with probability $O(1/n)$, independently of $\hE(1,2,1)\cap \hE(2,1,2)$) or $\{u_1,v_1,u_2,v_2\}\not\ni x_1$ and we can apply Lemma~\ref{lemma:P1_enhanced}  that gives a contribution of $\tO(n^{-1})$. By the same argument with $x_2$, the probability of $G_2$ conditional on $\hE(1,2,1)\cap \hE(2,1,2) \cap G_1$ is $O(n^{-1})$.

	By~\eqref{eq:121-212_hat} we thus obtain a total contribution of $\tO(n^{-6})$, and
by a union bound over all quadruples of vertices $(u_1,v_1,u_2,v_2)$, the probability of the final event is $O(n^{-2})$ as desired.
\end{proof}

\subsubsection{Contribution of $I_2=(1,1,2)$}

We move now to the case of $I_1=(1,2,1)$ and $I_2=(1,1,2)$.

Conditional on $\hat{E}(1,2,1)$, define $\hat{E}(1,1,2)$ to be the event that for randomly and independently chosen $u_2,v_2\in [n]$, the following holds
\begin{itemize}
\item[] $\hE_1(1,1,2)$: $\cyc{w_2}{v_2}=1$;
\item[] $\hE_2(1,1,2)$: either $(u_2,0)$ is a  generalized $\sigma_1$-lower-record at congruence $0$ in the $w_1$-thread of $(u_1,0)$, or vice versa;
\item[] $\hE_3(1,1,2)$: $\arrows{u_2}{r_2}{w_1}{v_2}{s_2}$.
\end{itemize}
In contrast to the first case, here we always assume that $\arrows{u_2}{r_2}{w_1}{v_2}{s_2}$.
The reason for this is that fact the $u_2$ might appear before or after $u_1$ in the $w_1$-thread.

	The vertices $u_2,v_2$ will play the role of the generalized lower records $\bb_{p_2}^{1}$ and $\bb_{q_2}^{2}$ in the event $E(I_2)$, respectively, while the vertex $u_1$ is needed because we need to consider simultaneously the event  $\hat{E}(1,2,1)$ -- in which we recall that $u_1$ plays the role of the vertex $\bb^1_{\min(p_1,q_1)}$ of event $E(I_1)$. Note that $u_2$ is always a lower-record.

The probability of $\hE(1,2,1)$ is already bounded by Lemmas~\ref{lemmaEXcoro:121} and~\ref{lemma:121+}.

\begin{lemma}\label{lemma:112}
We have
$$
\Pr(\hE(1,1,2)\mid \hE(1,2,1),\hE_4(1,2,1)^c)=\tO(n^{-2}).
$$
Moreover, the probability of the same event with the additional condition that $[\thread{v_2}{0}{w_2}]$ has length $k$ is $\tO(n^{-{5/2}})$.
\end{lemma}
\begin{proof}

Assume that we are under $E_{typ}$ and write $\hE_i=\hE_i(1,1,2)$ for $i\in [3]$. 

Observe that we only need to reveal one vertex in the $w_2$-thread, namely $v_2$. For $u_2$ we have three cases: (i) it is a   $\sigma_1$-lower-record at congruence $0$ in the $w_1$-thread of $(u_1,0)$ exposed by $\hE(1,2,1)$, or (ii) $u_1$ is a $\sigma_1$-lower-record at congruence $0$ in the $w_1$-thread of $(u_2,0)$. Let $t_1$ be the time after exposing $\hE(1,2,1)$.

We first discuss about $u_2$ and the probability of $\hE_2$. In case (i), since there are $\tO(1)$ many  $\sigma_1$-lower-records at congruence $0$ in the $w_1$-thread of $(u_1,0)$ (see Remark~\ref{rem:lower_record}), the probability $u_2$ is one of them is $\tO(1/n)$.
	In case (ii), recall that $u_1$ is a $\sigma_1$-lower record by $\hE(1,2,1)$. Since $u_2$ is uniformly random, either $u_2$ coincides with the starting point of one of the threads already revealed 
	 or we can apply Lemma~\ref{lemma:P1_enhanced} to show that $\hE_2$ has probability $\tO(1/n)$. Let $t_2$ be the time after exposing $[\arrows{u_2}{0}{w_1}{u_1}{0}]$ (if it was not determined by $t_1$) and $t_2=t_1$ otherwise. We split into two cases:
	
	\begin{itemize}
	\item[-] there is no $z\in [n]$ such that $[\thread{z}{0}{w_2}]$ is determined by $\cE_{t_2}$ and has length $k$. This implies that $\hE_1$ is not determined by $\cE_{t_2}$
If $v_2\in \cE_{t_2}$, which holds with probability $\tO(n^{-1/2})$, we can use Lemma~\ref{lemma:thread_short_cycle} to show that the probability of $\hE_1$ is $\tO(n^{-1/2})$ (or $\tO(1/n)$ with the additional condition that $v_2$ is in a short cycle), and we are done.	
Otherwise, $v_2\notin \cE_{t_2}$. Observe that, up to changing the value of $s_2$ in event $\hE_3$, we can assume
	that no $v_2$ appears in the interior of the thread $[\arrows{u_2}{r_2}{w_1}{v_2}{s_2}]$. The probability of $\hE_3$ is $\tO(n^{-1/2})$ by Lemma~\ref{lemma:findpath2}. Moreover, by our assumption, $v_2$ has out-degree $0$ after exposing $\hE_3$ and we can use Lemma~\ref{lemma:thread_short_cycle} to show that the probability of $\hE_1$ is $\tO(n^{-1/2})$ (or  $\tO(1/n)$ with the additional condition).
	
\item[-] there is $z\in [n]$ such that $[\thread{z}{0}{w_2}]$ is determined by $\cE_{t_2}$ and has length $k$.	As $\hE_4$ holds, we may decrease the probability of the event $\hE_2$ by using the second part of Lemma~\ref{lemma:2nd_ad_hoc} to show that it holds with probability $\tO(n^{3/2})$. Let us denote by $\cC_{t_2}$ the set of vertices $u$ for which $B_{t_2}(u,k)$ is not a path. By $E_{typ}$, $|\cC_{t_2}|=\tO(1)$. The probability that $v_2\in \cC_{t_2}$ is $\tO(1/n)$, so we may exclude this case. 
If $v_2\notin \cC_{t_2}$, then $\hE_1$ is not determined. By Lemma~\ref{lemma:thread_short_cycle}, its probability is at most $\tO(n^{-1/2})$ (or $\tO(1/n)$ with the additional condition), and we are done.
	\end{itemize}

\end{proof}

\begin{remark}\label{rem:112+}
As in Remark~\ref{rem:212+}, we have
$$
\Pr(\hE(1,1,2)\mid \hE(1,2,1), \hE_4(1,2,1))=\tO(n^{-3/2}).
$$
\end{remark}

We conclude,
\begin{lemma}\label{lemma:121-112}
We have
$$
\Pr(E(1,2,1), E(1,1,2))=\tO(n^{-2}).
$$
\end{lemma}
\begin{proof}
	The proof is identical to Lemma~\ref{lemma:121-212},
with vertices $u_1,v_1,u_2,v_2$ in the event $\hE(1,2,1)\cap \hE(1,1,2) \cap G_1\cap G_2$ 
	playing the role of vertices $\bb_{\min(p_1,q_1)}^{1}, \bb_{\max(p_1,q_1)}^{1}$ in the event $E(1,2,1)$, and $\bb_{p_2}^{1}, \bb_{q_2}^{2}$ in the event $E(1,1,2)$, respectively.
	Note that $\min(p_1,q_1)\leq p_1 \leq \ell_{1}$, so $u_1\equiv\bb_{\min(p_1,q_1)}^{1}$ is always a lower record on the thread of $x_1$.

	In the case where $v_2\equiv \bb_{q_2}^{2}$ is a lower record in the thread of $x_2$, we repeat the proof of Lemma~\ref{lemma:121-212} using Lemma~\ref{lemma:112} and Remarks~\ref{rem:112+} instead of Lemmas~\ref{lemma:212},~\ref{lemmaEXcoro:121},~\ref{lemma:121+} and Remark~\ref{rem:212+}. Moreover, we adapt the events $G_1$ and $G_2$, by connecting $x_1$ to $u_1$ or $u_2$ depending on their relative position in the $w_1$-thread (i.e. depending on whether $\min(p_1,q_1)\leq p_2$ or not), and $x_2$ to $v_2$.

	In the case that $v_2\equiv \bb_{q_2}^{2}$ is not a $\sigma_2$-lower-record of $[\thread{x_2}{0}{w_2}]$ but only a generalized one (i.e. $q_2=\ell_2+1$), a slight change is needed. In that case, we have the additional condition that $[\thread{v_2}{0}{w_2}]$ has length $k$, so we can use the second part of Lemma~\ref{lemma:112} to obtain an additional factor $\tO(n^{-1/2})$ and then apply Lemma~\ref{lemma:P1} instead of Lemma~\ref{lemma:P1_enhanced} to connect $x_2$ with $v_2$.
\end{proof}

\subsubsection{Contribution of $I_2=(2,1,1)$}

Conditional on $\hat{E}(1,2,1)$, define $\hat{E}(2,1,1)$ to be the event that for randomly and independently chosen $u_2,v_2\in [n]$, the following holds
\begin{itemize}
\item[] $\hE_1(2,1,1)$: $\cyc{w_2}{u_2}=1$;
\item[] $\hE_2(2,1,1)$: $(v_2,0)$ is a $\sigma_1$-lower-record at congruence $0$ in the $w_1$-thread of $(u_1,0)$ or vice versa;
\item[] $\hE_3(2,1,1)$: $\arrows{u_2}{r_2}{w_1}{v_2}{s_2}$.
\end{itemize}
Loosely speaking, the only difference between this definition and the definition of $\hat{E}(1,1,2)$ is that we reversed the direction of the path connecting the two threads: now it goes from the $w_2$-thread to the $w_1$-thread.

	The vertices $u_2,v_2$ will play the role of the generalized lower records $\bb_{p_2}^{2}$ and $\bb_{q_2}^{1}$ in the event $E(I_2)$, respectively, while the vertex $u_1$ is needed because we need to consider simultaneously the event  $\hat{E}(1,2,1)$ -- in which  we recall that $u_1$ plays the role of the vertex $\bb^1_{\min(p_1,q_1)}$ of event $E(I_1)$.

\begin{lemma}\label{lemma:211}
We have
$$
\Pr(\hE(2,1,1)\mid \hE(1,2,1),\hE_4(1,2,1)^c)=\tO(n^{-2}).
$$
\end{lemma}
\begin{proof}

	The proof of this result is very similar to the proof of Lemma~\ref{lemma:112}, and we only indicate the meaningful differences.
	As usual, we will prove the result under $E_{typ}$ and we let $\hE_i=\hE_i(2,1,1)$ for $i \in [3]$.

	Let $t_1$ be the time after exposing $\hE(1,2,1)$. We first place $v_2$. As in the proof of Lemma~\ref{lemma:112} for $u_2$, the probability of $\hE_2$ is $\tO(1/n)$ (or $\tO(n^{-3/2})$ with the additional condition that there is a $w_2$-cycle of length $1$). Let $t_2$ be the time after revealing the event $\hE_2$. Let $P=[\thread{u_2}{0}{w_2}]$ and $Q=[\arrows{u_2}{r_2}{w_1}{v_2}{s_2}]$.

If $u_2\in \cE_{t_2}$ the argument is identical as in the proof of Lemma~\ref{lemma:112}, thus we may assume that $u_2\notin \cE_{t_2}$. Up to changing the value of $s_2$, we may assume that $v_2$ does not appear in the interior of $Q$. We distinguish two cases. Expose $Q$ and let $t_3$ be the time after it. If there is an in-degree at least $2$ in the first $k$ steps of $Q$ in $\cE_{t_3}$, then we can apply Lemma~\ref{lemma:findpath4} with $x=u_2$ and $y=v_2$ and a union bound over $t'\leq t_{\max}$ to obtain that probability of $\hE_3$ under this conditioning is $\tO(1/n)$. Otherwise, there is no vertex of in-degree at least $2$ and, since $u_2\notin \cE_{t_2}$, we have that $u_2\notin \cC_{t_3}$ and $P$ (which has length at least $k$) cannot be determined. Applying Lemma~\ref{lemma:findpath2} for $Q$ and Lemma~\ref{lemma:thread_short_cycle} for $P$, we conclude that the probability of $\hE_1\cap \hE_3$ is $\tO(1/n)$.
\end{proof}

\begin{remark}\label{rem:211+}
As in Remark~\ref{rem:212+}, we have
$$
\Pr(\hE(2,1,1)\mid \hE(1,2,1), \hE_4(1,2,1))=\tO(n^{-3/2}).
$$
\end{remark}

We conclude,
\begin{lemma}\label{lemma:121-211}
We have
$$
\Pr(E(1,2,1), E(2,1,1))=\tO(n^{-2}).
$$
\end{lemma}
\begin{proof}
	The proof is identical to Lemma~\ref{lemma:121-212}, but we now use Lemma~\ref{lemma:211} and Remarks~\ref{rem:211+}.

	Vertices $u_1,v_1,u_2,v_2$ in the event $\hE(1,2,1)\cap \hE(2,1,1) \cap G_1\cap G_2$ 
	now play the role of vertices $\bb_{\min(p_1,q_1)}^{1}, \bb_{\max(p_1,q_1)}^{1}$ in the event $E(1,2,1)$, and $\bb_{p_2}^{2}, \bb_{q_2}^{1}$ in the event $E(2,1,1)$, respectively.

	Again, we adapt $G_1$ to connect $x_1$ to $u_1$ or $v_2$ depending on their relative positions,
	and we note that the target vertex $u_1\equiv \bb_{\min(p_1,q_1)}^{1}$ or $v_2\equiv \bb_{q_2}^{1}$ is always a lower record on its thread, since $\min(\min(p_1,q_1),q_2)\leq \ell_1$, so we can use Lemma~\ref{lemma:P1_enhanced} when connecting $x_1$ to it.
Similarly, we adapt $G_2$ to connect $x_2$ to $u_2\equiv\bb_{p_2}^{2}$, which is a lower record.

\end{proof}

\subsection{The rest of the cases}
\label{subsec:121-others}

	There are three cases remaining to treat, namely $(I_1,I_2)$ equal to one the three values
$$
((2,2,1),(1,1,2))\ , \ ((1,2,2),(2,1,1))\ , \ ((2,2,1),(2,1,1)).
$$
We will treat these three cases all together. In contrast to the previous cases we will consider the two events $\hE(I_1)\cap \hE(I_2)$ as a whole (we define these events below), and expose parts of them altogether. In all three cases, we will expose a $w_1$-thread and a $w_2$-thread with two lower-records in each thread that will both end with a short cycle. Then, we will have a $w_1$- or $w_2$-thread connecting one thread with the other. These connections can appear in different ways depending on the case and on the relative position of the lower-records in the threads. For $i\in \{1,2\}$, the vertices $u_i$ and $v_i$ will play the role of the two lower records of interest in the $w_i$-thread, in the order they appear. Namely, define $\nota{\hE}:=\hE_1\cap\hE_2\cap\hE_3\cap\hE_4$ where the events are:
\begin{itemize}
\item[] \nota{$\hE_1$}: $\cyc{w_1}{v_1,0}=1$;
\item[] \nota{$\hE_2$}: $(v_1,0)$ appears in $[\thread{u_1}{0}{w_1}]$ and is a generalized $\sigma_1$-lower-record at congruence $0$ in that thread.
\item[] \nota{$\hE_3$}: $\cyc{w_2}{v_2,0}=1$;
\item[] \nota{$\hE_4$}: $(v_2,0)$ appears in $[\thread{u_2}{0}{w_2}]$ and is a  generalized $\sigma_2$-lower-record at congruence $0$ in that thread.
\end{itemize}

\noindent The cases $(I_1,I_2)=((2,2,1),(1,1,2))$ and $(I_1,I_2)=((1,2,2),(2,1,1))$ give rise to the following path configurations connecting the two threads:
\begin{itemize}
\item[(a)] $\arrows{u_1}{r_1}{w_3}{v_2}{s_2}$ and $\arrows{u_2}{r_2}{w_4}{v_1}{s_1}$;
\item[(b)] $\arrows{v_2}{s_2}{w_3}{u_1}{r_1}$ and $\arrows{v_1}{s_1}{w_4}{u_2}{r_2}$;
\item[(c)] $\arrows{u_1}{r_1}{w_3}{u_2}{r_2}$ and $\arrows{v_2}{s_2}{w_4}{v_1}{s_1}$.
\end{itemize}
for any $w_3,w_4$ with $\{w_3,w_4\}=\{w_1,w_2\}$.
Note that we do not need to consider the case when $\arrows{u_2}{r_2}{w_3}{u_1}{r_1}$ and $\arrows{v_1}{s_1}{w_4}{v_2}{s_2}$, as it is covered by c) by exchanging the roles of $1$ and $2$.

\noindent The case $(I_1,I_2)=((2,2,1),(2,1,1))$ gives rise to two path configurations:
\begin{itemize}
\item[(d)] $\arrows{u_2}{r_2}{w_3}{v_1}{s_1}$ and $\arrows{v_2}{s_2}{w_4}{u_1}{r_1}$, %
\item[(e)] $\arrows{u_2}{r_2}{w_3}{u_1}{r_1}$ and $\arrows{v_2}{s_2}{w_4}{v_1}{s_1}$, %
\end{itemize}
where here again we do not need to consider more cases thanks the exchange of the roles of $1$ and $2$.
For each case in (a)-(e), we will denote \nota{$\hE_5$} the first event, and \nota{$\hE_6$} the second one.

The four vertices $u_1,u_2,v_1,v_2$ will play the role of the four vertices $\bb_{p_1}^{i_1}, \bb_{q_1}^{j_1},\bb_{p_2}^{i_2}, \bb_{q_2}^{j_2}$ of the event $E(I_1)\cap E(I_2)$, in some order depending on the case considered. 
In cases a), b) and c), both $u_i$ are $\sigma_i$-lower-records for $i\in [2]$. However, in cases d) and e), while $u_2$ is a $\sigma_2$-lower-record, $u_1$ might not be one. In such situation, $u_1=v_1= \bb_{\ell_1+1}^1$ is the only vertex at congruence $0$ in the cycle of length $k$, and cases d) and e) are the same.

\begin{figure}[h]
	\begin{center}
		\includegraphics[width=0.9\linewidth]{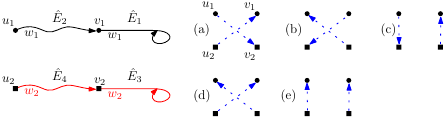}
			\caption{Pictorial view of the events of Section~\ref{subsec:121-others}. 
		Left: the events $\hE_1, \hE_2, \hE_3, \hE_4$. Right: the five cases for the types of collisions. In each case, the two (blue) paths use the two words $w_3,w_4$ with $\{w_3,w_4\}=\{w_1,w_2\}$.}
		\label{fig:E1E2-2}
	\end{center}
\end{figure}

\subsubsection{The auxiliary lemmas}

We prove two auxiliary lemmas that will help us controlling the cases (a)-(e).
\begin{lemma}\label{lemma:fir_cherry}
Let $w_1,w_2,w\in \mathcal{W}_k^{NC}$ with $w_1$ and $w_2$ non-conjugated. Let $y,z_1,z_2\in [n]$ be chosen uniformly at random and $i,j\in [[k]]$. Let $\ell=\ell(n)$ be such that $\ell=\tO(1)$. The probability that 
\begin{itemize}
\item[-] $\cyc{w_1}{z_1}=\cyc{w_2}{z_2}=1$;
\item[-] $\arrows{y}{0}{w_1}{z_1}{0}$ with a thread of length at most $\ell$;
\item[-] $\arrows{y}{i}{w}{z_2}{j}$ with a thread of length at most $\ell$;
\end{itemize}
is $\tO(n^{-3})$. Moreover, the probability of the previous event and that $[\thread{z_2}{0}{w_2}]$ has length $k$ is $\tO(n^{-7/2})$.
\end{lemma}
\begin{proof}
Fix $y\in [n]$ and 
	$z_2\in [n]$.
	By Lemma~\ref{lemma:findpath2} and a union bound over $t'\leq \ell$, the probability of the third event is $\tO(1/n)$. Let $t_1$ be the time after exploring $[\arrows{y}{i}{w}{z_2}{j}]$ with a thread of length at most $\ell$. We may assume we stop exploring at the first appearance of $z_2$ in the thread, so the out-degree of $z_2$ in $\cE_{t_1}$ is $0$. 
	It follows that $\cyc{w_i}{u,r}$ is not determined for any $(u,r,w_i)\in \cE_{t_1}$ and $i\in \{1,2\}$.
	Moreover observe that under the event of the lemma, $\cyc{w_1}{y}=1$. 
	So we may apply the first part of Lemma~\ref{lemma:thread_short_cycle_double} and obtain that the probability of $\cyc{w_1}{y}=\cyc{w_2}{z_2}=1$ is $\tO(1/n)$. We now choose $z_1$ at random, the probability that it appears in the first $\ell$ steps of $[\thread{y}{0}{w_1}]$ is $\tO(1/n)$, which concludes the proof. 

For the second statement, we use the second part of Lemma~\ref{lemma:thread_short_cycle_double} instead of the first one to obtain an additional factor $\tO(n^{-1/2})$.
\end{proof}

\begin{lemma}\label{lemma:sec_cherry}
Let $w_1,w_2,w\in \mathcal{W}_k^{NC}$ with $w_1$ and $w_2$ non-conjugated. Let $y,z_1,z_2\in [n]$ be chosen uniformly at random and $i,j\in [[k]]$. Let $\ell=\ell(n)$ be such that $\ell=\tO(1)$. The probability that 
\begin{itemize}
\item[-] $\cyc{w_1}{z_1}=\cyc{w_2}{z_2}=1$;
\item[-] $\arrows{y}{0}{w_1}{z_1}{0}$ with a thread of length at most $\ell$;
\item[-] $\arrows{z_2}{i}{w}{y}{j}$ with a thread of length at most $\ell$;
\end{itemize}
is $\tO(n^{-3})$. Moreover, the probability of the previous event and that $[\thread{z_1}{0}{w_1}]$ has length $k$ is $\tO(n^{-7/2})$.
\end{lemma}
\begin{proof}
	Fix $y,z_2\in [n]$. The proof is identical to the one of Lemma~\ref{lemma:fir_cherry} with the only difference that we start by exposing at most $\ell$ steps of $[\thread{z_2}{i}{w}]$ until finding $y$, instead of exposing at most $\ell$ steps of $[\thread{y}{i}{w}]$. For the second part we argue as in the second part of Lemma~\ref{lemma:fir_cherry}.
\end{proof}

\subsubsection{Proof of the remaining cases}

\begin{lemma}[Cases (a), (b) and (d)]\label{lemma:abd}
For random $u_1,u_2,v_1,v_2$, the probability that $\hE\cap \hE_5\cap\hE_6$ holds in any of the configurations (a),(b) or (d) is $\tO(n^{-4})$. In case (d), if in addition $[\thread{u_1}{0}{w_1}]$ has length $k$, the probability is  $\tO(n^{-9/2})$.
\end{lemma}
\begin{proof}
As usual, we will prove the lemma under $E_{typ}$. In particular, $E_{foll}$ and the fact that we will never explore more than $4$ different threads in the proof implies that there are at most $\ell:=16k^2h_{\max}(4)^2=\tO(1)$ consecutive following steps in the exploration process. In particular, any thread of length longer than $\ell$ cannot be determined by the previous exploration, unless it has been already explored.

We will do the detailed proof for (a), and explain how to modify it to obtain (b) and (d). 

	Fix $v_1,v_2\in [n]$ and define $Y_1$ as the set of vertices $y$ such that in $A$ the following holds: $\arrows{y}{0}{w_1}{v_1}{0}$ and  $\arrows{y}{r_1}{w_3}{v_2}{s_2}$ with threads of length at most $\ell$. In particular, if $u_1\in Y_1$ then $\hE_5$ is satisfied with short paths (of length at most $\ell$) departing from $u_1$.

Consider the event $F$ defined as $u_1\in Y_1$ and $\cyc{w_1}{v_1}=\cyc{w_2}{v_2}=1$.
We consider two cases:
\begin{itemize}
\item $F$ holds. By Lemma~\ref{lemma:fir_cherry} with $w=w_3$, $i=r_1$ and $j=s_2$, the probability of $F$ is $\tO(n^{-3})$. Note that $u_2$ has not been revealed so far. By Lemma~\ref{lemma:P1_enhanced}, the probability that $\arrows{u_2}{0}{w_2}{v_2}{0}$ and $(v_2,0)$ is a $\sigma_2$-lower-record at congruence $0$ of it is $\tO(1/n)$. However, it might be the case that $v_2$ is not a lower-record; then $[\thread{v_2}{0}{w_2}]$ has length $k$. Instead we use the second part of  Lemma~\ref{lemma:fir_cherry} to obtain a probability of $\tO(n^{-7/2})$, and then apply Lemma~\ref{lemma:P1} to get the additional $\tO(n^{-1/2})$ remaining. In either case, the final probability is $\tO(n^{-4})$.

\item $F$ does not hold. By Lemma~\ref{lemma:thread_short_cycle_double}, the probability of $\hE_1\cap \hE_3$ is $\tO(1/n)$. 
If $\hE\cap\hE_5\cap \hE_6$ holds, as $F$ does not hold, it implies that at least one of the two threads, $P=[\arrows{u_1}{0}{w_1}{v_1}{0}]$ or $Q=[\arrows{u_1}{r_1}{w_3}{v_2}{s_2}]$, has length larger than $\ell$, we call it \emph{long}. Recall that a long path cannot be determined by the previous exploration.
\begin{itemize}
\item[-] $P$ and $Q$ are long. Use Lemma~\ref{lemma:P1_enhanced} for $P$ to show that the probability of $\hE_2$ is $\tO(1/n)$. Since $Q$ is not determined, by Lemma~\ref{lemma:findpath2} it has probability $O(n^{-1/2})$ to appear.

\item[-] $P$ is long, $Q$ is short. If $Q$ is determined, choose $u_1$ at random, the probability of $\hE_5$ is at most the probability that $u_1\in B^-_{t_1}(v_2,\ell)$, which is $\tO(1/n)$ by $E_{ball}\subset E_{typ}$.  Otherwise, explore $Q$, by Lemma~\ref{lemma:findpath2} and a union bound over the length $t'\leq \ell$ of $Q$, the probability of $Q$ is also $\tO(1/n)$. Let $t_2$ be the time at the end of the exploration of $Q$. Since $P$ is long it cannot be determined by $\cE_{t_2}$. By Lemma~\ref{lemma:findpath2} the probability of $P$ is $\tO(n^{-1/2})$. 

\item[-] $P$ is short, $Q$ is long. The same argument as in the previous case  works, with the roles of $P$ and $Q$ reversed.

\end{itemize}
		In all the cases the probability contribution is $O(n^{-3/2})$. Since
all the events considered so far are independent of $u_2$,
	we can apply the same argument with the set $Y_2$, defined as the set of vertices $y$ for which  $\arrows{y}{0}{w_2}{v_2}{0}$ and  $\arrows{y}{r_2}{w_4}{v_1}{s_1}$ with threads of length at most $\ell$. The same argument can be applied with $u_2\in Y_2$, which gives an additional probability $O(n^{-3/2})$. Taking into account the contribution $\tO(n^{-1})$ of $\hE_1\cap \hE_3$, the total probability is $\tO(n^{-4})$. 

		To prove it for (b), define $Y_1$ to be the vertices $y$ such that  $\arrows{y}{0}{w_1}{v_1}{0}$ and $\arrows{v_2}{s_2}{w_3}{y}{r_1}$ with threads of length at most $\ell$, and $Y_2$ to be the vertices $y$ such that  $\arrows{y}{0}{w_2}{v_2}{0}$ and $\arrows{v_1}{s_1}{w_3}{y}{r_2}$ with threads of length at most $\ell$. With these new definitions, the proof of (a) can be followed verbatim, replacing the use of Lemma~\ref{lemma:fir_cherry} by Lemma~\ref{lemma:sec_cherry} instead.

		The proof for (d) is again similar. Define $Y_1$ as in case (b) and  $Y_2$ as in case (a), and use both Lemma~\ref{lemma:sec_cherry} and Lemma~\ref{lemma:fir_cherry} (respectively in the part of the proof that deals with the vertex $u_1$, and with the vertex $u_2$).
		In the particular case that $[\thread{u_1}{0}{w_1}]$ has length $k$, we need to use the second part of Lemma~\ref{lemma:sec_cherry}. The rest of the proof is identical.
\end{itemize}
\end{proof}

\begin{lemma}[Case (c) and (e)]\label{lemma:ce}
For random $u_1,u_2,v_1,v_2$, the probability that $\hE\cap \hE_5\cap\hE_6$ holds in any of the configurations (c) or (e) is $\tO(n^{-4})$. In case (e), if in addition $[\thread{u_1}{0}{w_1}]$ has length $k$, the probability is  $\tO(n^{-9/2})$.
\end{lemma}
\begin{proof}
We will prove the lemma under $E_{typ}$. Define $\ell$ as in the proof of Lemma~\ref{lemma:abd}. We will prove (c) and explain how to modify the proof to obtain (e). Before starting, we note that in both cases $(v_2,0)$ is a $\sigma_2$-lower-record in the thread $[\thread{u_2}{0}{w_2}]$.

	Fix $v_1,u_2\in [n]$ and define $Y_1$ as the set of vertices $y$ such that in $A$ the following holds: $\arrows{y}{0}{w_1}{v_1}{0}$ and  $\arrows{y}{r_1}{w_3}{u_2}{r_2}$ with threads of length at most $\ell$ (called as before \emph{short}). Let $F$ be the event that $u_1\in Y_1$ and $\cyc{w_1}{v_1}=\cyc{w_2}{u_2}=1$.

\begin{itemize}
\item $F$ holds. By Lemma~\ref{lemma:fir_cherry} with $w=w_3$, $i=r_1$ and $j=r_2$, the probability of $F$ is $\tO(n^{-3})$. Note that $v_2$ has not been revealed so far. Since $(v_2,0)$ is a $\sigma_2$-lower-record at congruence $0$ of $[\thread{u_2}{0}{w_2}]$ and there are $\tO(1)$ of them (see Remark~\ref{rem:lower_record}),
	the probability of selecting one of them is $\tO(1/n)$. The final probability is $\tO(n^{-4})$.

\item $F$ does not hold. 
	If $\hE\cap\hE_5\cap \hE_6$ holds, as $F$ does not hold, it implies that either $P=[\arrows{u_1}{0}{w_1}{v_1}{0}]$ or $Q=[\arrows{u_1}{r_1}{w_3}{u_2}{r_2}]$ have length larger than $\ell$, we call it \emph{long}.  Apply Lemma~\ref{lemma:findpath2} to show that the probability of $\arrows{v_2}{s_2}{w_4}{v_1}{s_1}$ is $\tO(n^{-1/2})$. We may assume there is no $v_1$ in the previous thread
		(up to changing the value of $s_1$). Let $t_1$ be the time after exposing it. For $i\in \{1,2\}$, as $v_1$ has only appeared once, $\cyc{w_i}{u,0}$ is not determined for any vertex $u$ with $(u,0,w_i)\in \cE_{t_1}$ and we can apply Lemma~\ref{lemma:thread_short_cycle_double} to obtain that the probability of $\hE_1\cap\hE_3$ is $\tO(1/n)$.
		We can apply Lemma~\ref{lemma:P1_enhanced} to connect $u_2$ with $v_2$, so the probability of $\hE_4$ given this is $\tO(1/n)$. We split the end of the proof depending on the length of $Q$.
\begin{itemize}
\item[-] $Q$ is long. By applying again Lemma~\ref{lemma:P1_enhanced} to connect $u_1$ with $v_1$, the probability of $\hE_2$ is $\tO(1/n)$.  Finally, as $Q$ is long, it cannot be determined by the previous exploration.  By Lemma\ref{lemma:findpath2} the probability it exists is $O(n^{-1/2})$. The total probability is $O(n^{-4})$.
\item[-] $Q$ is short, so $P$ is long. We use the same strategy as in the proof of the previous lemma. If $Q$ is determined, there are only $\tO(1)$ choices for $u_1$, and otherwise we can apply Lemma~\ref{lemma:findpath2} and a union bound over $t'\leq \ell$. In both cases the probability of $\hE_5$ is $\tO(1/n)$. Since $P$ is long, it cannot be determined by the exploration and, by Lemma~\ref{lemma:findpath2}, the probability it exists is $\tO(n^{-1/2})$. The total probability is $O(n^{-4})$. 
\end{itemize}

\end{itemize}

The case (e) is proved identically, by reversing the direction of $Q$, obtaining $Q=[\arrows{u_2}{r_2}{w_3}{u_1}{r_1}]$, which has no effect over the proof. In this case, the second statement follows directly from Lemma~\ref{lemma:abd}, as cases d) and e) are the same if $u_1=v_1$.
\end{proof}

We conclude,
\begin{lemma}\label{lemma:the_rest}
We have
\begin{align}
\Pr(E(2,2,1), E(1,1,2))=\tO(n^{-2}),\label{eq:221_112}\\
\Pr(E(1,2,2), E(2,1,1))=\tO(n^{-2}),\label{eq:122_211}\\
\Pr(E(2,2,1), E(2,1,1))=\tO(n^{-2}).\label{eq:221_211}
\end{align}
\end{lemma}
\begin{proof}
	The proof is similar to that of Lemma~\ref{lemma:121-212}, with only a bit more cases to handle. Define the events $G_1$ and $G_2$ as in that proof and recall that $u_2$ is always a lower-record of its corresponding thread. Let us assume that $u_1$ is also a lower-record.

Let us first prove~\eqref{eq:221_112}.
	To do this, consider the union $\hat{F}$ of the events
$\hE_1\cap\hE_2\cap\hE_3\cap\hE_4\cap\hE_5\cap\hE_6$
	over the following four cases: case (a) with $w_3=w_1$ (and $w_4=w_2$), case (b) with $w_3=w_2$ (and $w_4=w_1$), case (c) with $w_3=w_1$ (and $w_4=w_2$), and finally the same as case (c) with the role of $1$ and $2$ reversed.
This covers all the cases where 
the event $\hE_5$ connects a vertex in $\{u_2,v_2\}$ to a vertex in $\{u_1,v_1\}$ with a $w_2$-thread and 
	the event $\hE_6$ connects a vertex in $\{u_1,v_1\}$ to a vertex in $\{u_2,v_2\}$ with a $w_2$-thread  (representing respectively collisions of types $(2,2,1)$ and $(1,1,2)$).
	Therefore, the union of $\hat{F}\cap G_1 \cap G_2$ over all choices of $u_1,u_2,v_1,v_2$ is precisely $E(2,2,1) \cap E(1,1,2)$.

	Now, from Lemmas~\ref{lemma:abd}-\ref{lemma:ce}, the probability of $\hat{F}$ (with $u_1,v_1,u_2,v_2$ chosen at random) is $\tO(n^{-4})$.
	Moreover, for each of the cases defining $\hat{F}$, this event can be certified by four threads emanating from $u_1,v_1,u_2,v_2$, so we can proceed as in the proof of Lemma~\ref{lemma:121-212} by invoking  Lemma~\ref{lemma:P1_enhanced} twice, and we obtain that the probability of $G_1\cap G_2$ given $\hat{F}$ is $\tO(n^{-2})$, leading to a total contribution of $\tO(n^{-6})$. The result follows by union bound over $u_1,v_1,u_2,v_2$.

Proving~\eqref{eq:122_211} is identical, with the difference of the choice for $w_3$ and $w_4$ in each case. To prove~\eqref{eq:221_211}, we do the same steps, now using configurations (d) and (e).

Finally, if $u_1$ is not a lower record, $u_1=v_1=\bb_{\ell_1+1}^1$ and by Lemmas~\ref{lemma:abd} and~\ref{lemma:ce} we obtain an additional $\tO(n^{-1/2})$ for the probability of the corresponding event. We can mimic the previous proof, replacing the use of~Lemma~\ref{lemma:P1_enhanced} by Lemma~\ref{lemma:P1} when connecting $x_1$ to $u_1$.

\end{proof}

\subsection{Wrapping-up the proof of Lemma~\ref{lemma:AL5}.}

Putting together Lemmas~\ref{lemma:121-212},~\ref{lemma:121-112},~\ref{lemma:121-211} and~\ref{lemma:the_rest}, we obtain that for all $(I_1,I_2)\in\cI_1\times\cI_2$,
\begin{align*}
\mathbb{P}_{\AAbbll}(
		E(I_1)\cap E(I_2))
		 = \tO(n^{-2}),
\end{align*}
where we have used the previous observation that there are only sixe cases to consider instead of nine.
We thus have proved Lemma~\ref{lemma:AL5var}, which we already observed implies Lemma~\ref{lemma:AL5}.

\section{Proof of Lemma~\ref{lemma:AL6}}
\label{sec:telescopic}

\begin{figure}[h]
	\begin{center}
			\includegraphics[width=0.5\linewidth]{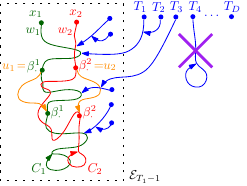}
			\caption{Pictorial view of the proof of Lemma~\ref{lemma:AL6} (in the case pictured $u_1,u_2$ appear on two different threads, this is not true for all choices of $I_1,I_2$). We start by exposing a certificate that the event of  Lemma~\ref{lemma:AL5} holds (in green/red/orange). Then, we successively pick unexplored vertices and explore their $w_1$-thread (in blue). We have to make sure that these threads attach to the previously explored part, rather than creating new cycles. To simplify the asymptotic analysis, we only start estimating the probabilities after a certain time $T_1$ after which the explored part has become sufficiently big. }
		\label{fig:telescopic}
	\end{center}
\end{figure}

In this section we prove Lemma~\ref{lemma:AL6}. Note that the event considered in  Lemma~\ref{lemma:AL6} is stronger than the one of  Lemma~\ref{lemma:AL5} as we want not only the $w_i$-cycle ending the threads of the vertices $v_i$ to have length $1$, but we want the $w_i$-threads of all other vertices of the automaton to eventually coincide with these cycles.
The proof consists in showing that this strengthening decreases the probability by an arbitrary polylog factor.

\begin{proof}[Proof of Lemma~\ref{lemma:AL6}]
By Proposition~\ref{prop:Etyp}, it suffices to bound all desired probabilities under $E_{typ}$.
	Since the event of  Lemma~\ref{lemma:AL5} is contained in the one we want to control, we will start by exposing a certificate that this event holds. However, to avoid subtle problems of bias when exposing threads from nonuniform random vertices, we will first re-randomize the problem so that vertices playing the roles of lower records are uniformly chosen. Namely, consider the event (of Lemma~\ref{lemma:AL5})
	\begin{align}
		E_{collision}=\left\{(A,v_1,v_2,\sigma_1,\sigma_2) \not \in (\FFF(I_1)\cup \FFF(I_2)) \mbox{ and } \cyc{w_1}{v_1}=\cyc{w_2}{v_2}=1.\right\}
	\end{align}
	We say that a pair of vertices $(u_1,u_2)$ \nota{certifies} the event $E_{collision}$ if for $m\in \{1,2\}$, $u_m$ is a $\sigma_{i_m}$-lower record at congruence $0$ on $[\thread{v_{i_m}}{0}{w_{i_m}}]$ and if it realizes the event~\eqref{eq:excludeT} with $u_m = \beta^{i}_p$ and $(i,h,j)=(i_m,h_m,j_m)\in I_m$ for some valid congruences $r=r_m,s=s_m$ (note that $h_m=3-m$, but we will not use this). By definition, $E_{collision}$ holds iff there is at least one pair that certifies it. Conversely, since the probability that a random permutation of size at most $n$ has more than $\log^2 n$ lower records is more than polynomially small, we can assume that if $E_{collision}$ holds, it is certified by at most $\log^4n = \tilde{O}(1)$ pairs.

Therefore, Lemma~\ref{lemma:AL5} can be reformulated by saying that there exists $K>0$ such that 
	\begin{align}\label{eq:first_prob}
		\mathbb{P}_{\AAbbll\times[n]^2}((u_1,u_2) \mbox{ certifies }E_{collision}\mbox{ in }
		(A,v_1,v_2,\sigma_1,\sigma_2)
		\mbox{ and } \cyc{w_1}{v_1}=\cyc{w_2}{v_2}=1)
		 = O\left(\frac{\log^K(n)}{n^4}\right),
	\end{align}
	where $u_1$ and $u_2$ are now uniform random vertices independent from the rest.

	Given that this event holds, we reveal the four threads $[\thread{v_1}{0}{w_1}]$, $[\thread{v_2}{0}{w_2}]$, $[\thread{u_1}{r_1}{w_{h_1}}]$, $[\thread{u_2}{r_2}{w_{h_2}}]$ (i.e., we reveal a certificate of $E_{collision}$). 
Here we assume that $r_1,r_2$ are two fixed (non-random) values in $[[k]]$ (we will use union bound on these values at the end of the proof).

Let $T$ be the time at the end of the exploration, which satisfies $T\leq 4 t_{\max}$. Let $S$ be the number of steps $t\in [T]$ that satisfy $w_t=w_2$, note that $S\leq 2t_{\max}$.
Denote by $C_1$ and $C_2$ the $w_1$- and $w_2$-cycles at the end of $[\thread{v_1}{0}{w_1}]$ and $[\thread{v_2}{0}{w_2}]$, respectively.

	If $(A,v_1,v_2)\in \widehat{\AAbb}$, we need that for \emph{all} $v\in [n]$, the $w_1$-thread of $(v,0)$ yields to $C_1$ and the $w_2$-thread of $(v,0)$ yields to $C_2$ (otherwise for some $i$ there would be more than one cycle in $A_{w_i}$, which would not be a tree). Computing the probability of such event is complicated, but we can study another (much less constrained) event whose probability will be small enough to conclude. 
Namely, we will only explore the $w_1$-threads of ``a few'' vertices yet unexplored,
and %
	insist that these $w_1$-threads connect to the already revealed automaton rather than creating new $w_1$-cycles. See Figure~\ref{fig:telescopic}.

Continue with the exploration process from time $T$ by exploring $w_1$-threads of $(v,0)$, where $v$ are vertices chosen arbitrarily among the ones for which $[\thread{v}{0}{w_1}]$ is not yet determined. Let $T_1$ be the smallest time $t$ at the end of a $w_1$-thread that satisfies $t\geq 3t_{\max}$. By $E_{typ}$, we have $T_1\leq 4 t_{\max}$.

	\medskip
	Fix $t'= (\log n)^{K+3}t_{\max}$. We will bound the probability that all the $w_1$-threads constructed between time $T_1$ and time $t'$ eventually synchronize with the revealed $w_1$-automaton. 
	Let $D$ be the number of $w_1$-threads explored between time $T_1$ and $t'$. For $i\geq 1$, let $T_{i+1}$ be the time at the start of the $i$-th thread explored after time $T_1$, and we let $T_{D+1}$ be the time at the end of the $D$-th thread. By the choice of $D$, $T_{D+1}\geq t'$.

	Suppose that at time $t$ we are exploring the $i$-th thread and let $r$ be the congruence of time $t$. 
		We will examine in which ways the process can hit a vertex  that terminates the exploration of the current thread (that is to say, such that the current time is the last exploring time on this thread).

First, there are at least 
\begin{align}\label{eq:num_closing}
\frac{t+1-S}{k}-(i+4)\geq  \frac{t+1-S}{k} (1-\tO(D/t)),
\end{align}
	vertices whose hit would terminate the exploration of the current thread.  Indeed, from the $t+1$ triplets in $\cE_t= ((u_s,r_s,w_s))_{s\leq t}$, all but $S$ have $w_s=w_1$, and, by Claim \ref{claim:equi}, among these ones a proportion $1/k$ (up to some small additive error related to the number of threads exposed) satisfy $r_s=r$. For each of these choices, the exploration encounters a triple vertex/congruence/word already visited, and thus stops.

		We now upper bound, for $t\in [T_i,T_{i+1})$,
	the number of ways to terminate the exploration of the current thread without creating any new cycle in $A_{w_1}$, that is, in such a way that the cycle ending the thread is the one already contained in $\cE_{T_i-1}$. There are two ways this could happen.
	First, we may hit a vertex $v$ such that $(v,r,w_1)\in \cE_{T_i-1}$; similarly as in \eqref{eq:num_closing}, there are at most $(T_i-S)/k+(i+4)$ such vertices. Second, we may hit a vertex $v$ with $(v,r,w_1)\notin \cE_{t}$ but whose future is determined (note that if we hit a triple in $\cE_t\setminus \cE_{T_i-1}$, we close the exploration by creating a \emph{new} cycle of $A_{w_1}$, so there are no other cases). If this is the case, then
	$B^+_t(v,k)$ is not a directed path, otherwise, as $[\thread{v}{0}{w_1}]$ has length at least $k$, we would obtain a contradiction by Claim~\ref{claim:trajectory}. By $E_{typ}$, there are at most $10k h_{max}(i+4)^2$ such vertices. Therefore the total number of ways to  
	terminate the exploration of the current thread without creating any new cycle in $A_{w_1}$
	is at most
$$
\frac{T_i-S}{k}+(i+4)+10k h_{max}(i+4)^2=  \frac{T_i-S}{k} (1-\tO(D^4/t)).
$$
	We conclude that, conditional on closing the $i$-th thread at time $t$ (i.e. $t=T_{i+1}-1$), the probability to terminate the exploration of the current thread without creating any new cycle in $A_{w_1}$ is at most
\begin{align*}
\frac{T_i-S}{T_{i+1}-S}(1+\tO(i^4/T_1)) = \frac{T_i-S}{T_{i+1}-S}(1+\tO(D^4 n^{-1/2})). 
\end{align*}
By a telescopic argument, the probability that for every $i$, the 
	exploration of the $i$-th thread stops without creating a  new cycle in $A_{w_1}$
	(which is required for it being a $w_1$-tree) is at most
\begin{align}\label{eq:tel_prob}
	\frac{T_1-S}{T_{D+1}-S}\left(1+\tO(D^4 n^{-1/2})\right)^D\leq \frac{4}{\log^{K+3}(n)}\left(1+\tO(D^5 n^{-1/2})\right)
\end{align}
where we used $\frac{a-c}{b-c}\leq \frac{a}{b}$ for $b\geq a\geq c>0$.

To conclude now, we only need to control $D$. Recall that $D-1$ is the number threads that have been closed up in $t'-T_1$ steps. At each time $t$, there are at most $t\leq t'$ vertices whose hit would close the thread (as only hitting vertices that appear in $\cE_t$ could close the thread). Therefore, $D-1$ is stochastically dominated by a binomial random variable with parameters $t'$ and $t'/n$. If $m:= k^2\log^{2K+6}(n)$,
$$
\Pr(D\geq m+1)=\Pr(D-1\geq m) \leq \binom{t'}{m} \left(\frac{t'}{n}\right)^m \leq \left(\frac{e (t')^2}{m n}\right)^m = \left(\frac{O(1)}{\log(n)}\right)^m =o(n^{-3}),
$$
where we used $\binom{t'}{m}\leq \big(\frac{et'}{m}\big)^m$ and the value of $t_{\max}$.

Therefore, we may assume that $D\leq m+1$, in particular $\tO(D^5 n^{-1/2})=\tO(n^{-1/2})$ and the error in \eqref{eq:tel_prob} is negligible. Combining~\eqref{eq:first_prob} and~\eqref{eq:tel_prob}, and using union bound on $(u_1,u_2)\in [n]^2$ and $(r_1,r_2)\in [[k]]^2$, we conclude the proof.

\end{proof}

\section{Lower bound}
\label{sec:lower}

In this section we prove Theorem~\ref{thm:lower}.

	\begin{proof}[Proof of Theorem~\ref{thm:lower}]
		Fix two integers $\ell, R >0$ (to be chosen later) such that $R \ell <n$.
	Let $w=w_1\dots w_{\ell}$  be a word of length $\ell$.
	We (partly) reveal a uniformly random automaton $A$, as follows:
\begin{itemize}[leftmargin=19pt, topsep=0pt, parsep=0pt, itemsep=1pt]
	\item		Pick a random vertex $V \in [n]$ uniformly at random, and reveal all the transitions in $A$ needed to compute a single $w$-transition from $V$, i.e. reveal the path $P: V \stackrel{w_1}{\longrightarrow}\dots \stackrel{w_{\ell}}{\longrightarrow}$. 
	\item Iterate $R$ times.
\end{itemize}
		If the word $w$ is synchronizing in $A$, it must be the case that at each round of the exploration (except from the first round), the path generated hits one of the vertices previously visited (indeed, all these paths end in the same point). Since there are $R-1$ non-initial rounds, and since the total number of visited vertices never exceeds $R\ell$, the probability of this event is at most
		$$
		\left(\ell \frac{R\ell}{n} \right)^{R-1},
		$$
		where the first factor of $\ell$ accounts for a union bound on the time of the first hit in each round, and where $\frac{R\ell}{n}$ upper bounds the probability that a uniform random vertex in $[n]$ is among the previously visited ones.

		Now, using union bound on the choice of the word $w$, the probability that \emph{some} word $w$ of length $\ell$ has this property is at most
	$$
		2^{\ell} \left(\ell \frac{R\ell}{n} \right)^{R-1}.
	$$
		Choosing $\ell=\lfloor n^{\alpha}\rfloor$ and $R=\lfloor n^\beta\rfloor$ with $\alpha+\beta<1$, this quantity is at most  
$2^{n^\alpha +  (2\alpha+\beta -1) (n^\beta-1) \log(n)} $. It goes to zero provided $\beta>\alpha$ and $2\alpha+\beta-1<0$, so we can choose $\beta=\frac{1}{3}+\epsilon$ and $\alpha=\frac{1}{3}-\epsilon$ for any $\epsilon>0$.
\end{proof}

\begin{remark}
	In the proof of the lower bound, we only take into account the fact that the path generated at each round has to hit one of the previously explored paths -- and we bound this probability by $R\ell^2/n$. To bound the probability that the word $w$ is synchronizing, one may want to replace this event by the stronger property that the path "hits a vertex which eventually synchronizes with the path generated in the first round". 
	Informally speaking, it is reasonable to expect that this synchronization happens only (or: almost only) when the path hits a vertex which was previously visited at the same congruence, and that a proportion about $\frac{1}{\ell}$ of hits have this property. Therefore, it seems reasonable to believe that the bound $R\ell^2/n$ in the proof could be, for this stronger event, replaced by $R\ell/n$, at least for typical words. This would change the final condition $2\alpha+\beta-1<0$ at the end of the proof to $\alpha+\beta-1<0$, and would eventually lead to a lower bound of $n^{\frac{1}{2}-\epsilon}$ for the length of synchronizing words. However, we do not know how to transform this intuition into a proof at the moment.
\end{remark}

\subsection*{Acknowledgements}
This project started after a beautiful talk given by Cyril Nicaud about his paper~\cite{Nicaud} at the \emph{Journées Combinatoires  de Bordeaux} in February 2020. The first author thanks Cyril for interesting discussions, and the organizers for this wonderful event.

 G.C. acknowledges funding from the European Research Council (ERC) under the European Union’s Horizon 2020 research and innovation programme (grant agreement No. ERC-2016-STG 716083 “CombiTop”) and from the grant  ANR-19-CE48-0011 ``COMBINÉ''. G.P was supported by the Spanish Agencia Estatal de Investigación under projects PID2020-113082GB-I00 and the Severo Ochoa and María de Maeztu Program for Centers and Units of Excellence in R\&{}D (CEX2020-001084-M). G.P. acknowledges an invitation in Paris funded by the ERC grant CombiTop, during which this project was started.

\bibliographystyle{alpha}
\bibliography{biblio}

\newcommand{\etalchar}[1]{$^{#1}$}
\begin{thebibliography}{ADG{\etalchar{+}}21}

\bibitem[ABBP20]{ABBP:rout-digraphs}
Louigi Addario-Berry, Borja Balle, and Guillem Perarnau.
\newblock Diameter and stationary distribution of random {$r$}-out digraphs.
\newblock {\em Electron. J. Combin.}, 27(3):Paper No. 3.28, 41, 2020.

\bibitem[ABD22]{ABD:randomTrees}
Louigi Addario-Berry and Serte Donderwinkel.
\newblock Random trees have height {$O(\sqrt{n})$}, 2022.

\bibitem[ACS17]{CameronSteinberg}
Jo\~{a}o Ara\'{u}jo, Peter~J. Cameron, and Benjamin Steinberg.
\newblock Between primitive and 2-transitive: synchronization and its friends.
\newblock {\em EMS Surv. Math. Sci.}, 4(2):101--184, 2017.

\bibitem[ADG{\etalchar{+}}21]{refCitingNicaud1}
Christoph Aistleitner, Daniele D'Angeli, Abraham Gutierrez, Emanuele Rodaro,
  and Amnon Rosenmann.
\newblock Circular automata synchronize with high probability.
\newblock {\em J. Combin. Theory Ser. A}, 178:Paper No. 105356, 30, 2021.

\bibitem[Ald93]{Aldous:CRT3}
David Aldous.
\newblock The continuum random tree. {III}.
\newblock {\em Ann. Probab.}, 21(1):248--289, 1993.

\bibitem[AZ99]{proofFromTheBook}
Martin Aigner and G{\"u}nter~M Ziegler.
\newblock Proofs from the book.
\newblock {\em Berlin. Germany}, 1, 1999.

\bibitem[Ber16]{Berlinkov:proof}
Mikhail~V. Berlinkov.
\newblock On the probability of being synchronizable.
\newblock In {\em Algorithms and discrete applied mathematics}, volume 9602 of
  {\em Lecture Notes in Comput. Sci.}, pages 73--84. Springer, [Cham], 2016.

\bibitem[BLL98]{Quebecois}
F.~Bergeron, G.~Labelle, and P.~Leroux.
\newblock {\em Combinatorial species and tree-like structures}, volume~67 of
  {\em Encyclopedia of Mathematics and its Applications}.
\newblock Cambridge University Press, Cambridge, 1998.
\newblock Translated from the 1994 French original by Margaret Readdy, With a
  foreword by Gian-Carlo Rota.

\bibitem[BN18]{BerlinkovNicaud}
Mikhail~V. Berlinkov and Cyril Nicaud.
\newblock Synchronizing random almost-group automata.
\newblock In {\em Implementation and application of automata}, volume 10977 of
  {\em Lecture Notes in Comput. Sci.}, pages 84--96. Springer, Cham, 2018.

\bibitem[Cam13]{Cameron}
Peter~J. Cameron.
\newblock Dixon's theorem and random synchronization.
\newblock {\em Discrete Math.}, 313(11):1233--1236, 2013.

\bibitem[CD17]{CaiDevroye}
Xing~Shi Cai and Luc Devroye.
\newblock The graph structure of a deterministic automaton chosen at random.
\newblock {\em Random Structures Algorithms}, 51(3):428--458, 2017.

\bibitem[CJ19]{refCitingNicaud2}
Costanza Catalano and Rapha\"{e}l~M. Jungers.
\newblock On random primitive sets, directable {NFA}s and the generation of
  slowly synchronizing {DFA}s.
\newblock {\em J. Autom. Lang. Comb.}, 24(2-4):185--217, 2019.

\bibitem[DMS19]{DoyenMassartMahsa}
Laurent Doyen, Thierry Massart, and Mahsa Shirmohammadi.
\newblock The complexity of synchronizing {M}arkov decision processes.
\newblock {\em J. Comput. System Sci.}, 100:96--129, 2019.

\bibitem[FJR{\etalchar{+}}98]{FJRST:action}
Joel Friedman, Antoine Joux, Yuval Roichman, Jacques Stern, and Jean-Pierre
  Tillich.
\newblock The action of a few permutations on r-tuples is quickly transitive.
\newblock {\em Random Structures \& Algorithms}, 12(4):335--350, 1998.

\bibitem[FO89]{Flajolet1989RandomMS}
Philippe Flajolet and Andrew~M. Odlyzko.
\newblock Random mapping statistics.
\newblock In {\em EUROCRYPT}, 1989.

\bibitem[Fra82]{Frankl:Cerny}
P.~Frankl.
\newblock An extremal problem for two families of sets.
\newblock {\em European J. Combin.}, 3(2):125--127, 1982.

\bibitem[FS09]{Flajolet:book}
Philippe Flajolet and Robert Sedgewick.
\newblock {\em Analytic combinatorics}.
\newblock cambridge University press, 2009.

\bibitem[GS15]{GS:inapprox}
Pawe\l\ Gawrychowski and Damian Straszak.
\newblock Strong inapproximability of the shortest reset word.
\newblock In {\em Mathematical foundations of computer science 2015. {P}art
  {I}}, volume 9234 of {\em Lecture Notes in Comput. Sci.}, pages 243--255.
  Springer, Heidelberg, 2015.

\bibitem[GV22]{GV:arxivPermutations}
Balázs Gerencsér and Zsombor Várkonyi.
\newblock Fast synchronization of inhomogenous random automata, 2022.

\bibitem[KKS13]{KKS-sqrtn}
Andrzej Kisielewicz, Jakub Kowalski, and Marek Szyku{\l a}.
\newblock A fast algorithm finding the shortest reset words.
\newblock In {\em Computing and combinatorics}, volume 7936 of {\em Lecture
  Notes in Comput. Sci.}, pages 182--196. Springer, Heidelberg, 2013.

\bibitem[Nic14]{Nicaud:survey}
Cyril Nicaud.
\newblock Random deterministic automata.
\newblock In {\em Mathematical foundations of computer science 2014. {P}art
  {I}}, volume 8634 of {\em Lecture Notes in Comput. Sci.}, pages 5--23.
  Springer, Heidelberg, 2014.

\bibitem[Nic19]{Nicaud}
Cyril Nicaud.
\newblock The \v{C}ern\'{y} conjecture holds with high probability.
\newblock {\em J. Autom. Lang. Comb.}, 24(2-4):343--365, 2019.
\newblock See also the conference paper {{\it Fast Synchronization of Random
  Automata}, APPROX-RANDOM 2016: 43:1-43:12}.

\bibitem[OU10]{OU:complexity}
J\"{o}rg Olschewski and Michael Ummels.
\newblock The complexity of finding reset words in finite automata.
\newblock In {\em Mathematical foundations of computer science 2010}, volume
  6281 of {\em Lecture Notes in Comput. Sci.}, pages 568--579. Springer,
  Berlin, 2010.

\bibitem[Pin83]{Pin:Cerny}
J.-E. Pin.
\newblock On two combinatorial problems arising from automata theory.
\newblock In {\em Combinatorial mathematics ({M}arseille-{L}uminy, 1981)},
  volume~75 of {\em North-Holland Math. Stud.}, pages 535--548. North-Holland,
  Amsterdam, 1983.

\bibitem[Pin17]{Pin:open}
Jean-Eric Pin.
\newblock {Open problems about regular languages, 35 years later}.
\newblock In Stavros Konstantinidis, Nelma Moreira, Rog{\'e}rio Reis, and
  Jeffrey Shallit, editors, {\em {The Role of Theory in Computer Science -
  Essays Dedicated to Janusz Brzozowski}}. {World Scientific}, 2017.

\bibitem[Pin21]{Pin:handbook}
Jean-\'{E}ric Pin.
\newblock Finite automata.
\newblock In {\em Handbook of automata theory. {V}ol. {I}. {T}heoretical
  foundations}, pages 3--38. EMS Press, Berlin, [2021] \copyright 2021.

\bibitem[QS22]{QuattropaniSau}
Matteo Quattropani and Federico Sau.
\newblock On the meeting of random walks on random dfa, 2022.

\bibitem[RS67]{RS:height}
Alfr{\'e}d R{\'e}nyi and George Szekeres.
\newblock On the height of trees.
\newblock {\em Journal of the Australian Mathematical Society}, 7(4):497--507,
  1967.

\bibitem[Shi19]{Shitov}
Yaroslav Shitov.
\newblock An improvement to a recent upper bound for synchronizing words of
  finite automata.
\newblock {\em J. Autom. Lang. Comb.}, 24(2-4):367--373, 2019.

\bibitem[ST11]{ST11:experimental}
Evgeny Skvortsov and Evgeny Tipikin.
\newblock Experimental study of the shortest reset word of random automata.
\newblock In {\em Implementation and application of automata}, volume 6807 of
  {\em Lecture Notes in Comput. Sci.}, pages 290--298. Springer, Heidelberg,
  2011.

\bibitem[SZ10]{SZ:sync}
Evgeny Skvortsov and Yulia Zaks.
\newblock Synchronizing random automata.
\newblock {\em Discrete Math. Theor. Comput. Sci.}, 12(4):95--108, 2010.

\bibitem[SZ22]{SzykulaZyzik}
Marek Szykuła and Adam Zyzik.
\newblock An improved algorithm for finding the shortest synchronizing words,
  2022.

\bibitem[Szy18]{Szykula}
Marek Szyku{\l a}.
\newblock Improving the upper bound and the length of the shortest reset words.
\newblock In {\em 35th {S}ymposium on {T}heoretical {A}spects of {C}omputer
  {S}cience}, volume~96 of {\em LIPIcs. Leibniz Int. Proc. Inform.}, pages Art.
  No. 56, 13. Schloss Dagstuhl. Leibniz-Zent. Inform., Wadern, 2018.

\bibitem[\v{C}64]{Cerny:original}
J\'{a}n \v{C}ern\'{y}.
\newblock A remark on homogeneous experiments with finite automata.
\newblock {\em Mat.-Fyz. \v{C}asopis. Sloven. Akad. Vied.}, 14:208--216, 1964.

\bibitem[Vol08]{Volkov:survey}
Mikhail~V. Volkov.
\newblock Synchronizing automata and the \v{C}ern\'{y} conjecture.
\newblock In {\em Language and automata theory and applications}, volume 5196
  of {\em Lecture Notes in Comput. Sci.}, pages 11--27. Springer, Berlin, 2008.

\bibitem[ZS12]{ZS:sync}
Yu.~I. Zaks and E.~S. Skvortsov.
\newblock Synchronizing random automata on a 4-letter alphabet.
\newblock {\em Zap. Nauchn. Sem. S.-Peterburg. Otdel. Mat. Inst. Steklov.
  (POMI)}, 402(Kombinatorika i Teoriya Grafov. IV):83--90, 219, 2012.

\end{thebibliography}

\end{document}